\documentclass[12pt]{article}

\usepackage{latexsym}
\usepackage{multirow}
\usepackage[linesnumbered,ruled,algonl,vlined]{algorithm2e}

\usepackage{amsfonts}

\usepackage{times}
\usepackage{amsmath}
\usepackage{graphicx}
\def\squareforqed{\hbox{${\Box}$}}
\def\qed{\ifmmode\squareforqed\else{\unskip\nobreak\hfil
\penalty50\hskip1em\null\nobreak\hfil\squareforqed
\parfillskip=0pt\finalhyphendemerits=0\endgraf}\fi
\vskip 6pt}

\newcounter{linenumber}
        {\end{list}}

\newcounter{tmp}

\newcommand\blankpagestyle{empty}

\def\cleardoublepage{\clearpage
  \if@twoside \ifodd\c@page
  \else
    \hbox{}\thispagestyle{\blankpagestyle}\newpage
    \if@twocolumn\hbox{}\newpage\fi
  \fi\fi
}

\newtheorem{definition}{Definition}[section]
\newtheorem{lemma}{Lemma}[section]
\newtheorem{theorem}[lemma]{Theorem}

\newcommand{\color}[1]{\lceil (#1+2)^2/2\rceil}

\newcommand{\eat}[1]{}

\date{}

\title{Privacy Preserving Moving KNN Queries}

\author{%
{Tanzima Hashem{\small $^{1}$} \hspace{0.3cm} Lars Kulik{\small $^{1}$} \hspace{0.3cm} Rui Zhang{\small $^{2}$}}%
\vspace{1.6mm}\\
\fontsize{10}{10}
$~^{1}$National ICT Australia\\
$~^{1,2}$Department of Computer Science and Software Engineering\\
University of Melbourne, Australia\\
\fontsize{9}{9}\selectfont\ttfamily\upshape
\{thashem,lars,rui\}@csse.unimelb.edu.au\\}

\begin{document}
\maketitle

\begin{abstract}
We present a novel approach that protects trajectory privacy of
users who access location-based services through a moving $k$
nearest neighbor (M$k$NN) query. An M$k$NN query continuously
returns the $k$ nearest data objects for a moving user (query
point). Simply updating a user's imprecise location such as a
region instead of the exact position to a location-based service
provider (LSP) cannot ensure privacy of the user for an M$k$NN
query: continuous disclosure of regions enables the LSP to follow
a user's trajectory. We identify the problem of trajectory privacy
that arises from the overlap of consecutive regions while
requesting an M$k$NN query and provide the first solution to this
problem. Our approach allows a user to specify the confidence
level that represents a bound of how much more the user may need to
travel than the actual $k^{th}$ nearest data object. By hiding a
user's required confidence level and the required number of
nearest data objects from an LSP, we develop a technique to
prevent the LSP from tracking the user's trajectory for M$k$NN
queries. We propose an efficient algorithm for the LSP to find $k$
nearest data objects for a region with a user's specified
confidence level, which is an essential component to evaluate an
M$k$NN query in a privacy preserving manner; this algorithm is at
least two times faster than the state-of-the-art algorithm.
Extensive experimental studies validate the effectiveness of our
trajectory privacy protection technique and the efficiency of our
algorithm.

\end{abstract}

\section{Introduction}
\label{sec:intro}

Location-based services (LBSs) are developing at an unprecedented
pace: having started as web-based queries that did not take a user's
actual location into account (e.g., Google maps), LBSs can
nowadays be accessed anywhere via a mobile device using the
device's location (e.g., displaying nearby restaurants on a cell
phone relative to its current location). While LBSs provide many
conveniences, they also threaten our privacy. Since
a location-based service provider (LSP) knows the locations of its
users, a user's continuous access of LBSs enables the LSP to
produce a complete profile of the user's trajectory with a high
degree of spatial and temporal precision. From this profile, the
LSP may infer private information about users. A threat to privacy
is becoming more urgent as positioning devices become more
precise, and a lack of addressing privacy issues may significantly
impair the proliferation of LBSs~\cite{web1,ITRoadGF.03}.

An important class of LBSs are moving $k$ nearest neighbor
(M$k$NN) queries. An M$k$NN query continuously returns the $k$
nearest data objects with regard to a moving query point. For
example, a driver may continuously ask for the closest gas station
during a trip and select the most preferred one;
similarly, a tourist may continuously query the five nearest
restaurants while exploring a city. However, accessing an M$k$NN
query requires continuous updates of user locations to the LSP,
which puts the user's privacy at risk. The user's trajectory (i.e., the sequence of updated locations) is
sensitive data and reveals private information. For example if
the user's trajectory intersects the region of a liver clinic, then the LSP might infer
that the user is suffering from a liver disease.

A popular approach to hide a user's location from the LSP is to
let the user send an imprecise location (typically a rectangular
region containing the user's location) instead of the exact
location~\cite{cheng06.PET,Damiani09.SPRINGL,ghinita09.GIS,Xue09.LOCA}.
This approach is effective when the user's location is fixed.
However, when the user moves and continuously sends the
rectangular regions containing her locations to the LSP, the LSP
can still approximate the user's trajectory if it takes into
account the overlap of consecutive rectangles, which poses a
threat to the \emph{trajectory privacy} of the user. This
privacy threat on the user's trajectory privacy is called the
\emph{overlapping rectangle attack}. Our aim is to protect a
user's trajectory privacy while providing M$k$NN answers. We call
the problem of answering M$k$NN queries with privacy protection,
the \emph{private moving $k$NN (PM$k$NN) query}. Although
different
approaches~\cite{cheng06.PET,chow07.SSTD,ghinita09.GIS,Gkoulalas-Divanis09.SDM,Xu09.TPDS,Xu07.ACMGIS}
have been developed for protecting a user's trajectory privacy in
continuous queries, none of them have considered the threat on a
user's trajectory privacy that arises from the overlapping
rectangle attack in M$k$NN queries. This paper is the first work
that addresses PM$k$NN queries.

In our approach, users have an option to specify the level of
accuracy for the query answers, which is motivated by the
following observation. In many cases, users would accept answers
with a slightly lower accuracy if they gain higher privacy
protection in return. For example, a driver looking for the
closest gas station might not mind driving to a gas station that
may be 5\% further than the actual closest one, if the slightly
longer trip considerably enhances the driver's privacy. In this
context, ``lower accuracy" of the answers means that the returned
data objects are not necessarily the $k$ nearest data objects:
they might be a subset of the $(k+x)$ nearest data objects, where
$x$ is a small integer. However, we guarantee that their distances
to the query point are within a certain \emph{ratio} of the actual
$k^{th}$ nearest neighbor's distance. We define a parameter called
\emph{confidence level} to characterize this ratio. In addition to
protecting privacy, we will show that a lower confidence level
also reduces the query processing overhead.

For every update of a user's imprecise location (a rectangle) in a PM$k$NN query, the LSP provides the user with a
candidate answer set that includes the specified number of nearest
data objects (i.e., $k$ nearest data objects) with the specified
confidence level for every possible point in the rectangle. The
key idea of our privacy protection strategy is to specify higher
values for the confidence level and the number of nearest data
objects than required by the user and not to reveal the required
confidence level and the required number of nearest data objects
to the LSP. Since the user's required confidence level and the
required number of nearest data objects are lower than the
specified ones, the candidate answer set must contain the required
query answers for an additional part of the user's trajectory,
which is unknown to the LSP. Based on this idea, we develop an
algorithm to compute the user's consecutive rectangles, that
resists the overlapping rectangle attack and prevents the
disclosure of the user's trajectory. Although our approach for
privacy works if either the required confidence level or the
required number of nearest data objects is hidden, hiding both
provides a user with a higher level of privacy.

In summary, we make the following contributions in this paper.
\begin{itemize}
    \item We identify the problem of trajectory privacy
that arises from the overlap of consecutive regions while
requesting an M$k$NN query. We propose the first approach to address PM$k$NN queries.
Specifically, a user (a client) sends requests for an M$k$NN query
based on consecutive rectangles, and the LSP (the server) returns
$k$ nearest neighbors (NNs) for any possible point in the rectangle. We show how to
compute the consecutive rectangles and how to find the $k$ NNs for
these rectangles so that the user's trajectory remains private.
\item We propose three ways to combat
the privacy threat in M$k$NN queries: by requesting (i) a higher
confidence level than required, (ii)  a higher number of NNs than
required, or (iii) higher values for both the confidence level and
the number of NNs than required to the LSP.
    \item We improve the efficiency of the algorithm for the LSP to find $k$ NNs for a rectangle with a user-customizable confidence level by exploiting different geometric properties.
    \item We present an extensive experimental study to
    demonstrate the efficiency and effectiveness of our
    approach. Our proposed algorithm for the LSP is at least two
times faster than the state-of-the-art.
\end{itemize}


The remainder of the paper is organized as follows.
Section~\ref{sec:pblm_stup} discusses the problem setup and
Section~\ref{sec:relwork} reviews existing work. In
Section~\ref{sec:sys_overview}, we give a overview of our system
and in Section~\ref{sec:conflevel}, we introduce the concept of
confidence level. Sections~\ref{sec:trj_priv} and~\ref{sec:knnq}
present our algorithms to request and evaluate a PM$k$NN query,
respectively. Section~\ref{sec:exp} reports our experimental
results and Section~\ref{sec:conc} concludes the paper with future
research directions.

\section{Problem Formulation}\label{sec:pblm_stup}

A moving $k$NN (M$k$NN) query is defined as follows.
\begin{definition} \label{def:2_1} \textbf{(\textit{M$k$NN query})} Let $D$ denote a set of data objects in a two dimensional database,
    $q$ the moving query point, and $k$ a positive integer.
    An M$k$NN query returns for every position of $q$, a set $A$ that consists of $k$ data
    objects whose distances from $q$ are less or equal to those of the
    data objects in $D-A$.
\end{definition}

A private \emph{static} $k$NN query protects a user's privacy while
processing a $k$NN query. Traditionally for private static $k$NN
queries, the user requests $k$ NNs\footnote{In this paper, we use
NN and nearest data object interchangeably.} to the LSP with a
rectangle that includes the current position of the
user~\cite{Damiani09.SPRINGL,gedik08.TMC,mohamed06.VLDB,Xue09.LOCA}.
Since the LSP does not know the actual location of the user in the
rectangle, it returns the $k$ nearest data objects with respect to
every point of the rectangle.

There is no universally accepted view on what privacy protection
implies for a user. On the one hand, it could mean hiding the
user's identity but revealing the user's precise location while
accessing an LBS, which prevents an LSP from knowing what type of
services have been accessed by whom. On the other hand, it could
mean protecting privacy of the user's location while disclosing
the user's identity to the LSP.

For the first scenario, a user reveals her location to the
LSP and requests an LBS via a third party (e.g., pseudonym service
provider) to hide her identity from the LSP. However, accessing an LBS anonymously does not always protect the user's privacy since the LSP
could infer the user's identity from the revealed location. For
example, if a user requests a service from her home, office or any
other place that is known to the LSP then the user can be
identified. To address this issue, $K$-anonymity
techniques~\cite{gedik05.ICDCS,gruteser03.MOBISYS} have been
developed. In $K$-anonymity techniques, the user's rectangle
includes $K-1$ other user locations in addition to the user's
location and thus make the user's identity indistinguishable from
$K-1$ other users even if the actual user locations are known to
the LSP.

In this paper, we consider the second scenario where the user's
location is unknown to the LSP since the user considers her
location as private and sensitive information. We address how to
protect privacy of the user's trajectory when the user's identity
is revealed, and do not use $K$-anonymity for the following
reasons:
\begin{enumerate}
    \item $K$-anonymity techniques hide the user's identity from the
LSP and assume that the user's location could be known to the LSP.
On the other hand, our focus is to protect the user's trajectory
privacy while disclosing the user's identity. Revealing the user's
identity enables the LSP to provide personalized query
answers~\cite{ghinita09.GIS,Xue09.LOCA}; as an example the LSP can
return only those gas stations as M$k$NN answers which provide a
higher discount for the user's credit card.
    \item $K$-anonymity techniques alone cannot protect privacy of the
user's location when the user's identity is revealed. For example
if a user is located at the liver clinic and there are other $K-1$
users at the same clinic, then the user's rectangle also resides
in the liver clinic. However, the rectangle needs to include other
places in addition to the liver clinic for protecting the privacy
of the user's location. The higher the number of different places
the rectangle includes in addition to the liver clinic, the lower
the probability that the user is located at the liver clinic.
Since integrating $K$-anonymity techniques in our approach do not
increase the level of privacy of a user's location, we do not
integrate $K$-anonymity techniques.
\end{enumerate}

In our approach, the user sets her rectangle area according to her
privacy requirement and the user's location cannot be refined to a
subset of that rectangle at the time of issuing the query. For
example, a user can set the size of the rectangle covering a suburb of the California or covering the whole California region if a high level privacy is required.

\begin{figure*}[htbp]
    \centering
        \includegraphics[width=0.7\textwidth]{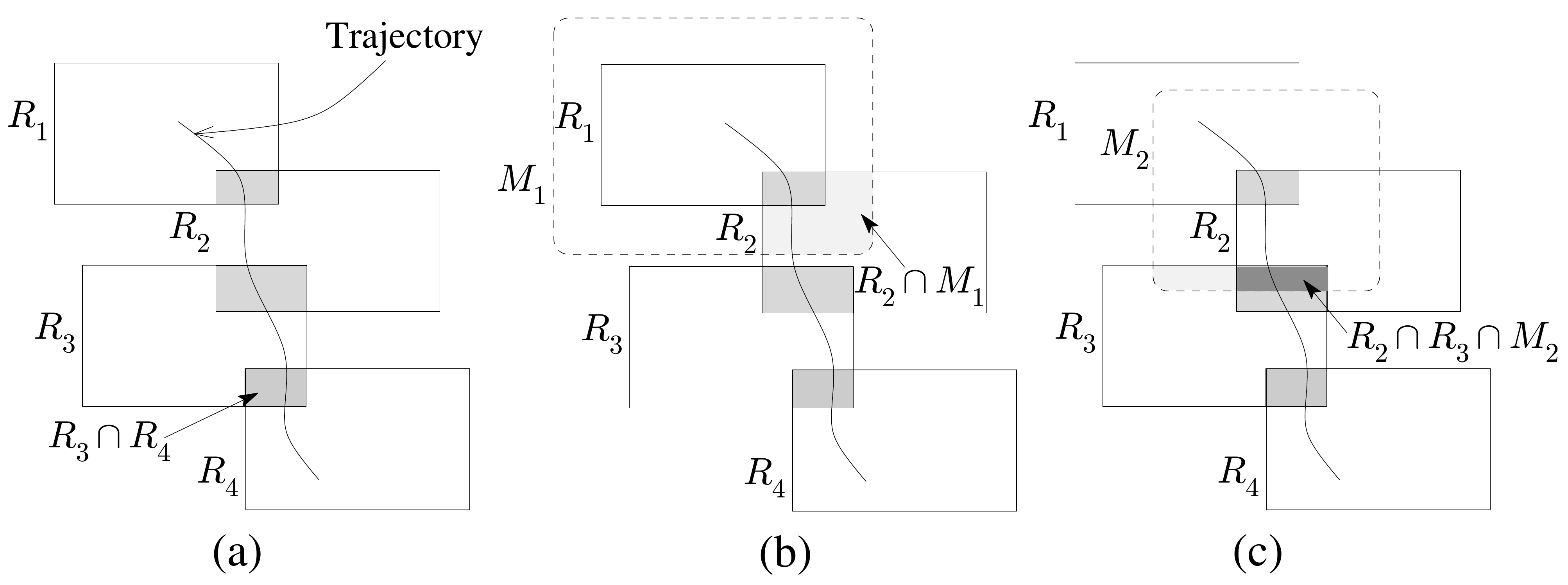}
    \caption{(a) Overlapping rectangle attack, (b) maximum movement bound attack, and (c) combined attack}
    \label{fig:traj_attack}
\end{figure*}

For a private moving $k$NN (PM$k$NN) query, a straightforward attempt to address the PM$k$NN query is to apply the private static
$k$NN query iteratively such that the user has the
$k$ nearest data objects for every position of $q$, where the
moving user's locations are updated in a periodic manner. However,
the straightforward application of private static $k$NN queries
for processing an M$k$NN query cannot protect the user's
\emph{trajectory privacy}, which is explained in the next section.

\subsection{Threat model for M$k$NN queries}


Applying private static $k$NN queries to a PM$k$NN query requires
that the user (the moving query point) continuously updates her
location as a rectangle to an LSP so that the $k$NN answers are
ensured for every point of her trajectory. The LSP simply returns
the $k$ NNs for every point of her requested rectangle. Thus, the
moving user already has the $k$ NNs for every position in the
current rectangle. Since an M$k$NN query provides answers for
every point of the user's trajectory, the next request for a new
rectangle can be issued at any point before the user leaves the
current rectangle. We also know that in a private static $k$NN
query, a rectangle includes the user's current location at the
time of requesting the rectangle to the LSP. Therefore, a
straightforward application of private static $k$NN queries for
processing an M$k$NN query requires the overlap of consecutive
rectangles as shown in Figure~\ref{fig:traj_attack}(a). These
overlaps refine the user's locations within the disclosed
rectangles to the LSP and decrease the privacy of the user's
location. In the worst case, a user can issue the next request for
a new rectangle when the user moves to the boundary of the current
rectangle to ensure the availability of $k$NN answers for every
point of the user's trajectory in real time. Even in this worst
case scenario, the consecutive rectangles needs to overlap at
least at a point, which is the user's current location. We define
the above described privacy threat as the \emph{overlapping
rectangle attack}.

\begin{definition} \label{def:2_2} \textbf{(\textit{Overlapping rectangle attack})} Let $\{R_1, R_2,...,R_n\}$ be a set of $n$ consecutive rectangles
requested by a user to an LSP in an M$k$NN query, where $R_w$ and
$R_{w+1}$ overlap for $1\leq w<n$. Since a user's location lies in
the rectangle at the time it is sent to the LSP and the moving
user requires the $k$ NNs for every position, the user's location
has to be in $R_w \cap R_{w+1}$ at the time of sending $R_{w+1}$,
and the user's trajectory must intersect $R_w \cap R_{w+1}$. As
$(R_w \cap R_{w+1})\subset R_w, R_{w+1}$, the overlapping
rectangle attack enables an LSP to render more precise locations
of a user and gradually reveal the user's trajectory.
\end{definition}

There is another possible attack on a user's trajectory privacy
for M$k$NN queries when the user's maximum velocity is known.
Existing
research~\cite{cheng06.PET,ghinita09.GIS,Hu09.TKDE,Xu09.TPDS} has
shown that if an LSP has rectangles from the same user at
different times and the LSP knows the user's maximum velocity,
then it is possible to refine a user's approximated location from
the overlap of the current rectangle and the maximum movement
bound with respect to the previous rectangle, called \emph{maximum
movement bound attack}. Figure~\ref{fig:traj_attack}(b) shows an
example of this attack in an M$k$NN query that determines more
precise location of a user in the overlap of $R_2$ and the maximum
movement bound $M_1$ with respect to $R_1$ at the time of sending
$R_2$.

For an M$k$NN query, the maximum movement bounding attack is
weaker than the overlapping rectangle attack as $(R_w \cap
R_{w+1}) \subset (M_w \cap R_{w+1})$. However, we observe that the
combination of overlapping rectangle and maximum movement bound
attacks can be stronger than each individual attack as shown in
Figure~\ref{fig:traj_attack}(c). In this example at the time of
issuing $R_3$, the LSP derives $M_2$ from $R_1 \cap R_2$ rather
than from $R_2$ and identifies the user's more precise location as
$R_2 \cap R_3 \cap M_2$, where $(R_2 \cap R_3 \cap M_2)\subset
(R_2 \cap R_3)$ and $(R_2 \cap R_3 \cap M_2)\subset (R_3 \cap
M_2)$.

With the above described attacks, the LSP can progressively find
more precise locations of a user and approximate the user's
trajectory. As a result the LSP could also generate a complete
profile of the user's activities from the identified trajectory.
Hence, protecting the trajectory privacy of users as much as
possible while processing an M$k$NN query is essential.

\subsection{Trajectory privacy for M$k$NN queries}

Trajectory privacy protection with respect to a rectangle is
defined as follows:

\begin{definition} \label{def:2_3} \textbf{(\textit{Trajectory privacy protection with respect to a
rectangle})} The user's trajectory privacy is protected with
respect to a rectangle, if the following conditions hold:

\begin{enumerate}
    \item The user's location at the time of sending a rectangle
    cannot be refined to a subset of that rectangle.
    \item The user's trajectory cannot be refined to a subset of that rectangle.
\end{enumerate}
\end{definition}

The first condition removes the certainty that the location of a
user at the time of issuing a rectangle is within the overlap of
rectangles and the maximum movement bound. The second condition
ensures that a user's trajectory does not have to intersect the
overlap of consecutive rectangles.

A privacy protection technique that satisfies Definition~\ref{def:2_3} can
overcome the overlapping rectangle attack and the maximum movement bound attack
that refine parts of a user's trajectory within the rectangles. However, the LSP
can still refine the user's trajectory within the data space from the available
knowledge of the LSP. Since there is no measure to quantify trajectory privacy,
we measure trajectory privacy as the (smallest) area to which an adversary can
refine the trajectory location relative to the data space. We call it
\emph{trajectory area} and define it in the Section Experiments, as it requires
concepts which are introduced later in the paper.
Note that the larger the trajectory area is, the higher is the user's trajectory
privacy and the higher is the probability that the area is associated with
different sensitive locations and, as a result, the lower is the probability
that the user's trajectory could be linked to a specific location.
We also measure a user's trajectory privacy by the number of
requested rectangles per trajectory for a fixed area, i.e., the
\emph{frequency}, the smaller the number of requested rectangles, the less
spatial constraints are available to the LSP for predicting the trajectory.

\subsubsection{Overview of our approach for PM$k$NN queries}

A na{\"i}ve solution to avoid overlapping rectangles is to request
next rectangle after the user leaves the current rectangle.
However, this solution cannot provide an answer for the
part of the trajectory between two rectangles: this violates the
definition of M$k$NN query, which asks for $k$ NNs for every point of the trajectory. Figure~\ref{fig:naive_approach} shows
an example, where a user requests non overlapping rectangles and
thus the user does not have $k$NN answers for parts of the
trajectory between points $q_1$ and $q_2$, and $q_3$ and $q_4$.

\begin{figure}[htbp]
    \centering
        \includegraphics[height=50mm]{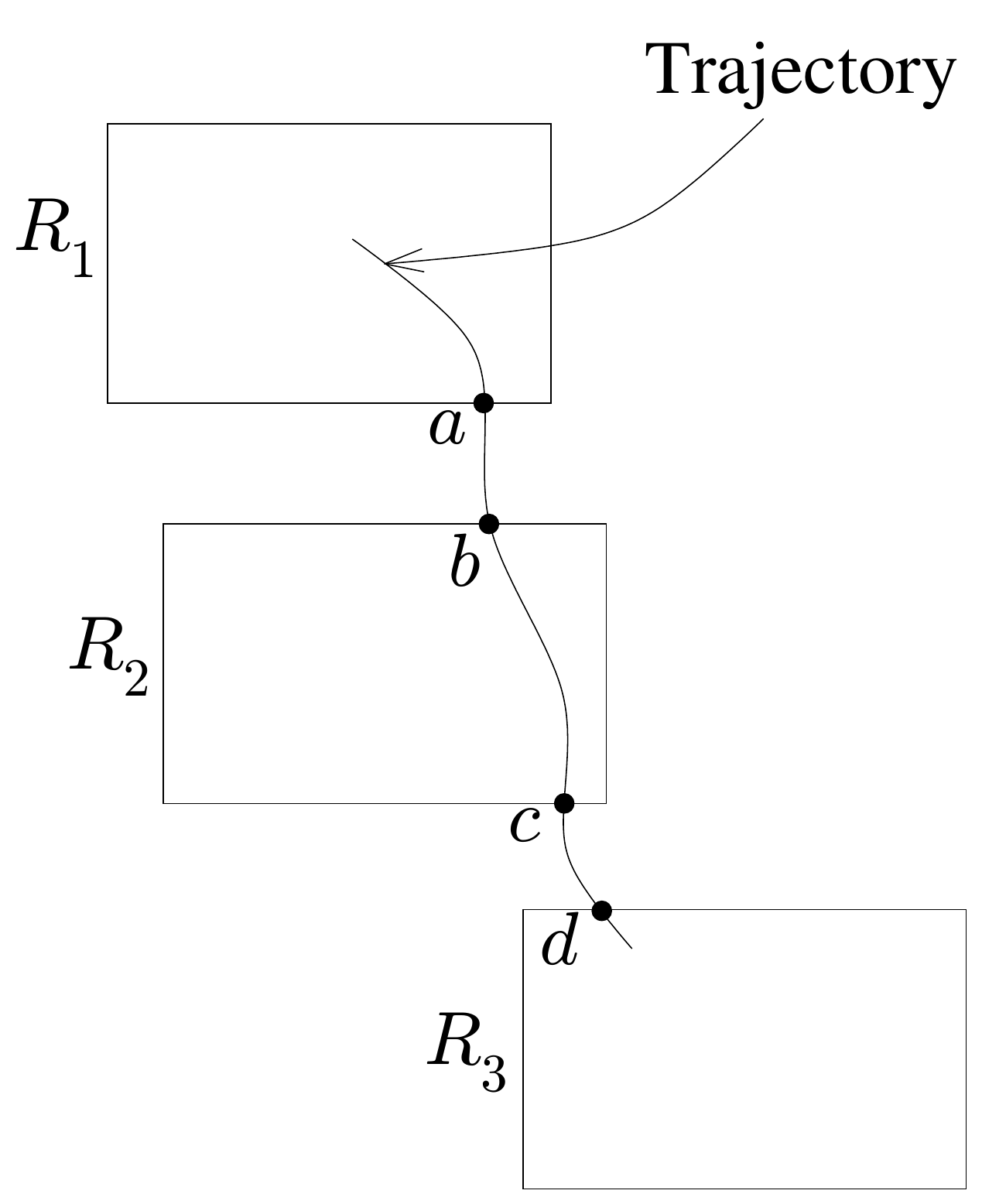}
    \caption{A na{\"i}ve solution: $k$NN answers may not be available to the user for parts of the trajectory between $q_1$ and $q_2$, and $q_3$ and $q_4$}
    \label{fig:naive_approach}
\end{figure}

In this paper, we propose a solution to overcome the overlapping
rectangle attack on the user's trajectory privacy for M$k$NN
queries. We ensure that the proposed solution satisfies the two
required conditions for trajectory privacy protection (see
Definition 3) for every rectangle requested to the LSP and
provides the user $k$NN answers for every point of her trajectory.
In our approach, a user does not always need to send
non-overlapping rectangles to avoid the overlapping rectangle
attack. We show that our approach does not allow the LSP to refine
the user's location or trajectory within the rectangle even if the
user sends overlapping rectangles. The underlying idea is to have
the required answers for an additional part of the user's
trajectory without the LSP's knowledge. As the user has the
required answers for an additional part of her trajectory, the
consecutive rectangles do not have to always overlap. Even if the
rectangles overlap, there is no guarantee that the user is located
in the overlap at the time of sending the rectangle to the LSP and
the user's trajectory passes through the overlap. To achieve the
answers for an additional part of the user's trajectory without
informing the LSP, the user requests a higher confidence level and
a higher number of NNs than required and does not reveal the
required values to the LSP. Our approach also prevents the maximum
movement bound attack based on the existing
solutions~\cite{cheng06.PET,ghinita09.GIS,Hu09.TKDE,Xu09.TPDS} in
the literature if the LSP knows the user's maximum velocity.

\section{Related Work}\label{sec:relwork}

Section~\ref{sec:relwork_privacy} surveys existing research on
protecting trajectory privacy in continuous LBSs and
Section~\ref{sec:relwork_trajectoryprivacy} highlights the
trajectory privacy concern in other applications.

\subsection{Privacy protection in continuous LBSs}
\label{sec:relwork_privacy} Most research on user privacy in LBSs
has focused on static location-based queries that include nearest
neighbor
queries~\cite{gedik08.TMC,hashem07.UBICOMP,hashem.PMC,kalnis07.TKDE,mohamed06.VLDB,yiu08.ICDE},
group nearest neighbor queries~\cite{hashem10.EDBT} and proximity
services~\cite{Mascetti09.MDM}. Different strategies such as
$K$-anonymity, obfuscation, $l$-diversity, and cryptography have
been proposed to protect the privacy of users.

$K$-anonymity techniques
(e.g.,~\cite{gruteser03.MOBISYS,mohamed06.VLDB}) make a user's
identity indistinguishable within a group of $K$ users.
Obfuscation techniques
(e.g.,~\cite{duckham05.PERVASIVE,yiu08.ICDE}) degrade the quality
of a user's location by revealing an imprecise or inaccurate
location and $l$-diversity techniques
(e.g.,~\cite{Damiani09.SPRINGL,Xue09.LOCA}) ensure that the user's
location is indistinguishable from other $l-1$ diverse locations.
Both obfuscation, and $l$-diversity techniques focus on hiding the
user's location from the LSP instead of the identity.
Cryptographic techniques
(e.g.,~\cite{ghinita08.SIGMOD,khoshgozaran07.SSTD}) allow users to
access LBSs without revealing their locations to the LSP, however,
these techniques incur cryptographic overhead and require an
encrypted database. In this paper, we assume that the LSP
evaluates a PM$k$NN query on a non-encrypted database.

$K$-anonymity, obfuscation, or $l$-diversity based approaches for
private static queries cannot protect privacy of users for
continuous LBSs because they consider each request of a continuous
query as an independent event, i.e., the correlation among the
subsequent requests is not taken into account. Recently different
approaches~\cite{bettini05.SDM,chow07.SSTD,Xu07.ACMGIS,Gkoulalas-Divanis09.SDM,cheng06.PET,ghinita09.GIS,Xu09.TPDS,Xu08.INFOCOM}
have been proposed to address this issue.

The authors in $K$-anonymity based
approaches~\cite{bettini05.SDM,chow07.SSTD,Xu07.ACMGIS,Gkoulalas-Divanis09.SDM}
for continuous queries focus on the privacy threat on a user's
identity that arises from the intersection of different sets of
$K$ users involved in the consecutive requests of a continuous
query. Since we focus on how to hide a user's trajectory while
disclosing the user's identity to the LSP, these approaches are
not applicable for our purpose. On the other hand, existing
obfuscation and $l$-diversity based
approaches~\cite{cheng06.PET,ghinita09.GIS,Xu09.TPDS} for
continuous queries have only addressed the threat of the maximum
movement bound attack. However, none of these approaches have
identified the threat on trajectory privacy that arises from the
overlap of consecutive regions (e.g., rectangles). The trajectory
anonymization technique proposed in~\cite{Xu08.INFOCOM} assumes
that a user knows her trajectory in advance for which an LBS is
required, whereas other approaches including ours consider an
unknown future trajectory of the user.

\subsubsection{Existing $k$NN algorithms}\label{sec:relwork_query}

To provide the query answers to the user, the LSP needs an
algorithm to evaluate a $k$NN query for the user's location.
\textit{Depth first search} (DFS)~\cite{roussopoulos95.SIGMOD} and
\textit{best first search} (BFS)~\cite{hjaltason95.SSD} are two
well known algorithms to find the $k$ NNs with respect to a point
using an $R$-tree~\cite{guttman84.SIGMOD}. If the value of $k$ is
unknown, e.g., for an incremental $k$NN queries, the next set of
NNs can be determined with BFS. We use BFS in our proposed
algorithm to evaluate a $k$NN query with respect to a rectangle.
The BFS starts the search from the root of the $R$-tree and stores
the child nodes in a priority queue. The priority queue is ordered
based on the minimum distance between the query point and the
\textit{minimum bounding rectangles} (MBRs) of $R$-tree nodes or
data objects. In the next step, it removes an element from the
queue, where the element is the node representing the MBR with the
minimum distance from the query point. Then the algorithm again
stores the child nodes or data objects of the removed node on the
priority queue. The process continues until $k$ data objects are
removed from the queue.

Researchers have also focused on developing
algorithms~\cite{chow09.TODS,chow09.SSTD,Hu06.TKDE,kalnis07.TKDE,mohamed06.VLDB,Xu09.TPDS}
for evaluating a $k$NN query for a user's imprecise location such
as a rectangle or a circle. In~\cite{chow09.SSTD}, the authors
have proposed an approximation algorithm that ensures that the
answer set contains one of the $k$ NNs for every point of a
rectangle. The limitation of their approximation is that users do
not know how much more they need to travel with respect to the
actual NN, i.e., the accuracy of answers. Our algorithm allows
users to specify the accuracy of answers using a confidence level.

To prevent the overlapping rectangle attack, our proposed approach
requires a $k$NN algorithm that returns a candidate answer set
including all data objects of a region in addition to the $k$ NNs
with respect to every point of a user's imprecise location. The
availability of all data objects for a \emph{known region} to the
user in combination with the concept of hiding the user's required
confidence level and the required number of NNs from the LSP can
prevent the overlapping rectangle attack (see
Section~\ref{sec:trj_priv}). Among all existing $k$NN algorithms
for a user's imprecise
location~\cite{chow09.TODS,chow09.SSTD,Hu06.TKDE,kalnis07.TKDE,mohamed06.VLDB,Xu09.TPDS},
only Casper~\cite{mohamed06.VLDB} supports a known region; the
algorithm returns all data objects of a \emph{rectangular region}
(i.e., the known region) that include the NNs with respect to a
rectangle. However, Casper can only work for NN queries and it is
not straightforward to extend Casper for $k>1$. Thus, even if
Casper is modified to incorporate the confidence level concept, it
can only support PM$k$NN queries for $k=1$.

Moreover, for a single nearest neighbor query, Casper needs to
perform on the database multiple searches, which incur high
computational overhead. Casper executes four individual single
nearest neighbor queries with respect to four corner points of the
rectangle. Then using these neighbors as filters, Casper expands
the rectangle in all directions to compute a range that contains
the NNs with respect to all points of the rectangle. Finally,
Casper has to again execute a range query to retrieve the
candidate answer set. We propose an efficient algorithm that finds
the $k$NNs with a specified confidence level for a rectangle in a
single search.

\subsection{Trajectory privacy in other applications}
\label{sec:relwork_trajectoryprivacy}

Protecting a user's trajectory privacy has also received much
attention in other
domains~\cite{Abul08.ICDE,beresford03.PC,hoh07.CCS,Nergiz09.DataPrivacy,Yarovoy09.EDBT}.
The advancement and widespread use of location aware devices
(e.g., GPS equipped mobile phone or vehicle) have enabled users to
share their trajectories with others. Such trajectory data allows
organizations and researchers to perform useful analyses for many
applications such as urban planning, traffic monitoring, and
mining human behavior. To protect user trajectories, they are
modified before they are released so that both user privacy and
data utility are maintained. Recent research has developed a few
anonymization
approaches~\cite{Abul08.ICDE,Nergiz09.DataPrivacy,Yarovoy09.EDBT}
for publishing privacy preserving trajectory data, where a trusted
server first collects trajectories from users and then publishes
them in public after their anonymization. Prior
studies~\cite{beresford03.PC,hoh07.CCS} also consider scenarios
without a trusted server, which means a user's trajectory is
anonymized before it is shared with anyone. The purpose of these
approaches is to protect trajectory privacy through anonymization
while maintaining the utility of trajectory data for different
analyses. On the other hand, our approach protects trajectory
privacy while answering M$k$NN queries in a personalized manner
(i.e., the user's identity is revealed); therefore our studied
problem is orthogonal to the above problem.

\section{System Overview}
\label{sec:sys_overview}

Our approach for PM$k$NN queries is based on the client-server
model. In our system, a client is a moving user who sends a
PM$k$NN query request and the server is the LSP that processes
the query. The moving user sends her imprecise location as a
rectangle to the LSP, which we call \emph{obfuscation rectangle}
in the remainder of this paper.

We introduce the parameter confidence level, which provides a user
with an option to trade the accuracy of the query answers for trajectory
privacy. Intuitively, the confidence level of the user for a data
object guarantees that the distance of the data object to the user's
location is within a bound of the actual nearest data object's
distance. In Section~\ref{sec:conflevel}, we formally define and
show how a user and an LSP can compute the confidence level for a data
object.

In our system, a user does not reveal the required confidence
level and the required number of NNs to the LSP while requesting a
PM$k$NN query; instead the user \emph{specifies higher values than
the required ones}. This allows the user to have the required
number of NNs with the required confidence level for an additional
part of her trajectory, which is unknown to the LSP, and thus the
LSP cannot apply the overlapping rectangle attack by correlating
the user's current obfuscation rectangle with the previous one. In
Section~\ref{sec:trj_priv}, we present a technique to compute a
user's consecutive obfuscation rectangles for requesting a PM$k$NN
query. Another important advantage of our technique is that for
the computation of the consecutive obfuscation rectangles, the
user does not need to trust any other party such as an
intermediary trusted server~\cite{mohamed06.VLDB}.

An essential component of our approach for a PM$k$NN query is an
algorithm for the LSP that finds the specified number of NNs for
the obfuscation rectangle with the specified confidence level. In
Section~\ref{sec:knnq}, we exploit different properties of the
confidence level with respect to an obfuscation rectangle to
develop an efficient algorithm in a single traversal of the
$R$-tree.

\section{Confidence Level}\label{sec:conflevel}

%


The confidence level represents a measure of the accuracy for a
nearest data object with respect to a user's location. If the
confidence level of a user for the $k$ nearest data objects is 1
then they are the actual $k$ NNs. If the confidence level is less
than 1 then it provides a worst case bound of how much more a user
may need to travel than the actual $k^{th}$ nearest data object.
For example, a nearest data object with 0.5 confidence level means
that the user has to travel twice the distance to the actual NN in
the worst case.

To determine the confidence level of a user for any nearest data
object, we need to know the locations of other data objects
surrounding the user's location. The region where the location of
all data objects are known is called the \emph{known region}. We
first show how an LSP and a user compute the known region, and
then discuss the confidence level.

\begin{figure}[htbp]
    \centering
        \includegraphics[height=40mm]{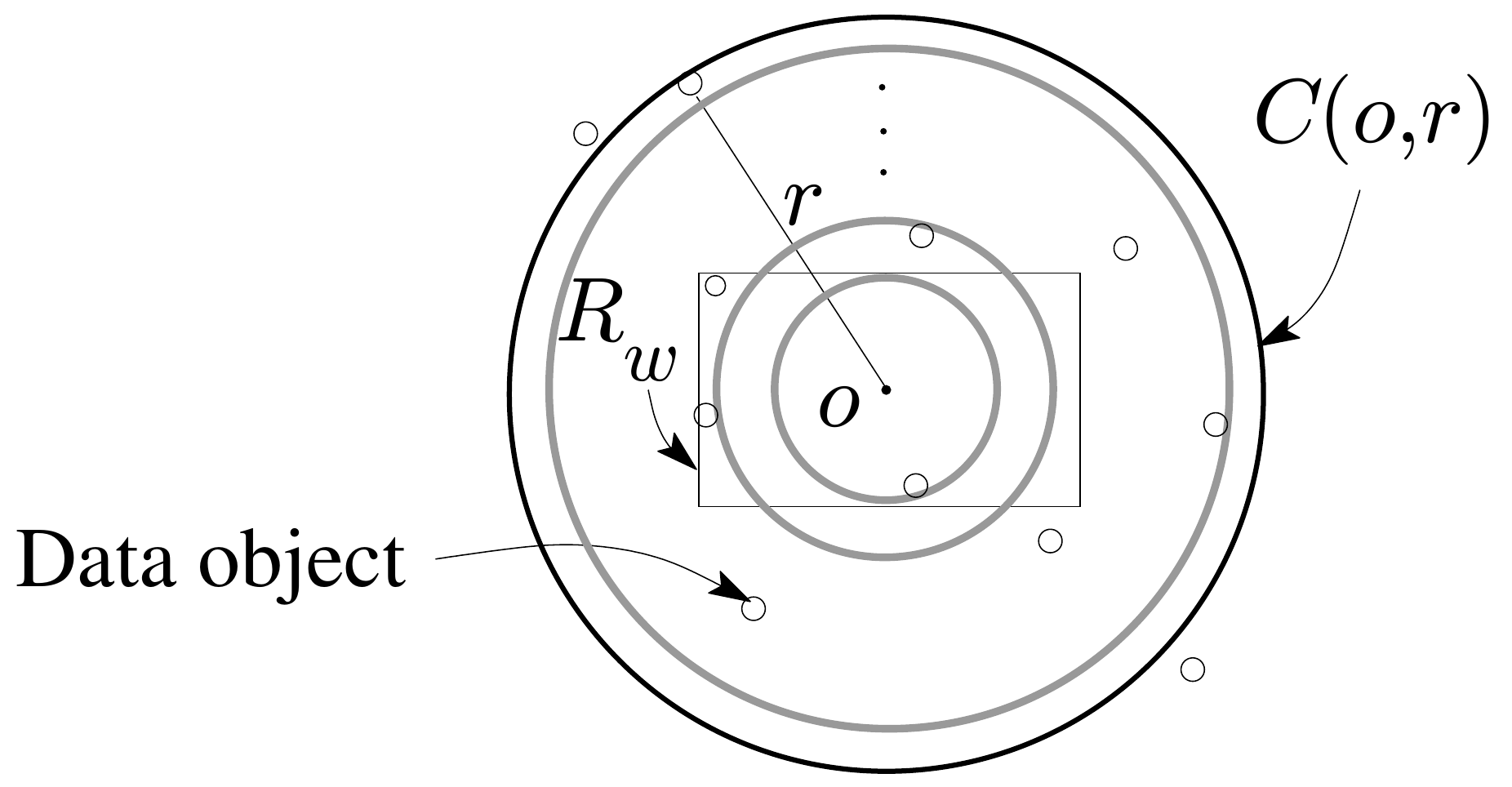}
    \caption{Known Region}
    \label{fig:knownarea}
\end{figure}

\subsection{Computing a known region}
\label{sec:knownregion} Suppose a user provides an obfuscation
rectangle $R_w$ for any positive integer $w$, to the LSP while
requesting a PM$k$NN query. For the ease of explanation, we assume
at the moment that the user specifies confidence level of 1, i.e.,
the answer set returned by the LSP includes the actual $k$NN
answers for the given obfuscation rectangle. Our proposed
algorithm for the LSP to evaluate $k$NN answers, starts a best
first search (BFS) considering the center $o$ of $R_w$ as the
query point and incrementally finds the next NN from $o$ until the
actual $k$ NNs are discovered for all points of $R_w$. The search
region covered by BFS at any stage of its execution is a circular
region $C(o,r)$, where the center $o$ is the center of $R_w$ and
the radius $r$ is the distance between $o$ and the last discovered
data object. Since the locations of all data objects in $C(o,r)$
are already discovered, $C(o,r)$ is the known region for the LSP.
The LSP returns all data objects located within $C(o,r)$ to the
user, although some of them might not be the $k$ NNs with respect
to any point of $R_w$. This enables the user to have $C(o,r)$ as
the known region. This enables the user to have $C(o,r)$ as the
known region, where the center $o$ is the center of $R_w$ and the
radius $r$ is the distance between $o$ and the farthest retrieved
data object from $o$.

\begin{figure}[htbp]
    \centering
        \hspace{-7mm}
        \includegraphics[width=120mm]{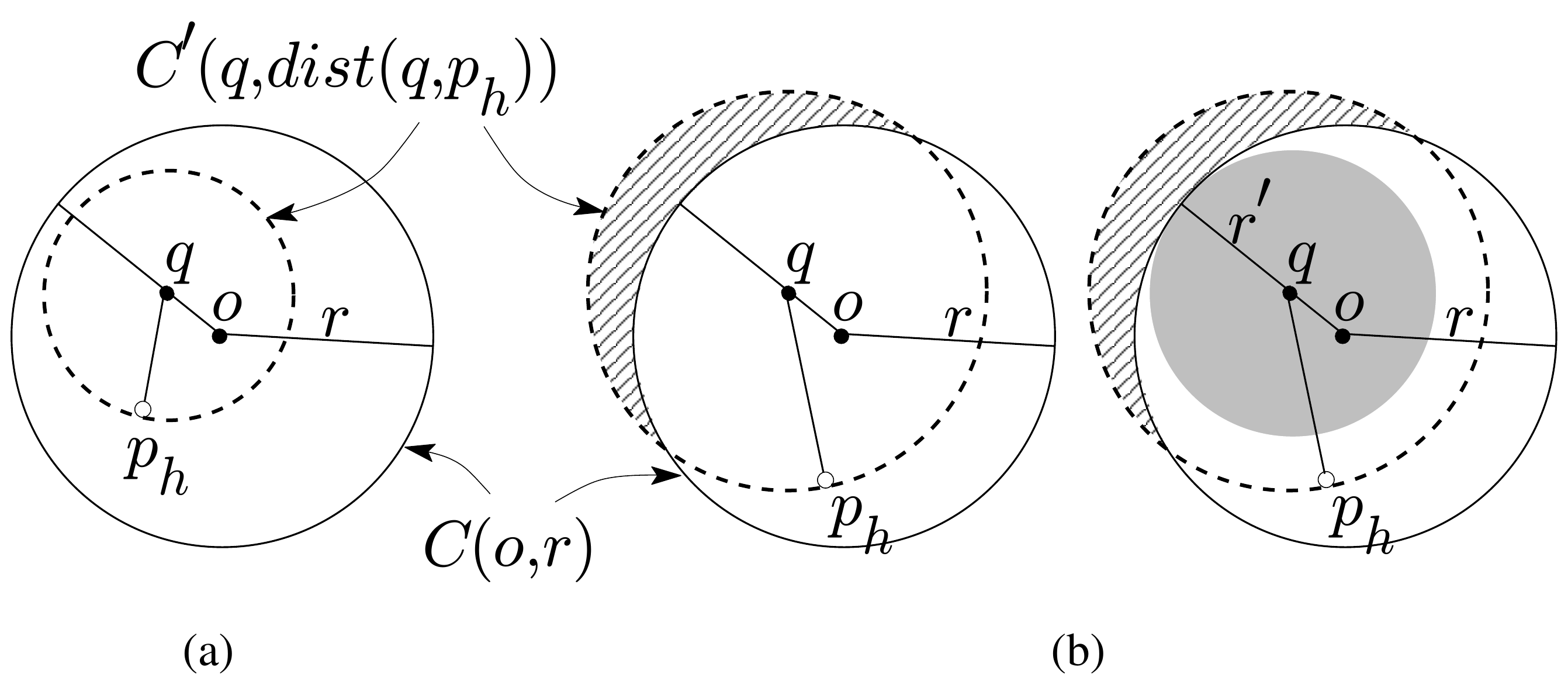}
        \hspace{-5mm}
    \caption{Confidence Level}
    \label{fig:conflevel}
\end{figure}

\subsection{Measuring the confidence level}

%




Since the confidence level can have any value in the range (0,1],
we remove our previous assumption of a fixed confidence level of 1
in Section~\ref{sec:knownregion}. In our approach, the knowledge
about the known region $C(o,r)$ is used to measure the confidence
level. Let $p_h$ be the nearest data object among all data objects
in $C(o,r)$ from a given location $q$, where $h$ is an index to
name the data objects and let $dist(q,p_h)$ represent the
Euclidean distance between $q$ and $p_h$. There are two possible
scenarios based on different positions of $p_h$ and $q$ in
$C(o,r)$. Figure~\ref{fig:conflevel}(a) shows a case where the
circular region $C^\prime(q,dist(q,p_h))$ centered at $q$ with
radius $dist(q,p_h)$ is within $C(o,r)$. Since $p_h$ is the
nearest data object from $q$ within $C(o,r)$, no other data object
can be located within $C^\prime(q,dist(q,p_h))$. This case
provides the user at $q$ with a confidence level 1 for $p_h$.
However, $C^\prime(q,dist(q,p_h))$ might not be always completely
within the known region. Figure~\ref{fig:conflevel}(b)(left) shows
such a case, where a part of $C^\prime(q,dist(q,p_h))$ falls
outside $C(o,r)$ and as the locations of data objects outside
$C(o,r)$ are not known, there might be some data objects located
in the part of $C^\prime(q,dist(q,p_h))$ outside $C(o,r)$ (i.e.,
$C^{\prime\prime}=C^\prime(q,dist(q,p_h))\setminus C(o,r)$) that
have a smaller distance than $p_h$ from $q$. Since $p_h$ is the
nearest data object from $q$ within $C(o,r)$, there is no data
object within distance $r^\prime$ from $q$
(Figure~\ref{fig:conflevel}(b)(right)), where $r^\prime$ is the
radius of the maximum circular region within $C(o,r)$ centered at
$q$. But there might be other data objects within a fixed distance
$d_f$ from $q$, where $r^\prime<d_f\leq dist(q,p_h)$. In this case
the confidence level of the user at $q$ regarding $p_h$ is less
than 1. On the other hand, if $q$ is outside of $C(o,r)$ then the
confidence level of the user at $q$ for $p_h$ is 0 because
$r^\prime$ is 0. We formally define the confidence level of a user
located at $q$ for $p_h$ in the more general case, where $p_h$ is
any of the nearest data object in $C(o,r)$.



\begin{definition}
\label{def:1} \textbf{(\textit{Confidence level})} Let $C(o,r)$ be
the known region, $P$ the set of data objects in $C(o,r)$, $q$ the
point location of a user, $p_h$ the $j^{th}$ nearest data object
in $P$ from $q$ for $1 \leq j \leq |P|$. The distance $r^\prime$
represents the radius of the maximum circular region within
$C(o,r)$ centered at $q$. The confidence level of the user located
at $q$ for $p_h$, $CL(q,p_h)$, can be expressed as:
    \begin{displaymath}{CL(q,p_h) :=
        \left\{ \begin{array}{lll}
            0 & \mbox{if $q \notin C(o,r)$}\\
            1 & \mbox{if $q\in C(o,r)$ $\wedge$ $dist(q,p_h)\leq r^\prime$}\\
            \frac{r^\prime}{dist(q,p_h)} & \mbox{otherwise}.
        \end{array} \right.}
    \end{displaymath}
\end{definition}


Since our focus is on NN queries, we use distance instead of area
as the metric for the confidence level. A distance-based metric
ensures that there is no other data object within a fixed distance
from the position of a user. Thus, the distance-based metric is a
measure of accuracy for a data object to be the nearest one. On
the other hand, an area-based metric is based on the percentage of
the area of $C^\prime(q, dist(q,p_h))$ that intersects with
$C(o,r)$. Thus, an area-based metric only could be used to express
the likelihood of an data object to be the nearest one. However,
an area-based metric cannot measure the accuracy of the data
object to be the nearest one. Furthermore, such a metric would
assume a uniform random distribution of data objects. Consider an
example where $q$ is outside $C(o,r)$ and $p_h$ is the nearest
data object from $q$ in $C(o,r)$. According to the area-based
metric the confidence level of the user for $p_h$ would be greater
than 0, i.e., $(C^\prime(q, dist(q,p_h)) \cap C(o,r)) /
C^\prime(q, dist(q,p_h)))$, although there is nothing known about
the data objects outside the known region. This measure based on
the area-based metric does not represent a bound of how much more
a user may need to travel for $p_h$ than the actual nearest data
object in the worst case.


\section{Client-side Processing}
\label{sec:trj_priv}

We present a technique for computing consecutive obfuscation
rectangles of a user to request a PM$k$NN query, where the LSP
cannot apply the overlapping rectangle attack to invade the user's
trajectory privacy. Suppose a user requests an obfuscation
rectangle $R_w$ and a confidence level $cl$ at any stage of
accessing the PM$k$NN query. The LSP returns $P$, the set of data
objects in the known region $C(o,r)$, that includes the $k$ NNs
with a confidence level at least $cl$ for every point of $R_w$.
The availability of $C(o,r)$ allows a moving user to compute the
confidence level for the $k$ NNs even from outside of $R_w$.

\begin{figure}[htbp]
  \begin{center}
    \begin{tabular}{ccc}
        \hspace{5mm}
      \resizebox{40mm}{!}{\includegraphics{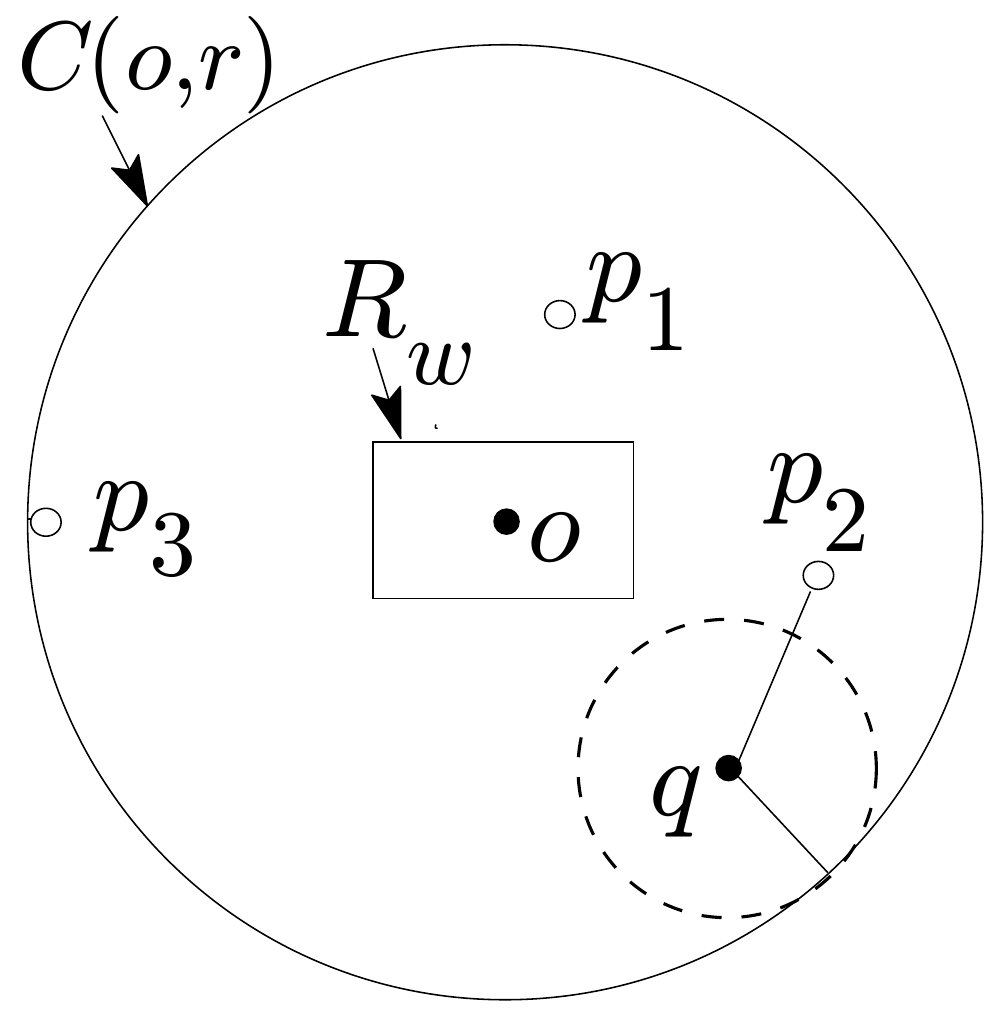}} &
       \hspace{4mm}
      \resizebox{40mm}{!}{\includegraphics{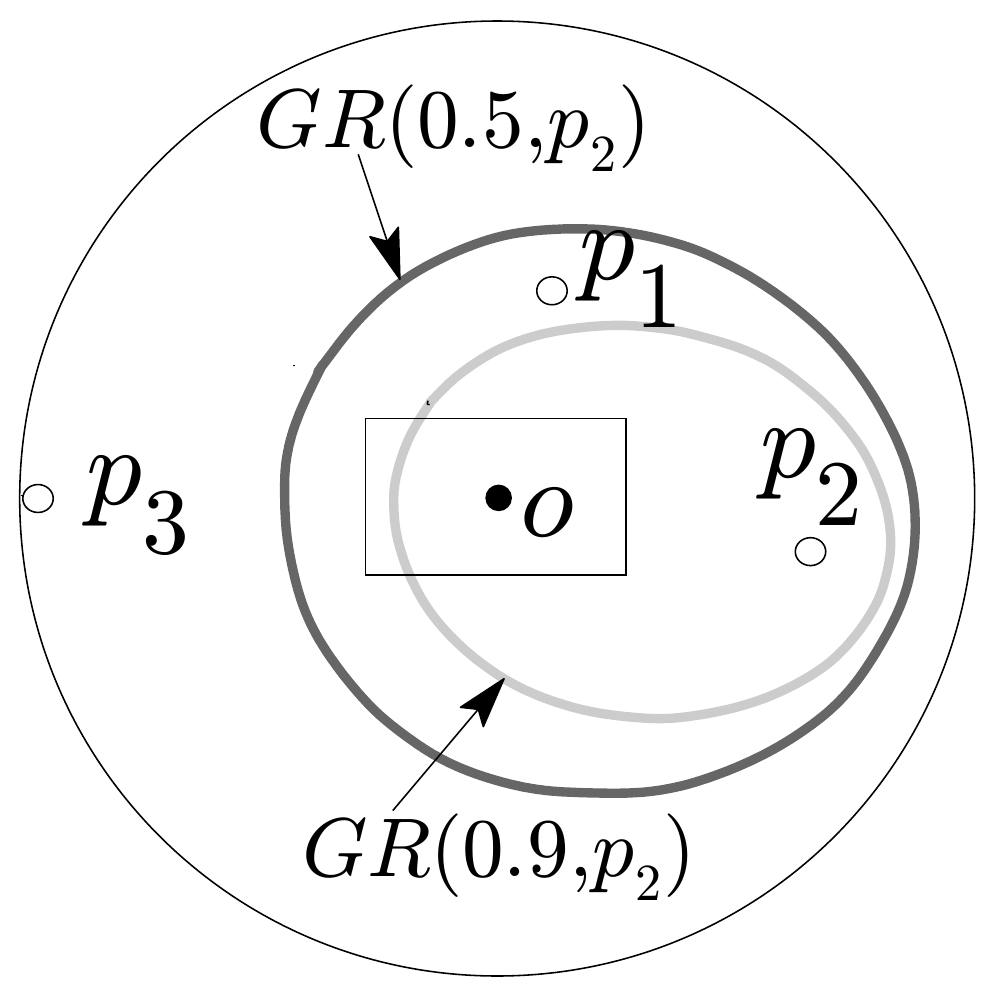}} &
      \hspace{4mm}
      \resizebox{40mm}{!}{\includegraphics{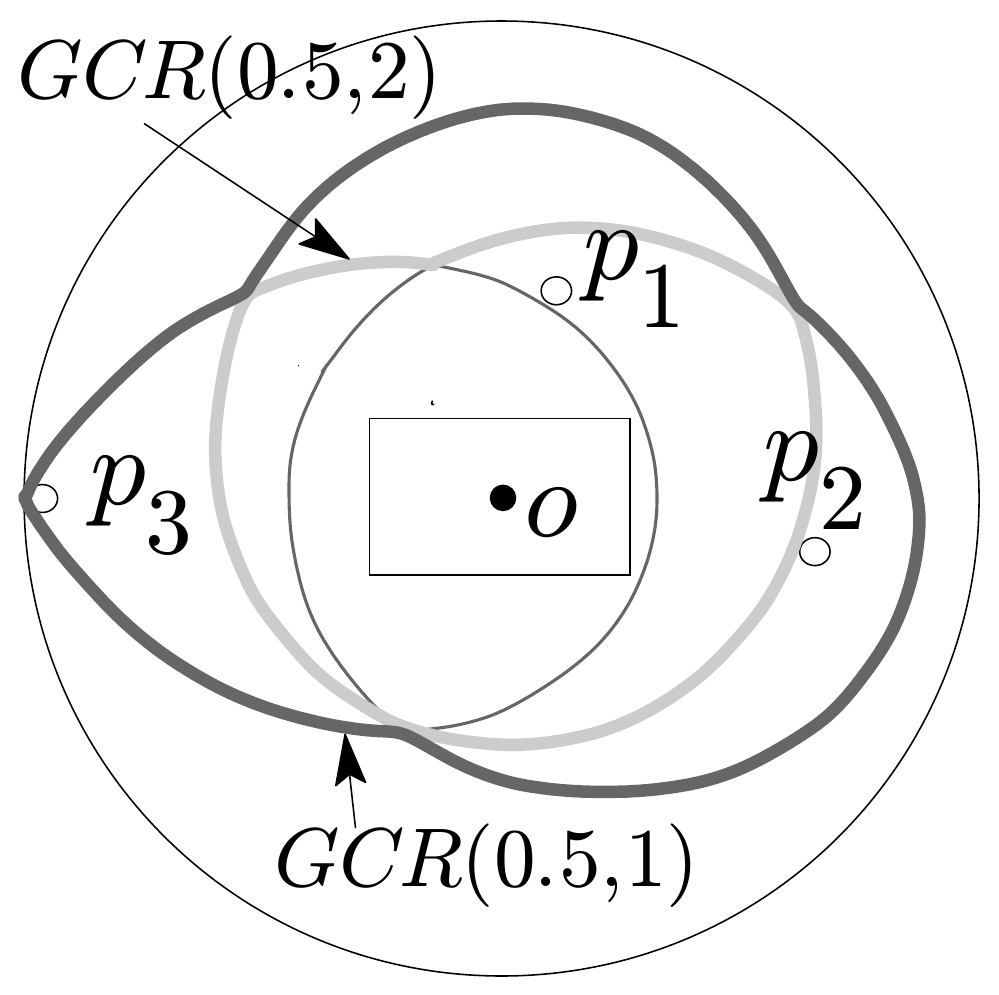}} \\
       \scriptsize{\hspace{5mm}(a)\hspace{4mm}} & \scriptsize{(b)\hspace{4mm}}& \scriptsize{(c)}
        \end{tabular}
  \end{center}
     \caption{(a) $CL(q,p_2)$, (b) $GR(cl,p_2)$ and (c) $GCR(cl,k)$}
    \label{fig:shape}
\end{figure}
Although some data objects in $P$ might not be the $k$ NNs for any
point of $R_w$, they might be $k$ NNs for a point outside $R_w$
with a confidence level at least $cl$. In addition, some data
objects, which are the $k$ NNs for some portions of $R_w$, can be
also the $k$ NNs from locations outside of $R_w$ with a confidence
level at least $cl$. For example for $cl=0.5$ and $k=1$,
Figure~\ref{fig:shape}(a) shows that a point $q$, located outside
$R_w$, has a confidence level\footnote{The confidence level of any
point represents the confidence level of a user located at that
point.} greater than 0.5 for its nearest data object $p_2$. On the
other hand, from a data object's viewpoint,
Figure~\ref{fig:shape}(b) shows two regions surrounding a data
object $p_2$,  where for any point inside these regions a user has
a confidence level at least 0.90, and 0.50, respectively for
$p_2$\footnote{Note that, whenever we mention the confidence level
of a point for a data object then the data object can be any of
the $j^{th}$ NN from that point, where $1 \leq j \leq |P|$.}. We
call such a region \emph{guaranteed region}, denoted as
$GR(cl,p_h)$ with respect to a data object $p_h$ for a specific
confidence level $cl$. We define $GR(cl,p_h)$ as follows.

\begin{definition}
\label{def:6_1} \textbf{(\textit{Guaranteed region})} Let $C(o,r)$
be the known region, $P$ the set of data objects in $C(o,r)$,
$p_h$ a data object in $P$, and $cl$ the confidence level. The
guaranteed region with respect to $p_h$, $GR(cl,p_h)$, is the set
of all points such that $\{CL(q,p_h) \geq cl\}$ for any point $q
\in GR(cl,p_h)$.
\end{definition}

From the guaranteed region of every data object in $P$ we compute
the \emph{guaranteed combined region}, denoted as $GCR(cl,k)$,
where for any point in this region a user has at least $k$ data
objects with a confidence level at least $cl$.
Figure~\ref{fig:shape}(c) shows an example, where $P=\{p_1, p_2,
p_3\}$ and $cl=0.5$. Then for $k=1$, the black bold line shows the
boundary of $GCR(0.5,1)$, which is the union of $GR(0.5,p_1)$,
$GR(0.5,p_2)$ and $GR(0.5,p_3)$. For $k=2$, the ash bold line
shows the boundary of $GCR(0.5,2)$, which is the union of
$GR(0.5,p_1) \cap GR(0.5,p_2)$, $GR(0.5,p_2) \cap GR(0.5,p_3)$ and
$GR(0.5,p_3) \cap GR(0.5,p_1)$. We define $GCR(cl,k)$ as follows.

\begin{definition}
\label{def:6_2} \textbf{(\textit{Guaranteed combined region})} Let
$C(o,r)$ be the known region, $P$ the set of data objects in
$C(o,r)$, $p_h$ a data object in $P$, $cl$ the confidence level,
$k$ the number of data objects, and $GR(cl,p_h)$ the guaranteed
region. The guaranteed combined region, $GCR(cl,k)$, is the union
of the regions where at least $k$ $GR(p_h,cl)$ overlap, i.e., $\cup_{{P^\prime} \subseteq P \wedge |{P^\prime}|=k}\{\cap_{h \in P^\prime} GR(p_h,cl)\}$.
\end{definition}

Since for any point in $GCR(cl,k)$, a user has at least $k$ data
objects with a confidence level at least $cl$, the following lemma
shows that for any point in $GCR(cl,k)$ the user also has the $k$
NNs with a confidence level at least $cl$.

\begin{lemma}
\label{lemma:6_zero} If the confidence level of a user located at
$q$ is at least $cl$ for any $k$ data objects, then the confidence
level of the user is also at least $cl$ for the $k$ NNs from $q$.
\end{lemma}

\begin{proof}
(By contradiction) Assume to the contrary that for the user at $q$
has a confidence level less than $cl$ for the $i^{th}$ NN among
the data objects, where $1 \leq i \leq k$. We know that the user
at $q$ has $k$ data objects with at least confidence level $cl$.
According to the assumption these $k$ data objects must not be the
user's $k$ NNs; at least one of them, say $p_1$, is at a greater
distance than the $k^{th}$ NN from $q$. But according to
Definition~\ref{def:1}, we know that the confidence level of the
user for the $j^{th}$ NN is greater than the $(j+1)^{th}$ NN for
$1 \leq j \leq |P|-1$. This implies that since $CL(q,p_1)\geq cl$
and $p_1$ is located farther than the $k$ NNs from $q$, the user
has a confidence level at least $cl$ for the $k$ NNs, which
contradicts our assumption. \hfill\hfill
\end{proof}

In our technique, the moving user can use the retrieved data objects from the outside
of $R_w$ and delay the next request with a new obfuscation
rectangle $R_{w+1}$ until the user leaves $GCR(cl,k)$. Although
delaying the next request with $R_{w+1}$ in this way may allow a
user to avoid an overlap of $R_w$ and $R_{w+1}$, the threat to
trajectory privacy is still in place. Since the LSP can also
compute $GCR(cl,k)$, similar to the overlapping rectangle attack,
the user's location can be computed more precisely by the LSP from
the overlap of the new obfuscation rectangle $R_{w+1}$ and current
$GCR(cl,k)$ (see Figure~\ref{fig:csaferegion}(a) for $GCR(0.5,k)
\cap R_{w+1}$).
\begin{figure*}[htbp]
  \begin{center}
    \begin{tabular}{cccc}
        \hspace{-5mm}
      \resizebox{40mm}{!}{\includegraphics{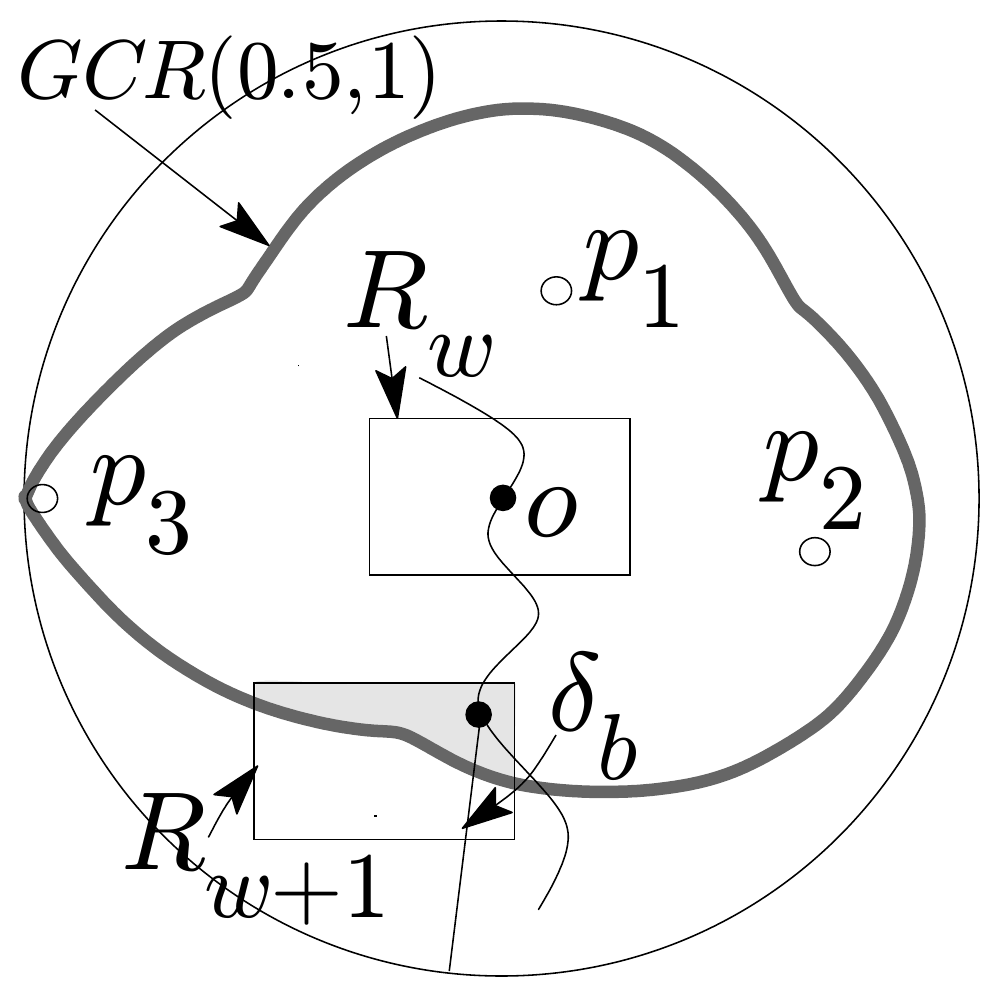}} &
       \hspace{-4mm}
      \resizebox{46mm}{!}{\includegraphics{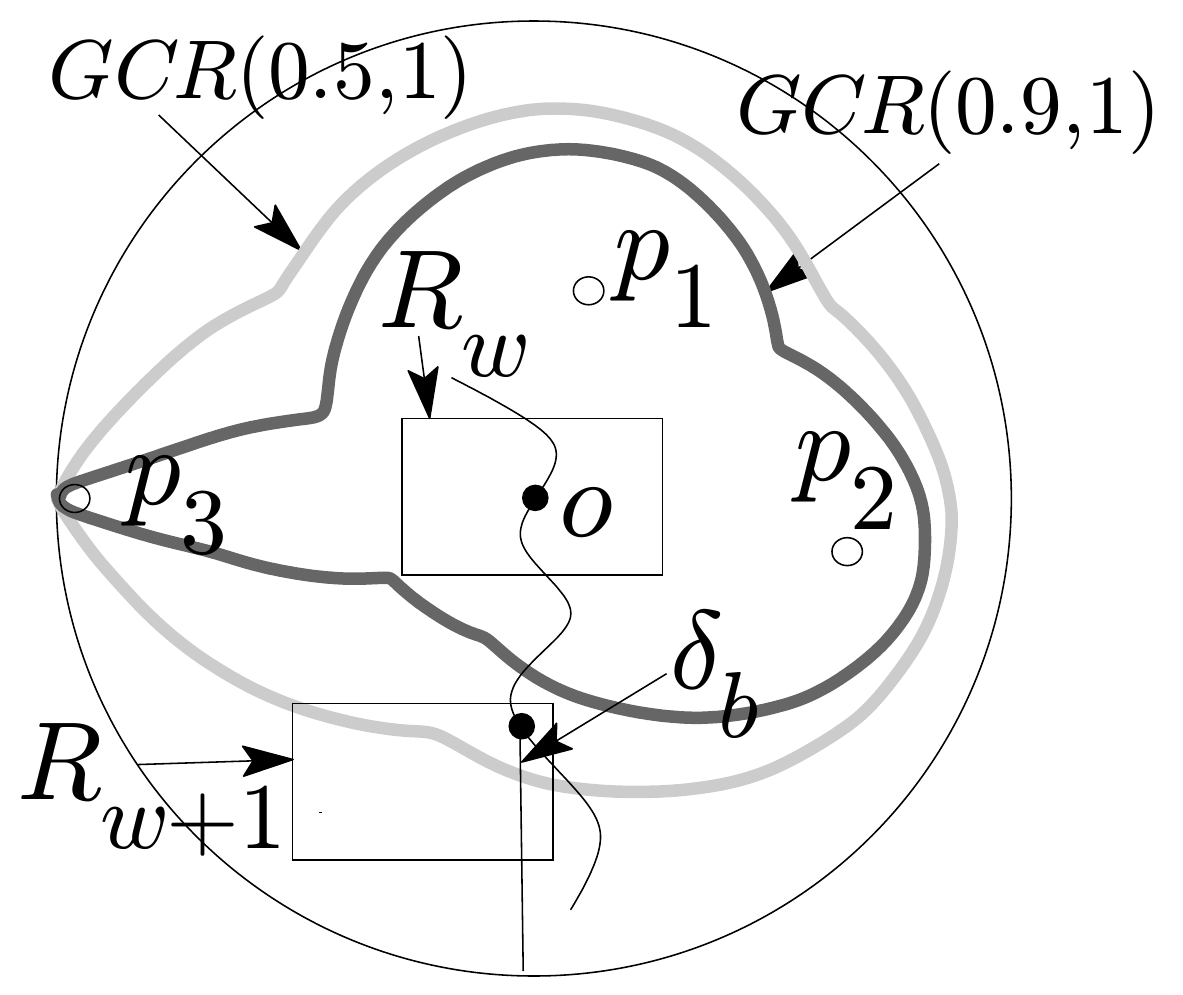}} &
       \hspace{-4mm}
      \resizebox{40mm}{!}{\includegraphics{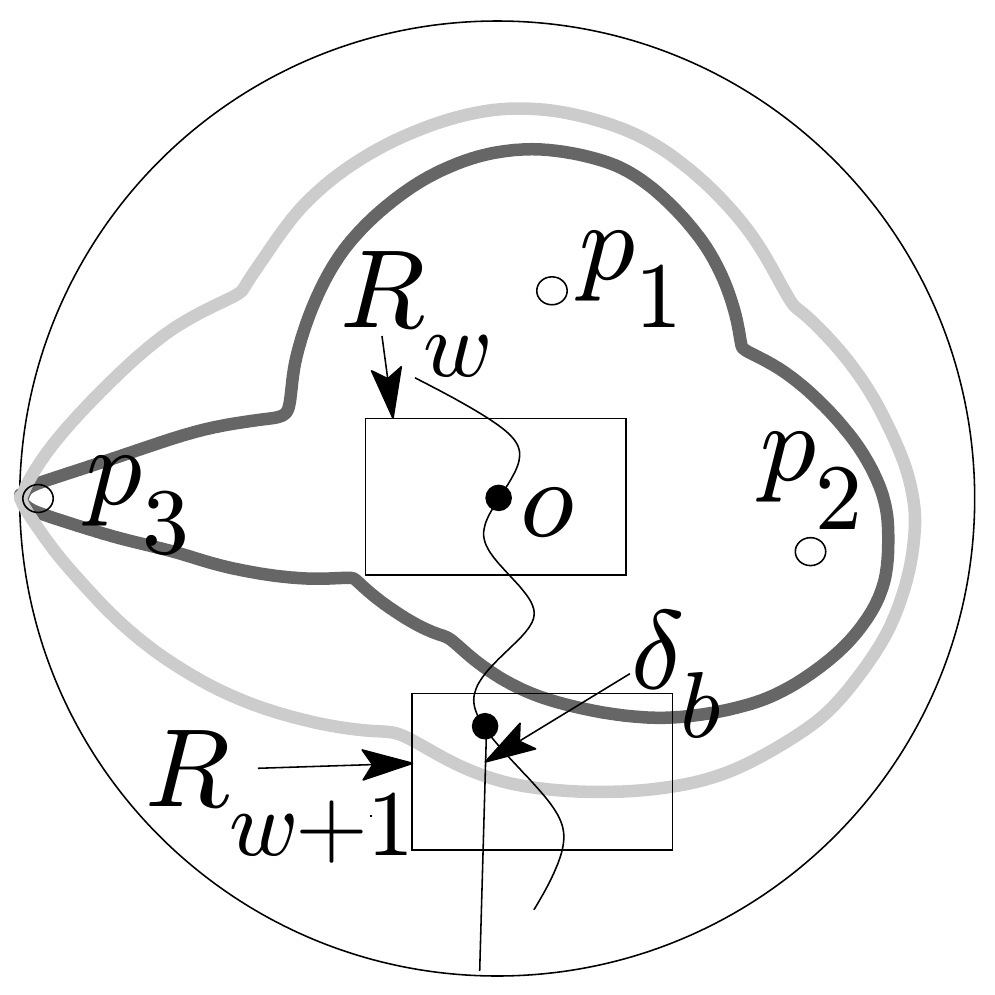}} &
      \hspace{-4mm}
      \resizebox{40mm}{!}{\includegraphics{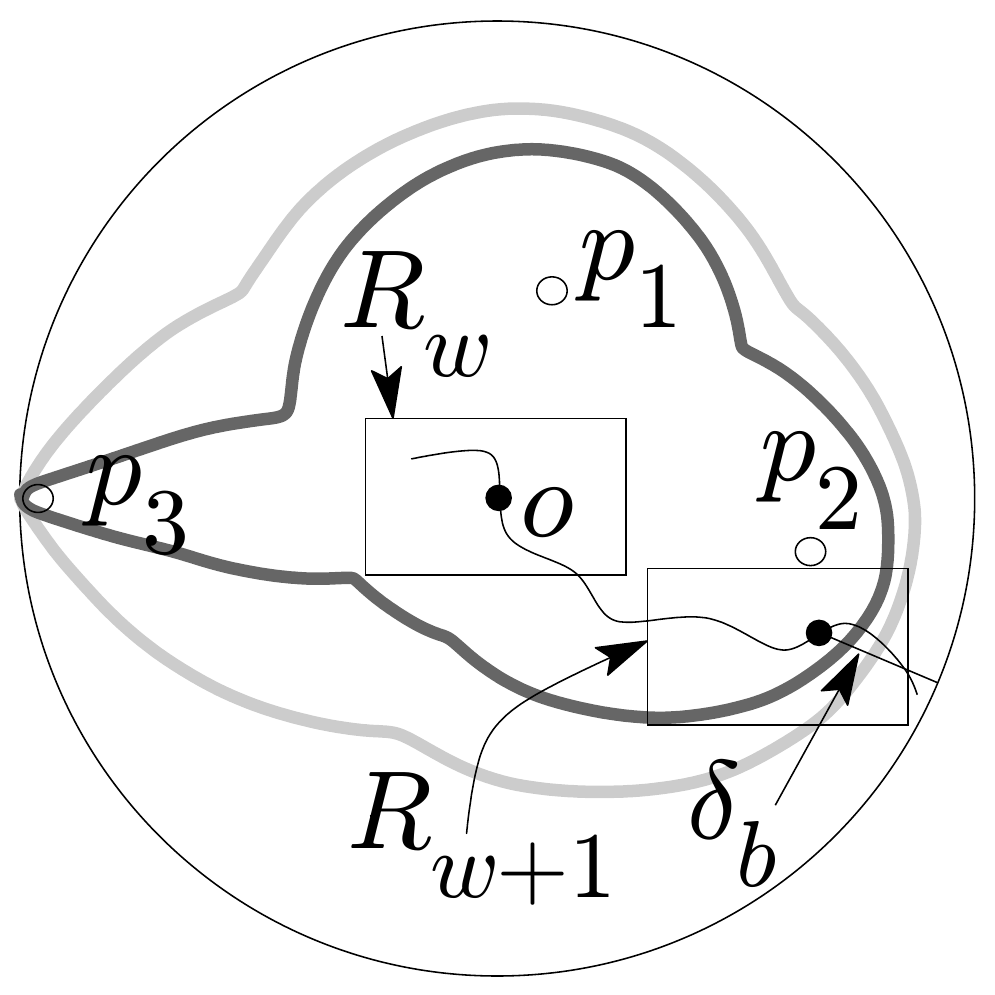}} \\
       \scriptsize{(a)\hspace{0mm}} & \scriptsize{(b)\hspace{-0mm}}& \scriptsize{(c)}& \scriptsize{(d)}
        \end{tabular}
        \end{center}
    \caption{(a) An attack from $R_{w+1} \cap GCR(0.5,1)$, (b)-(d)
    Removal of attacks with $cl_{r}=0.5$ and $cl=0.9$}
    \label{fig:csaferegion}
\end{figure*}

To overcome the above mentioned attack and the overlapping
rectangle attack, the key idea of our technique is \emph{to increase the
size of $GCR$ without informing the LSP about this extended
region}. To achieve the extended region of $GCR$ without informing
the LSP, the user has three options while requesting a
PM$k$NN query: the user specifies a higher
value than (i) the required confidence level or (ii) the required
number of nearest data objects or (iii) both. It is important to note that the user does not
reveal the required confidence level and the required number of
NNs to the LSP. Let $cl_r$ and $k_r$ represent the required
confidence level and the required number of NNs for a user,
respectively, and $cl$ and $k$ represent the specified confidence
level and the specified number of NNs to the LSP by the user,
respectively.

Consider the first option, where a user specifies a higher value
than the required confidence level, i.e., $cl > cl_r$. We know
that the $GCR$ is constructed from $GR$s of data objects in $P$
and the $GR$ of a data object becomes smaller with the increase of
the confidence level for a fixed $C(o,r)$ as shown in
Figure~\ref{fig:shape}(b), which justifies the following lemma.

\begin{lemma}
\label{lemma:6_one} Let $cl>cl_r$ and $k=k_r$. Then
$GCR(cl_r,k_r)\supset GCR(cl,k)$ for a fixed $C(o,r)$.
\end{lemma}

Since $GCR(cl_r,k_r)\supset GCR(cl,k)$, now the user can delay the
next request with a new obfuscation rectangle $R_{w+1}$ until the
user leaves $GCR(cl_r,k_r)$. Since the LSP does not know about
$GCR(cl_r,k_r)$, it is not possible for the LSP to find more
precise trajectory path from the overlap of $GCR(cl_r,k_r)$ and
$R_{w+1}$. Figure~\ref{fig:csaferegion}(b) shows an example for
$k=1$, where a user's required confidence level is $cl_r=0.5$ and
the specified confidence level is $cl=0.9$. The LSP does not know
about the boundary of $GCR(0.5,1)$ and thus cannot find the user's
precise location from the overlap of $GCR(0.5,1)$ and $R_{w+1}$.

However, the next location update $R_{w+1}$ has to be in $C(o,r)$
of $R_w$. Otherwise, the LSP is able to determine more precise
location of the user as $R_{w+1} \cap C(o,r)$ at the time of
requesting $R_{w+1}$. For any location outside $C(o,r)$, the user
has a confidence level $0$ which in turn means that $q$ cannot be
within the region of $R_{w+1}$ that falls outside $C(o,r)$ at the
time of requesting $R_{w+1}$. As a result whenever $C(o,r)$ is
small, then the restriction might cause a large part of $R_{w+1}$
to overlap with $GCR(cl,k)$ and $R_w$. The advantage of our
technique is that this overlap does not cause any privacy threat
for the user's trajectory due to the availability of
$GCR(cl_r,k_r)$ to the user. Since there is no guarantee that the
user's trajectory passes through the overlap or not, the LSP is
not able to determine the user's precise trajectory path from the
overlap of $R_{w+1}$ with $GCR(cl,k)$ and $R_w$. Without loss of
generality, Figures~\ref{fig:csaferegion}(c)
and~\ref{fig:csaferegion}(d) show two examples, where $R_{w+1}$
overlaps with $GCR(0.9,1)$ for $cl_r=0.5$, $cl=0.9$, and $k=1$. In
Figure~\ref{fig:csaferegion}(c) we see that the user's trajectory
does not pass through $GCR(0.9,1) \cap R_{w+1}$, whereas
Figure~\ref{fig:csaferegion}(d) shows a case, where the user's
trajectory passes through the overlap.

Another possible threat on the user's trajectory privacy could
arise if $R_{w+1}$ overlaps with $GCR(cl,k)$ and $R_w$. A user
does not need to send the next request with $R_{w+1}$ as long as
the user is in $GCR(cl_r,k_r)$ which in turn means the user's
location must not be within $GCR(cl_r,k_r) \cap R_{w+1}$ at the
time of sending $R_{w+1}$ to the LSP. Since the LSP does not know
$GCR(cl_r,k_r)$, the LSP cannot identify the overlap of
$GCR(cl_r,k_r)$ with $R_{w+1}$ and determine more precise location
of the user as $R_{w+1} \setminus (GCR(cl_r,k_r) \cap R_{w+1})$.
However consider the case when $R_{w+1}$ overlaps with $GCR(cl,k)$
and $R_w$: since $GCR(cl,k), R_w \subset GCR(cl_r,k_r)$ and the
LSP knows $GCR(cl,k)$ and $R_w$, the LSP can refine more precise
location of the user at the time of sending $R_{w+1}$ as $R_{w+1}
\setminus (GCR(cl,k) \cap R_{w+1})$. To overcome the above
mentioned privacy threat, we use two variables $\delta_{b}$ and
$\delta$:

\begin{itemize}
\item \emph{Boundary distance} $\delta_{b}$: the minimum distance of user's
current position $q$ from the boundary of $C(o,r)$.
\item \emph{Safe distance} $\delta$: the user specified distance, which is used to determine when the next request needs to be sent.
\end{itemize}

In our technique, the user's next request is sent to the LSP as
soon as $\delta_b$ becomes less or equal to $\delta$. Using
$\delta$, whose value is unknown to the LSP, there is no possible
privacy attack from the overlap of $R_{w+1}$ with $GCR(cl,k)$ and
$R_w$ as the user might need to send $R_{w+1}$ in advance due to
the constraint of $\delta_b \leq \delta$.
Figure~\ref{fig:csaferegion}(d) shows a case where the user's
location at the time of requesting $R_{w+1}$ is within $GCR(0.9,1)
\cap R_{w+1}$ to satisfy $\delta_b \leq \delta$.

In the second option of achieving the extended region of GCR
without informing the LSP, a user specifies a higher value than
the required number of NNs, i.e., $k>k_r$. From the construction
method of a $GCR$, we know that $GCR(cl,k+1)\subset GCR(cl,k)$ for
a fixed $C(o,r)$, which leads to the following lemma.

\begin{lemma}
\label{lemma:6_two} Let $cl=cl_r$ and $k>k_r$. Then
$GCR(cl_r,k_r)\supset GCR(cl,k)$ for a fixed $C(o,r)$.
\end{lemma}

Since we also have $GCR(cl_r,k_r)\supset GCR(cl,k)$ for the second
option, similar to the case of first option, a user can protect
her trajectory privacy using the extended region, which is used
when the user cannot sacrifice the accuracy of answers.

In the third option, a user requests higher values for both
confidence level and the number of NNs than required and can
obtain a larger extension for the $GCR(cl_r,k_r)$ as both $cl$ and
$k$ contribute to extend the region. The larger extension ensures
a user with a higher level of trajectory privacy because
$GCR(cl_r,k_r)$ covers a longer part of the user's trajectory,
which in turn reduces the number of times the user needs to send
the obfuscation rectangle. The level of trajectory privacy also
increases with the increase of the difference between $cl$ and
$cl_r$ or $k$ and $k_r$ because with the decrease of $cl_r$ or
$k_r$, the size of $GCR(cl_r,k_r)$ increases for a fixed $C(o,r)$
and with the increase of $cl$ or $k$, $C(o,r)$ becomes larger,
which results in a larger $GCR(cl_r,k_r)$. Thus, the difference
between $cl$ and $cl_r$ or $k$ and $k_r$ can be increased by
either incurring a higher query processing overhead (i.e.,
specifying a higher value for $cl$ or $k$) or sacrificing the
required quality of the answers (i.e., specifying a lower value
for $cl_r$ or $k_r$). Note that, a large value for $cl$ or $k$
incurs higher query processing overhead as more data objects need
to be retrieved.

The parameters $cl$, $cl_r$, $k$, $k_r$, $\delta$, and the size of
the obfuscation rectangle can be changed according to the user's
privacy profile and quality of service requirements. A user can
specify a high level of privacy requirement in her profile for
some locations that are more sensitive to her. Different values
for $cl$, $cl_r$, $k$, $k_r$, and $\delta$ in consecutive requests
prevent an LSP from gradually learning or guessing any bound of $cl_r$ and $k_r$ to
apply reverse engineering and predict a more precise user location
within the obfuscation rectangle.

Based on the above discussion of our technique, we present the
algorithm that protects the user's trajectory privacy while
processing an M$k$NN query. Before going to the details of the
algorithm, we summarize commonly used symbols in
Table~\ref{table:symb}.

\begin{table}[htbp]
\begin{center}
\begin{tabular}{|c|c|}
  \hline
  Symbol& Meaning\\
  \hline
  $R_w$ & Obfuscation Rectangle \\
  \hline
  $cl_r$ & Required confidence level \\
  \hline
   $cl$ & Specified confidence level \\
  \hline
  $k_r$ & Required number of NNs \\
  \hline
   $k$ & Specified number of NNs \\
  \hline
  $C(o,r)$ & Known region \\
  \hline
  $GCR(.,.)$ & Guaranteed combined region \\
  \hline
 $\delta$ & Safe distance \\
  \hline
  $\delta_b$ & Boundary distance \\
  \hline
\end{tabular}
\caption{Symbols} \label{table:symb}
\end{center}
\end{table}

\subsection{Algorithm}
Algorithm~\ref{algo:RequestkNN}, \textsc{Request\_PM$k$NN}, shows
the steps for requesting a PM$k$NN query. A user initiates an
M$k$NN query by generating an obfuscation rectangle $R_w$ that
includes her current location $q$. The parameters $cl$, $cl_r$,
$k$, $k_r$, and $\delta$ are set according to the user's
requirement. Then a request is sent with $R_w$ to the LSP for $k$
NNs with a confidence level $cl$. The LSP returns the set of data
objects $P$ that includes the $k$ NNs for every point of $R_w$
with a confidence level at least $cl$. Then according to
Lemma~\ref{lemma:6_zero}, the user continues to have the $k_{r}$
NNs with a confidence level at least $cl_{r}$ as long as the user
resides within $GCR(cl_r,k_r)$. In this paper, we do not focus on
developing algorithms to maintain the rank of $k_{r}$ NNs from $P$
for every position of the user's trajectory, because this is
orthogonal to our problem of protecting privacy of users'
trajectories for an M$k$NN query. For this purpose, any of the
existing approaches (e.g.,~\cite{kulik06.GIS}) can be used.

%

\setlength{\algomargin}{2em} \dontprintsemicolon
\begin{algorithm}[htbp]
\label{algo:RequestkNN} \caption{\textsc{Request\_PM$k$NN}}
    $w\leftarrow 1$\;
    $cl, cl_r\leftarrow$ user specified and required confidence level\;
    $k, k_r\leftarrow$ user specified and required number of NNs\;
    $\delta \leftarrow$ user specified safe distance\;
    $R_w\leftarrow GenerateRectangle(q)$\;
    $P\leftarrow RequestkNN(R_w,cl,k)$\;
    \While{service required}{
        $q\leftarrow NextLocationUpdate()$\;
        ${p_h}_k\leftarrow {k_r}^{th}$ NN from $q$\;
        $cl, cl_r\leftarrow$ user specified and required confidence level\;
        $k, k_r\leftarrow$ user specified and required number of NNs\;
        $\delta \leftarrow$ user specified safe distance\;
        $\delta_{b}\leftarrow r - dist(o,q)$\;
        \If{$(r \leq cl_r \times dist({p_h}_k,q) + dist(o,q))$ or $(\delta_{b} \leq \delta)$}
        {
            $R_{w+1}\leftarrow GenerateRectangle(q, C(o,r))$\;
            $P\leftarrow RequestkNN(R_{w+1},cl,k)$\;
              $w\leftarrow w+1$\;
        }
    }
\end{algorithm}

For every location update, the algorithm checks two conditions:
whether the user's current position $q$ is outside her current
$GCR(cl_r,k_r)$ or the minimum boundary distance from $C(o,r)$,
$\delta_b$, has become less or equal to the user specified
distance, $\delta$. To check whether the user is outside her
$GCR(cl_r,k_r)$, the algorithm checks the constraint $r \leq cl_r
\times dist({p_h}_k,q) + dist(o,q)$, where $r$ is the radius of
current known region and $cl_r \times dist({p_h}_k,q) + dist(o,q)$
represents the required radius of the known region to have $k_r$
NNs with confidence level at least $cl_r$ from the current
position $q$. For the second condition, $\delta_b$ is computed by
subtracting $dist(o,q)$ from $r$ (Line 1.13). If any of the two
conditions in Line 1.14 becomes true, then the new obfuscation
rectangle $R_{w+1}$ is computed with the restriction that it must
be included within the current $C(o,r)$. After computing
$R_{w+1}$, the next request is sent and $k$ NNs are retrieved for
$R_{w+1}$ with a confidence level at least $cl$. The process
continues as long as the service is required.



The function $GenerateRectangle$ is used to compute an obfuscation
rectangle for a user according to her privacy requirement. We
assume that a user can compute her rectangle based on any existing
obfuscation or $l$-diversity
techniques~\cite{Damiani09.SPRINGL,Xue09.LOCA,yiu08.ICDE} and
therefore a detailed discussion for the function
$GenerateRectangle$ goes beyond the scope of this paper.

The following theorem shows the correctness of the algorithm
\textsc{Request\_PM$k$NN}.
\begin{theorem}
\label{th:0} The algorithm \textsc{Request\_PM$k$NN} protects a
user's trajectory privacy for M$k$NN queries.
\end{theorem}
\begin{proof}

The obfuscation rectangles $R_{w+1}$ for a user requesting a
PM$k$NN query always overlaps with $GCR(cl_r,k_r)$ and sometimes
also overlaps with $GCR(cl,k)$ and $R_w$. We will show that these
overlaps do not reveal a more precise user location to the LSP,
i.e., the user's trajectory privacy is protected.

The LSP does not know about the boundary of $GCR(cl_r,k_r)$, which
means that the LSP cannot compute $GCR(cl_r,k_r) \cap R_{w+1}$.
Thus, the LSP cannot refine a user's location at the time of
requesting $R_{w+1}$ or the user's trajectory path from
$GCR(cl_r,k_r) \cap R_{w+1}$.

Since the LSP knows $GCR(cl,k)$ and $R_w$, it can compute the
overlaps, $GCR(cl,k) \cap R_{w+1}$ and $R_w \cap R_{w+1}$, when it
receives $R_{w+1}$. However, the availability of $GCR(cl_r,k_r)$
to the user and the option of having different values for $\delta$
prevent the LSP to determine whether the user is located within
$GCR(cl,k) \cap R_{w+1}$ and $R_w \cap R_{w+1}$ at the time of
requesting $R_{w+1}$ or whether the user's trajectory passes
through these overlaps.

In summary there is no additional information to render a more
precise user position or user trajectory within the rectangle.
Thus, every obfuscation rectangle computed using the algorithm
\textsc{Request\_PM$k$NN} satisfies the two required conditions
(see Definition~\ref{def:2_3}) for protecting a user's trajectory
privacy.
\end{proof}

\subsubsection{The maximum movement bound attack}
As we have discussed in Section~\ref{sec:pblm_stup}, if a user's
maximum velocity is known to the LSP then the maximum movement
bounding attack can identify a user's more precise position. To
prevent the maximum movement bound attack, existing
solutions~\cite{cheng06.PET,ghinita09.GIS,Hu09.TKDE,Xu09.TPDS}
have proposed that $R_{w+1}$ for the next request needs to be
computed in a way so that $R_{w+1}$ is completely included within
the maximum movement bound of $R_w$, denoted as $M_w$. Our
proposed algorithm to generate $R_{w+1}$ can also consider this
constraint of $M_w$ whenever the LSP knows the user's maximum
velocity. Incorporating the constraint of $M_{w}$ in our algorithm
does not cause any new privacy violation for users.

Note that, Algorithm~\ref{algo:RequestkNN} to protect a user's
trajectory privacy for an M$k$NN query with obfuscation rectangles
can be also generalized for the case where a user uses other
geometric shapes (e.g., a circle) to represent the imprecise
locations if the known region for other geometric shapes is also a
circle. For example, if a user uses obfuscation circles instead of
obfuscation rectangles then the overlapping rectangle attack turns
into overlapping circle attack. From
Algorithm~\ref{algo:RequestkNN}, we observe that our technique to
protect overlapping rectangle attack is independent of any
parameter of obfuscation rectangle; it only depends on the center
and radius of the known region. Thus, as long as the
representation of the known region is a circle, our technique can
be also applied for an overlapping circle attack.
%

\section{Server-side Processing}
\label{sec:knnq}

For a PM$k$NN query with a customizable confidence level, an LSP
needs to provide the $k$ NNs with the specified confidence level
for all points of every requested obfuscation rectangle.
Evaluating the $k$ NNs with a specified confidence level for every
point of an obfuscation rectangle separately is an expensive
operation and doing it continuously for a PM$k$NN query incurs
large overhead. In this section, we develop an efficient algorithm
that finds the $k$ NNs for every point of an obfuscation rectangle
with a specified confidence level in a single search using an
$R$-tree. Our proposed algorithm allows an LSP to provide the user
with a known region, which helps protecting the user's trajectory
privacy and further to reduce the overall PM$k$NN query processing
overhead.


We show different properties of a confidence level for an
obfuscation rectangle, which we use to improve the efficiency of
our algorithms. Let $R_w$ be a user's obfuscation rectangle with
center $o$ and corner points $\{c_1, c_2, c_3, c_4\}$, and
$m_{ij}$ be the middle point of $\overline{c_ic_j}$, where $(i,j)
\in \{(1,2),(2,3),(3,4),(4,1)\}$. To avoid the computation of the
confidence level for a data object with respect to every point of
$R_w$ while searching for the query answers, we exploit the
following properties of the confidence level. We show that if two
endpoints, i.e., a corner point and its adjacent middle point or
the center and a point in the border of $R_w$, of a line have a
confidence level at least $cl$ for a data object then every point
of the line has a confidence level at least $cl$ for that data
object. Formally, we have the following theorems.

%
%
%



%

\begin{theorem}
\label{th:1} Let $c_i, c_j$ be any two adjacent corner points of
an obfuscation rectangle $R_w$ and $m_{ij}$ be the middle point of
$\overline{c_ic_j}$. For $t\in \{i,j\}$, if $c_t$ and $m_{ij}$
have a confidence level at least $cl$ for a data object $p_{h}$
then all points in $\overline{m_{ij}c_t}$ have a confidence level
at least $cl$ for $p_{h}$.
\end{theorem}

\begin{theorem}
\label{th:2} Let $o$ be the center of an obfuscation rectangle
$R_w$, $c_i, c_j$ be any two adjacent corner points of $R_w$, and
$c$ be a point in $\overline{c_ic_j}$. If $o$ and $c$ have a
confidence level at least $cl$ for a data object $p_{h}$ then all
points in $\overline{oc}$ have a confidence level at least $cl$
for $p_{h}$.
\end{theorem}

Next we discuss the proof of Theorem~\ref{th:1}. We omit the proof
of Theorem~\ref{th:2}, since a similar proof technique used for
Theorem~\ref{th:1} can be applied for Theorem~\ref{th:2} by
considering $o$ as $m_{ij}$ and $c$ as $c_t$.

As mentioned in Section~\ref{sec:conflevel}, our algorithm to
evaluate $k$NN answers expands the known region $C(o,r)$ until the
$k$ NNs with the specified confidence level for every point of
$R_w$ are found. Since any point outside $C(o,r)$ has a confidence
level 0 (see Definition~\ref{def:1}), $C(o,r)$ needs to be at
least expanded until $R$ is within $C(o,r)$ to ensure $k$NN
answers with a specified confidence level greater than 0. Hence,
we assume that $R \subset C(o,r)$ at the current state of the
search. Let the extended lines
$\overrightarrow{om_{ij}}$\footnote{We use the symbol
$\rightarrow$ for directional line segments.} and
$\overrightarrow{oc_t}$ intersect the border of $C(o,r)$ at
$m_{ij}^\prime$ and $c_t^\prime$, respectively, where $t \in
\{i,j\}$. Figure~\ref{fig:basicfig}(a) shows an example for $i=1$,
$j=2$, and $t=j$. For a data object $p_h$ in $C(o,r)$, the
confidence levels of $c_t$ and $m_{ij}$, $CL(c_t,p_h)$ and
$CL(m_{ij},p_h)$, can be expressed as
$\frac{dist(c_t,{c_t}^\prime)}{dist(c_t,p_h)}$ and
$\frac{dist(m_{ij},{m_{ij}}^\prime)}{dist(m_{ij},p_h)}$,
respectively.

\begin{figure*}[htbp]
  \begin{center}
    \begin{tabular}{ccc}
      \resizebox{40mm}{!}{\includegraphics{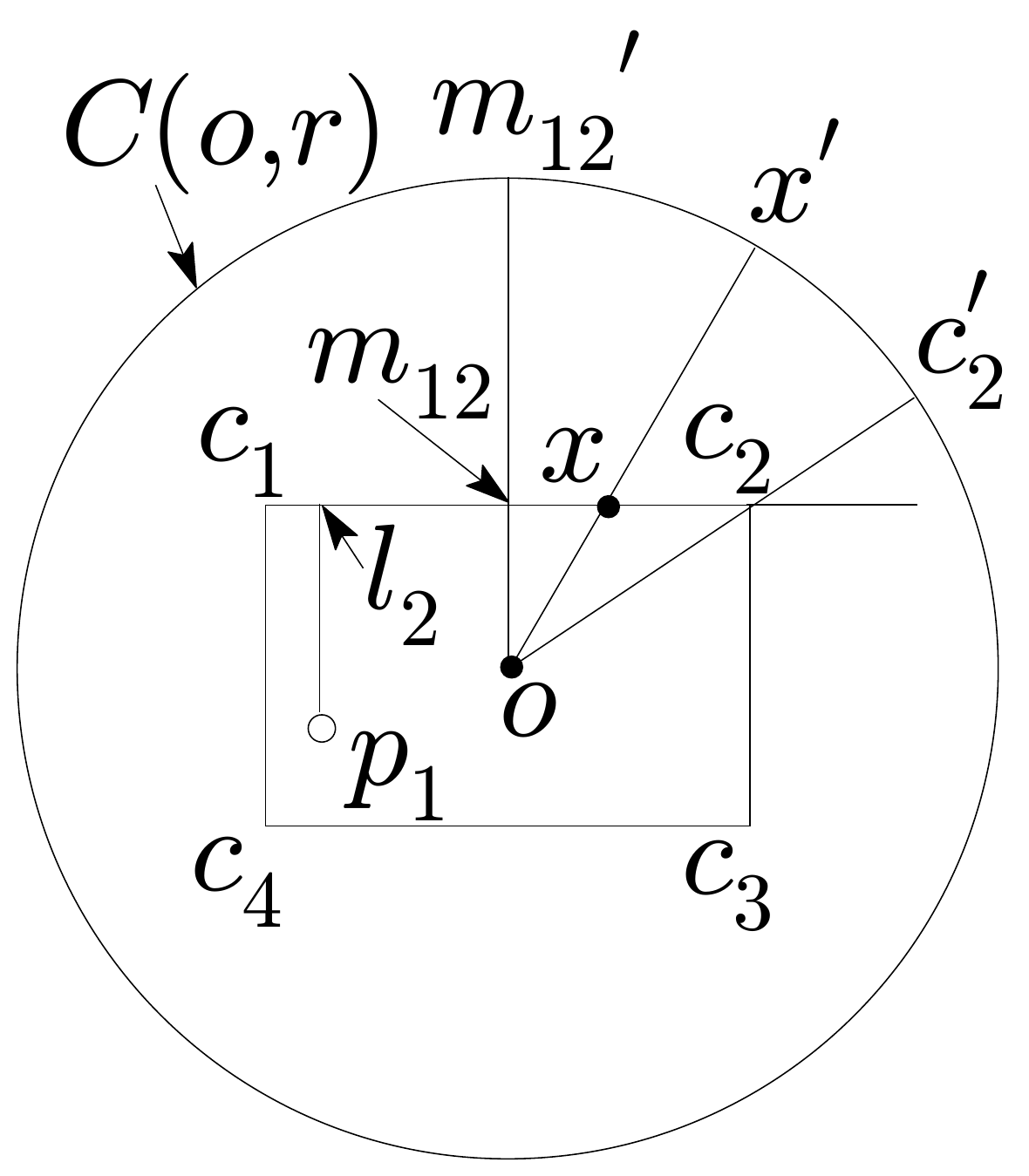}} &
      \resizebox{40mm}{!}{\includegraphics{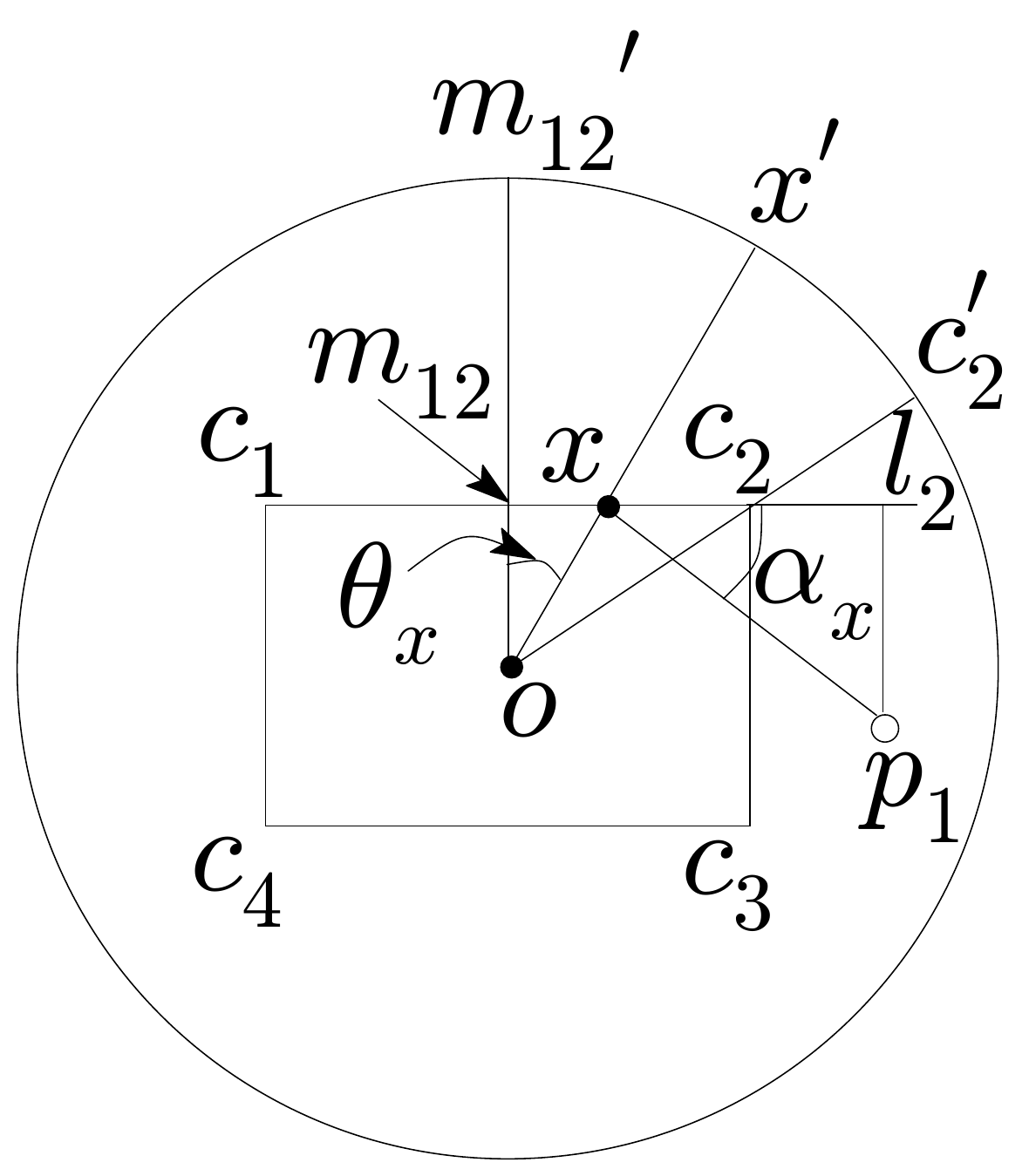}} &
      \resizebox{40mm}{!}{\includegraphics{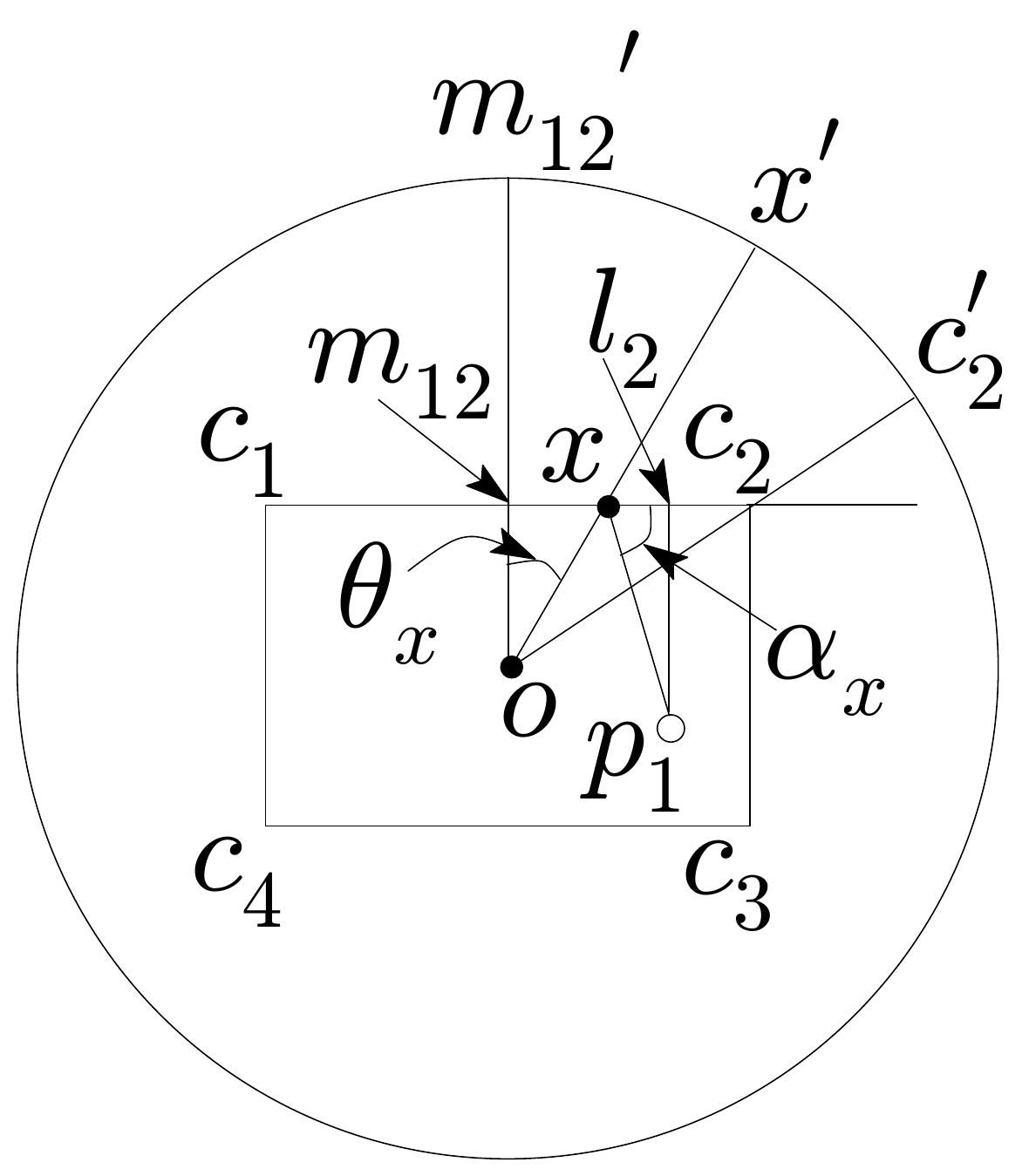}} \\
       \scriptsize{(a)\hspace{0mm}} & \scriptsize{(b)\hspace{-0mm}}& \scriptsize{(c)}
        \end{tabular}
    \caption{Impact of different positions of a data object}
    \label{fig:basicfig}
  \end{center}
\end{figure*}

Let $x$ be a point in $\overline{m_{ij}c_t}$, and
$\overrightarrow{ox}$ intersect the border of $C(o,r)$ at
$x^\prime$. For a data object $p_h$ in $C(o,r)$, the confidence
level of $x$, $CL(x,p_h)$, is measured as
$\frac{dist(x,x^\prime)}{dist(x,p_h)}$. As $x$ moves from $c_t$
towards $m_{ij}$, although $dist(x,x^\prime)$ always increases,
$dist(x,p_h)$ can increase or decrease (does not maintain a single
trend) since it depends on the position of $p_h$ within $C(o,r)$.
Without loss of generality we consider an example in
Figure~\ref{fig:basicfig}, where $p_1$ is a data object within
$C(o,r)$. Based on the position of $p_1$ with respect to $m_{12}$
and $c_2$, we have three cases: the perpendicular from $p_{1}$
intersects the extended line $\overrightarrow{c_2m_{12}}$ (see
Figure~\ref{fig:basicfig}(a)) or the extended line
$\overrightarrow{m_{12}c_2}$ (see Figure~\ref{fig:basicfig}(b)) or
the segment $\overline{m_{12}c_2}$ (see
Figure~\ref{fig:basicfig}(c)) at $l_2$. In the first case,
$dist(x,p_1)$ decreases as $x$ moves from $c_2$ towards $m_{12}$
as shown in Figure~\ref{fig:basicfig}(a). In the second case,
$dist(x,p_1)$ decreases as $x$ moves from $m_{12}$ towards $c_2$
as shown in Figure~\ref{fig:basicfig}(b). In the third case,
$dist(x,p_h)$ is the minimum at $x=l_2$, i.e., $dist(x,p_1)$
decreases as $x$ moves from $c_2$ or $m_{12}$ towards $l_2$ as
shown in Figure~\ref{fig:basicfig}(c). From these three cases we
observe that for different positions of $p_h$, $dist(x,p_h)$ can
decrease for moving $x$ in both directions, i.e., from $c_t$
towards $m_{ij}$ or from $m_{ij}$ towards $c_t$.

For the scenario, where $dist(x,p_h)$ decreases as $x$ moves from
$c_t$ towards $m_{ij}$ (first case) or from $c_t$ towards $l_t$
(third case), i.e., $dist(x,p_h) \leq dist(c_t,p_h)$, we have the
following lemma.
\begin{lemma}
\label{lemma:one} If $dist(x,p_h) \leq dist(c_t,p_h)$ and
$CL(c_t,p_h) \geq cl$ then $CL(x,p_h) \geq cl$, for any point $x
\in \overline{m_{ij}c_t}$.
\end{lemma}

The proof of this lemma directly follows from $dist(x,x^\prime)
\geq dist(c_t,{c_t}^\prime)$.

In the other scenario, $dist(x,p_h)$ decreases as $x$ moves from
$m_{ij}$ towards $c_t$ (second case) or from $m_{ij}$ towards
$l_t$ (third case). In the general case, let $u_t$ be a point that
represents $c_t$ for the second case and $l_t$ for the third case.
To prove that $CL(x,p_h) \geq cl$, in contrast to
Lemma~\ref{lemma:one} where we only need to have $CL(c_t,p_h)\geq
cl$, for the current scenario we need to have the confidence level
at least equal to $cl$ for $p_h$ at both end points, i.e.,
$m_{ij}$ and $u_t$. According to the given conditions of
Theorem~\ref{th:1}, we already have $CL(m_{ij},p_h)\geq cl$ and
$CL(c_t,p_h)\geq cl$. Since $u_t$ is $c_t$ in the second case and
$l_t$ in the third case, we need to compute the confidence level
of $l_t$ for $p_h$ in the third case and using
Lemma~\ref{lemma:one} we find that $CL(l_t,p_h) \geq cl$. Thus, we
have the confidence level of both $m_{ij}$ and $u_t$ for $p_h$ at
least equal to $cl$.

However, showing $CL(x,p_h) \geq cl$ if both $m_{ij}$ and $u_t$
have a confidence level of at least $cl$ for $p_h$ is not
straightforward, because in the current scenario both
$dist(x,p_h)$ and $dist(x,x^\prime)$ decrease with the increase of
$dist(m_{ij},x)$. Thus, we need to compare the rate of decrease
for $dist(x,x^\prime)$ and $dist(x,p_h)$ as $x$ moves from
$m_{ij}$ to $u_t$. Assume that $\angle{xom_{ij}}=\theta_x$ and
$\angle{p_hxl_t}=\alpha_x$. The range of $\theta_x$ can vary from
0 to $\theta$, where $\theta_{m_{ij}}=0$, $\theta_{u_t}=\theta$,
and $\theta \leq \frac{\Pi}{4}$. For a fixed range of $\theta_x$
the range of $\alpha_x$, $[\alpha_{m_{ij}},\alpha_{u_t}]$, can
have any range from $[0,\frac{\Pi}{2}]$ depending upon the
position of $p_{h}$. We express $dist(x,x^\prime)$ and
$dist(x,p_{h})$ as follows:
\begin{displaymath}
\label{eq:one}
    dist(x,x^\prime)=r-dist(o,m_{ij}) \times \sec\theta_x\\
    \end{displaymath}
\begin{displaymath}
{dist(x,p_{h}) =
        \left\{ \begin{array}{ll}
            dist(p_{h}, l_t) \times \csc\alpha_x & \mbox{if $\alpha_x \neq 0$}\\
            dist(m_{ij},p_{h})-dist(m_{ij},x) & \mbox{otherwise}.
        \end{array} \right.}
\end{displaymath}

\begin{figure*}[htbp]
  \begin{center}
    \begin{tabular}{ccc}
      \resizebox{40mm}{!}{\includegraphics{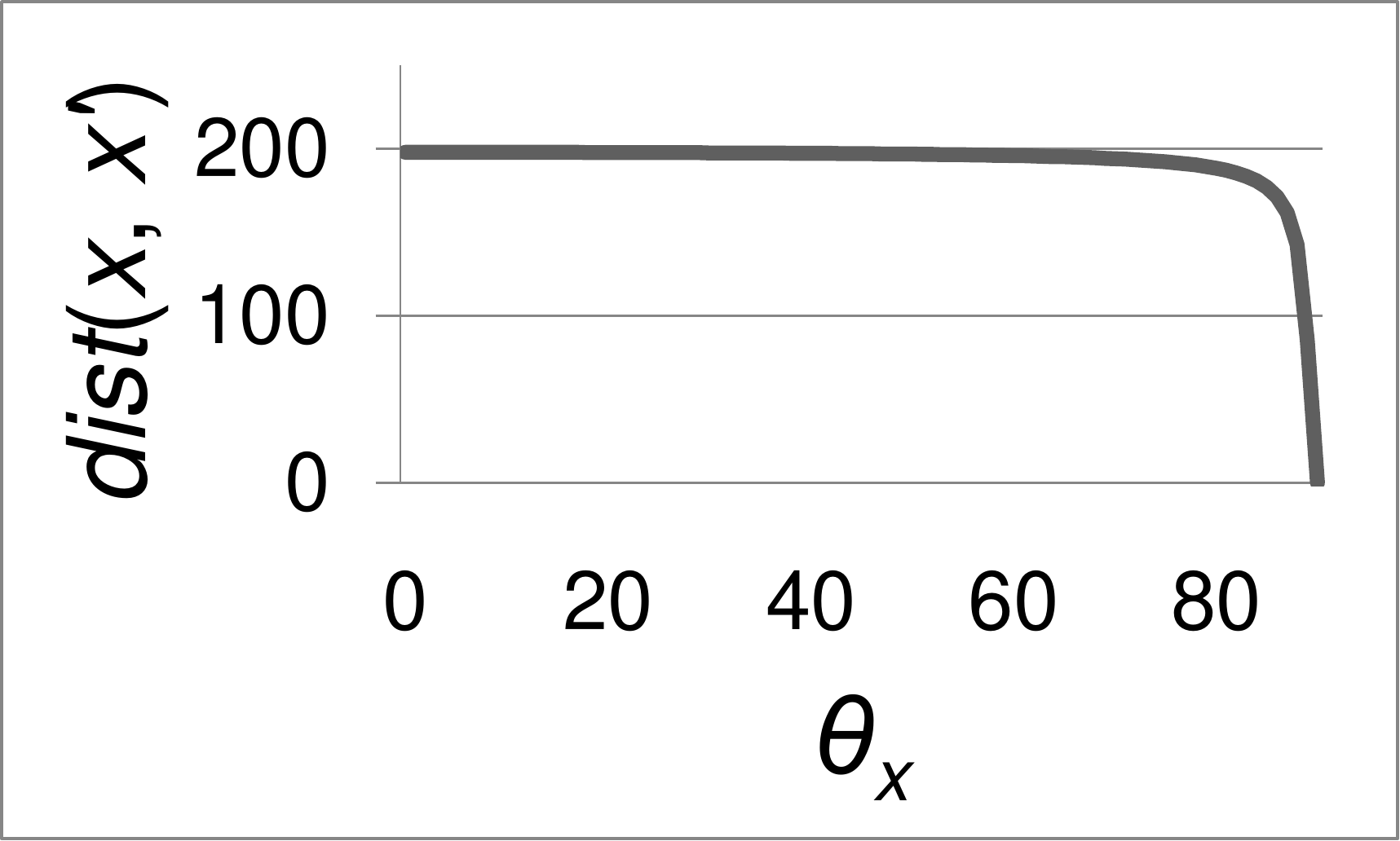}} &
      \resizebox{40mm}{!}{\includegraphics{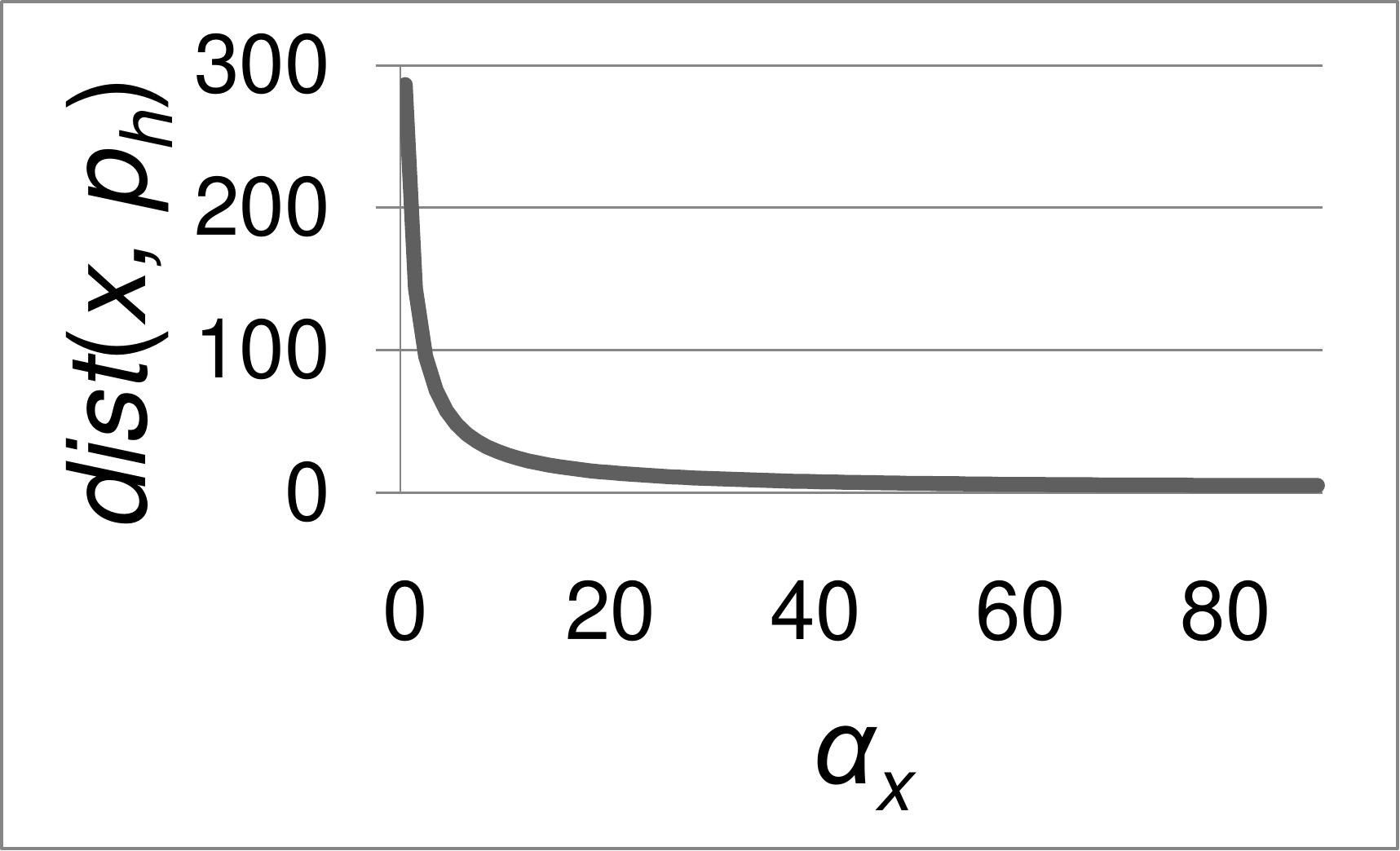}} &
      \resizebox{40mm}{!}{\includegraphics{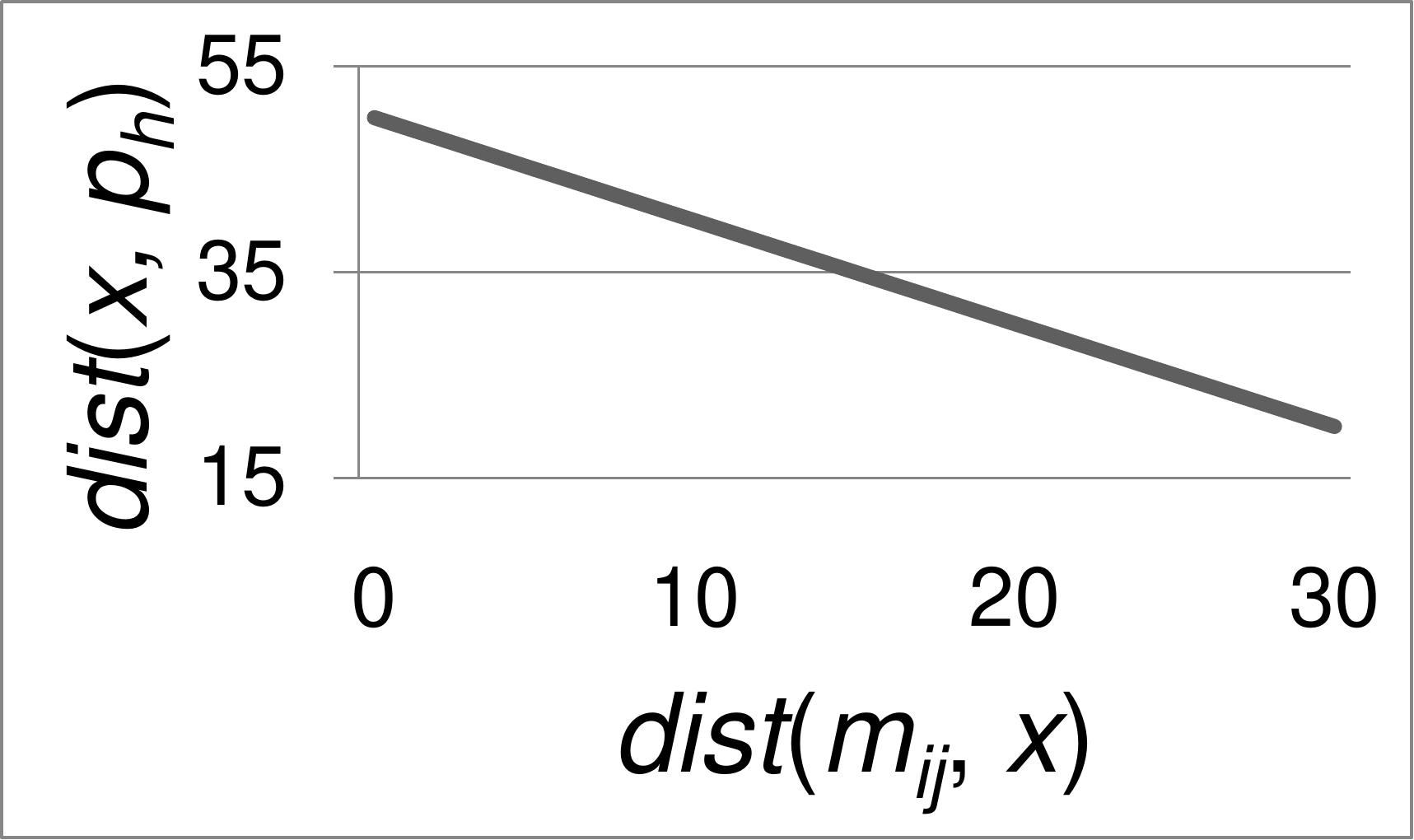}} \\
       \scriptsize{(a)\hspace{0mm}} & \scriptsize{(b)}& \scriptsize{(c)}
        \end{tabular}
    \caption{Curve sketching}
    \label{fig:curve}
  \end{center}
\end{figure*}

The rate of decrease for $dist(x,x^\prime)$ and $dist(x,p_h)$ are
not comparable by computing their first order derivative as they
are expressed with different variables and there is no fixed
relation between the range of $\theta_x$ and $\alpha_x$.
Therefore, we perform a curve sketching and consider the second
order derivative in Figure~\ref{fig:curve}. From the second order
derivative, we observe in Figure~\ref{fig:curve}(a) that the rate
of decreasing rate of $dist(x,x^\prime)$ increases with the
increase of $\theta_x$, whereas in Figure~\ref{fig:curve}(b) the
rate of decreasing rate of $dist(x,p_{h})$ decreases with the
increase of $\alpha_x$ for $\alpha_x \neq 0$ and in
Figure~\ref{fig:curve}(c) the rate of decreasing rate remains
constant with the increase of $dist(m_{ij},x)$ for $\alpha_x = 0$.
The different trends of the decreasing rate and the constraint of
confidence levels at two end points $m_{ij}$ and $u_t$ allow us to
make a qualitative comparison between the rate of decrease for
$dist(x,x^\prime)$ and $dist(x,p_{h})$ with respect to the common
metric $dist(m_{ij},x)$, as $dist(m_{ij},x)$ increases with the
increase of both $\theta_x$ and $\alpha_x$ for a fixed $p_h$. We
have the following lemma.


\begin{lemma}
\label{lemma:two} Let $dist(x,p_h)$ decrease as $x$ moves from
$m_{ij}$ to $u_t$  for any point $x \in \overline{m_{ij}u_t}$. If
$CL(m_{ij},p_h) \geq cl$ and $CL(u_t,p_h) \geq cl$, then
$CL(x,p_h) \geq cl$.
\end{lemma}
\begin{proof}
(By contradiction) Assume to the contrary that there is a point $x
\in \overline{m_{ij}u_t}$ such that $CL(x,p_h) < cl$, i.e.,
$\frac{dist(x,x^\prime)}{dist(x,p_{h})} < cl$. Then we have the
following relations.
\begin{equation}
\label{eq:three}
\frac{dist(m_{ij},m_{ij}^\prime)-dist(x,x^\prime)}{dist(m_{ij},x)}>
\frac{dist(m_{ij},p_{h})-dist(x,p_{h})}{dist(m_{ij},x)}
\end{equation}
\begin{equation}
\label{eq:four}
\frac{dist(x,x^\prime)-dist(u_t,u_t^\prime)}{dist(x,u_t)}<
\frac{dist(x,p_{h})-dist(u_t,p_{h})}{dist(x,u_t)}
\end{equation}
Since we know that for $dist(x,x^\prime)$, the rate of decreasing
rate increases with the increase of $dist(m_{ij},x)$ and for
$dist(x,p_{h})$, the rate of decreasing rate decreases or remains
constant with the increase of $dist(m_{ij},x)$, we have the
following relations.
\begin{equation}
\label{eq:five}
\frac{dist(m_{ij},m_{ij}^\prime)-dist(x,x^\prime)}{dist(m_{ij},x)}<\frac{dist(x,x^\prime)-dist(u_t,u_t^\prime)}{dist(x,u_t)}
\end{equation}
\begin{equation}
\label{eq:six}
\frac{dist(m_{ij},p_{h})-dist(x,p_{h})}{dist(m_{ij},x)} \geq
\frac{dist(x,p_{h})-dist(u_t,p_{h})}{dist(x,u_t)}
\end{equation}
>From Equations~\ref{eq:three},~\ref{eq:four}, and~\ref{eq:five} we
have,
\begin{displaymath}
\frac{dist(m_{ij},p_{h})-dist(x,p_{h})}{dist(m_{ij},x)}<\frac{dist(x,p_{h})-dist(u_t,p_{h})}{dist(x,u_t)}
\end{displaymath}
which contradicts Equation~\ref{eq:six}, i.e., our assumption.
\end{proof}

Finally, from Lemmas~\ref{lemma:one} and~\ref{lemma:two}, we can
conclude that if $CL(c_t,p_h) \geq cl$ and $CL(m_{ij},p_h) \geq
cl$, then $CL(x,p_h) \geq cl$ for any point $x \in
\overline{m_{ij}c_t}$, which proves Theorem~\ref{th:1}.

\subsection{Algorithms}\label{sec:server_algo}

We develop an efficient algorithm, \emph{\textsc{Clappinq}}
(Confidence Level Aware Privacy Protection In Nearest Neighbor
Queries), that finds the $k$ NNs for an obfuscation rectangle with
a specified confidence level. Algorithm~\ref{algo:Rect-kNN} gives
the pseudo code for \textsc{Clappinq} using an $R$-tree. The input
to Algorithm~\ref{algo:Rect-kNN} are an obfuscation rectangle
$R_w$, a confidence level $cl$, and the number of NNs $k$ and the
output is the candidate answer set $P$ that includes the $k$ NNs
with a confidence level at least $cl$ for every point of $R_w$.
\dontprintsemicolon
\begin{algorithm}[htbp]
\label{algo:Rect-kNN} \caption{\textsc{Clappinq}($R,cl,k$)}
    $P\leftarrow\emptyset$\;
    $status \leftarrow$ 0\;
    $Enqueue(Q_p,root,0)$\;
    \While{$Q_p$ is not empty and $status \geq 0$}
    {
        $p \leftarrow Dequeue(Q_p)$\;
        $r \leftarrow MinDist(o,p)$\;
        \If{$status>0$ and $status<r$}
            {
                $status\leftarrow -1$\;
            }

        \uIf {$p$ is a data object}
        {
            $P\leftarrow P \cup p$\;
            \If{$status=0$}
            {
                $status \leftarrow UpdateStatus(R,cl,k,P,r)$\;
            }
        }
        \Else
        {

                \For{each child node $p_c$ of $p$}
                {
                    $d_{min}(p_c) \leftarrow MinDist(o,p_c)$\;
                    $Enqueue(Q_p,p_c,d_{min}(p_c))$\;
                }

        }
    }
    \Return $P$;
\end{algorithm}

As mentioned in Section~\ref{sec:conflevel}, the basic idea of our
algorithm is to start a best first search (BFS) considering the
center $o$ of the given obfuscation rectangle $R_w$ as the query
point and continue the search until the $k$ NNs with a confidence
level of at least $cl$ are found for all points of $R_w$. The
known region $C(o,r)$ is the search region covered by BFS and $P$
is the set of data objects located within $C(o,r)$. $Q_p$ is a
priority queue used to maintain the ordered data objects and
$R$-tree nodes based on the minimum distance between the query
point $o$ and the data objects/MBRs of $R$-tree nodes (by using
the function $MinDist$). Since the size of the candidate answer
set is unknown, we use $status$ to control the execution of the
BFS. Based on the values of $status$, the BFS can have three
states: (i) when $status=0$, each time the BFS discovers the next
nearest data object, it checks whether $status$ needs to be
updated, (ii) when $status>0$, the BFS executes until the radius
of the known region becomes greater than the value of $status$,
and (iii) when $status=-1$, the BFS terminates. Initially,
$status$ is set to 0. Each time a data object/$R$-tree node $p$ is
dequeued from $Q_p$ the current radius $r$ is updated. When $p$
represents a data object, then $p$ is added to the current
candidate set $P$ and the procedure $UpdateStatus$ is called if
$status$ equals $0$.

The pseudo code for $UpdateStatus$ is shown in
Algorithm~\ref{algo:UpdateStatus}. The notations used for this
algorithm are summarized below.
\begin{enumerate}
    \item $count(c_t,cl,P)$: the number of data objects in $P$ for which a corner point
    $c_t$ of $R_w$
has a confidence level at least $cl$.
    \item $d_i^k$ $(d_j^k)$: the $k^{th}$ minimum distance from a middle point
    $m_{ij}$ of $R_w$
    to the data objects in $P_i$ $(P_j)$, where $P_i$ $(P_j) \subset
    P$ and $P_i$ $(P_j)$ is the set of data objects with
    respect to $c_i$ $(c_j)$ with a confidence level of at least
    $cl$.
    \item $d_{max}$: the maximum of all $d_{max}(m_{ij})$,
    where each $d_{max}(m_{ij})$ is the maximum of $d_i^k$ and $d_j^k$
    (see Figure~\ref{fig:algo2b}(a)).
    \item $d_{safe}$: the minimum distance of all $d_{safe}(m_{ij})$,
where $d_{safe}(m_{ij})$ represents the radius of the maximum
circular region within $C(o,r)$ centered at $m_{ij}$ (see
    Figure~\ref{fig:algo2b}(b)).
\end{enumerate}
\begin{figure}[htbp]
    \centering
        \includegraphics[width=4in]{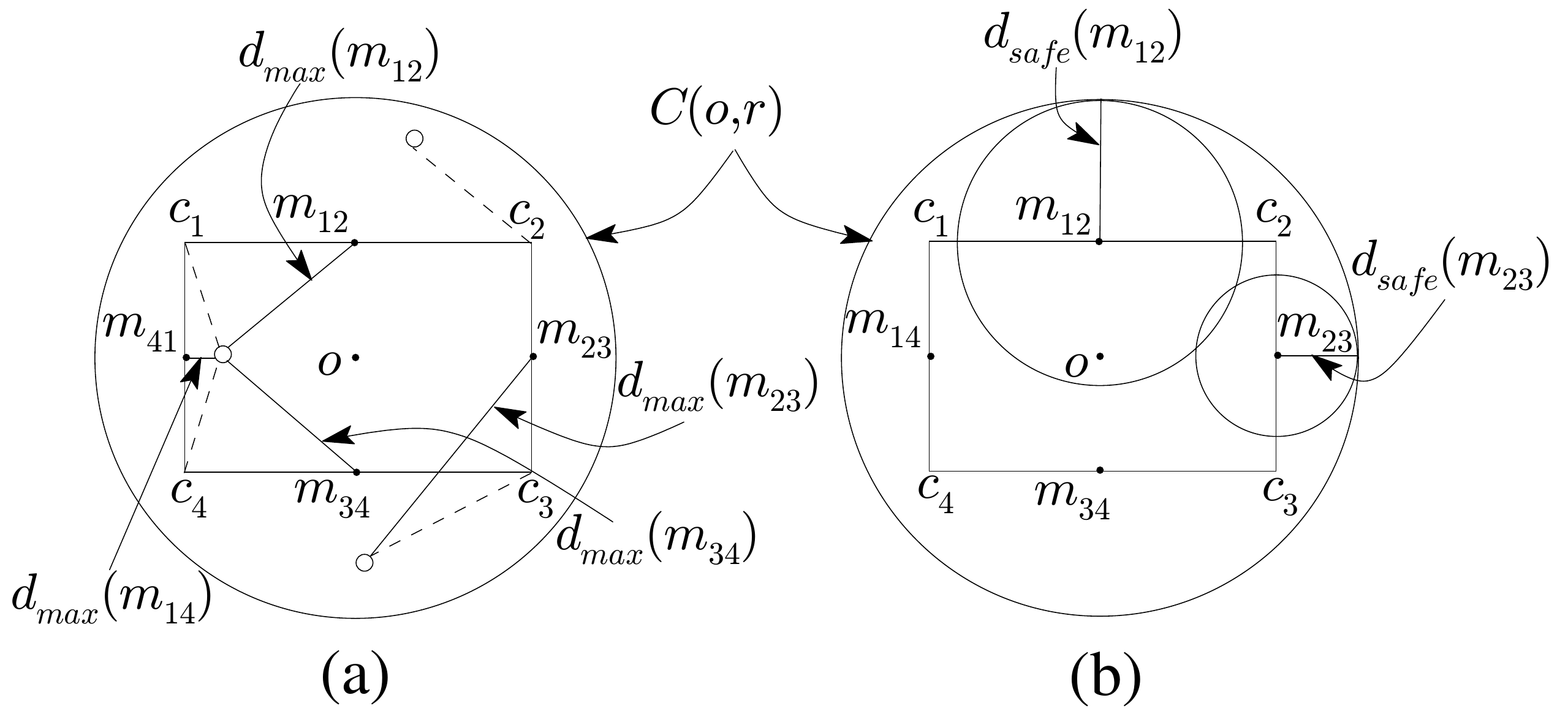}
    \caption{(a) $d_{max}=d_{max}(m_{23})$ and (b) $d_{safe}=d_{safe}(m_{23})$}
    \label{fig:algo2b}
\end{figure}

$UpdateStatus$ first updates $count(c_t,cl,P)$ using the function
$UpdateCount$. For each $p\in P$, $UpdateCount$ increments
$count(c_t,cl,P)$ by one if $CL(c_t,p)>=cl$. Note that corner
points of $R_w$ can have more than $k$ data objects with
confidence level at least $cl$ because the increase of $r$ for a
corner point of $R_w$ can make other corner points to have more
than $k$ data objects with a confidence level at least $cl$. In
the next step if $count(c_t,cl,P)$ is less than $k$ for any corner
point $c_t$ of $R_w$, $UpdateStatus$ returns the control to
Algorithm~\ref{algo:Rect-kNN} without changing $status$.
Otherwise, it computes the radius of the required known region for
ensuring the $k$ NNs with respect to $R_w$ and $cl$ (Lines
3.5-3.16). For each $m_{ij}$, $UpdateStatus$ first computes
$d_i^k$ and $d_j^k$ with the function $K_{min}$ and takes the
maximum of $d_i^k$ and $d_j^k$ as $d_{max}(m_{ij})$. Then
$UpdateStatus$ finds $d_{max}$ (Lines 3.10-3.11) and $d_{safe}$
(Line 3.12). Finally, $UpdateStatus$ checks if the size of the
current $C(o,r)$ is already equal or greater than the required
size. If this is the case then the algorithm returns $status$ as
-1, otherwise the value of the radius for the required known
region. After the call of $UpdateStatus$, \textsc{Clappinq}
continues the $BFS$ if $status \geq 0$ and terminates if
$status=-1$. For $status$ greater than 0, each time a next nearest
data object/MBR is found, \textsc{Clappinq} updates $status$ to
$-1$ if $r$ becomes greater than $status$ (Lines 2.7-2.8).

\begin{algorithm}[htbp]
\label{algo:UpdateStatus} \caption{UpdateStatus($R,cl,k,P,r$)}
    $UpdateCount(R,cl,k,P,r,count)$\;
    \uIf {$count(c_t,cl,P)\neq k$, for any corner point $c_t \in R$}
    {
        \Return 0\;
    }
    \Else
    {
        $d_{max} \leftarrow 0$\;
        \For{each middle point $m_{i,j}$}
        {
            $d_i^k\leftarrow K_{min}(m_{ij},c_i,cl,k,P)$\;
            $d_j^k\leftarrow K_{min}(m_{ij},c_j,cl,k,P)$\;
            $d_{max}(m_{ij})\leftarrow \max\{d_i^k,d_j^k\}$\;
            \If{$d_{max}(m_{ij})>d_{max}$} {$d_{max} \leftarrow d_{max}(m_{ij})$\;}
        }
        $d_{safe}\leftarrow
        r-\frac{1}{2}\times \max\{|\overline{c_1c_2}|,|\overline{c_2c_3}|\}$\;
        \uIf{$cl \times d_{max} > d_{safe}$}
        {
            \Return $(r+cl \times d_{max}-d_{safe})$\;
        }
        \Else
        {
            \Return $-1$\;
        }
    }
\end{algorithm}

In summary, \textsc{Clappinq} works in three steps. In step 1, it
runs the BFS from $o$ until it finds the $k$ NNs with a confidence
level of at least $cl$ for all corner points of $R_w$. In step 2,
from the current set of data objects it computes the radius of the
required known region to confirm that the answer set includes the
$k$ NNs with a confidence level of at least $cl$ with respect to
all points of $R_w$. Finally, in step 3, it continues to run the
BFS until the radius of the current known region is equal to the
required size.

\begin{figure}[htbp]
    \hspace{-6mm}
    \centering
        \includegraphics[width=4in]{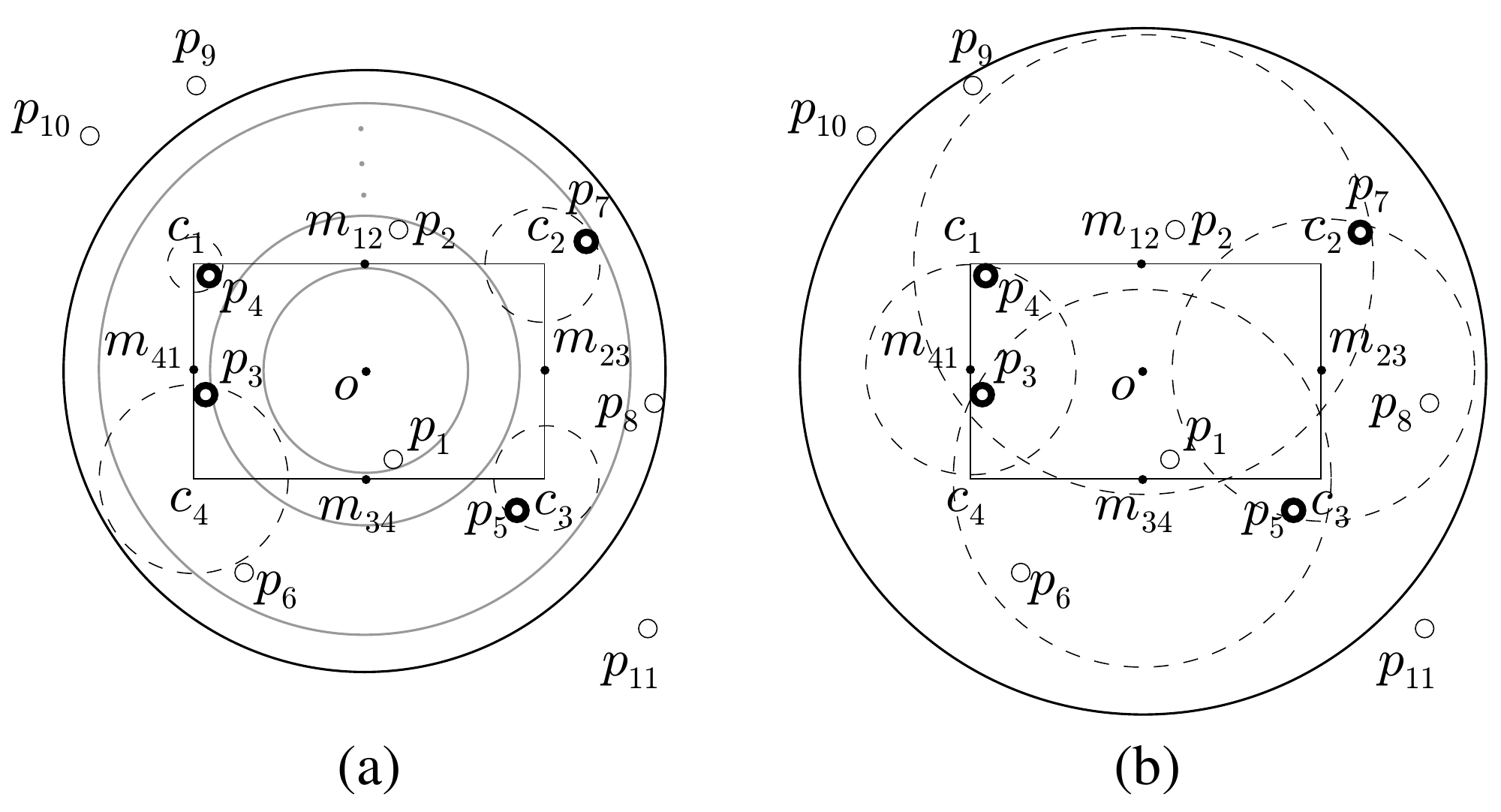}
    \caption{Steps of \textsc{Clappinq}: an example for $k=1$ and $cl=1$}
    \label{fig:algo_cknn}
\end{figure}

Figure~\ref{fig:algo_cknn} shows an example of the execution of
\textsc{Clappinq} for $k=1$ and $cl=1$. Data objects are labeled
in order of the increasing distance from $o$. \textsc{Clappinq}
starts its search from $o$ and continues until the NNs with
respect to four corner points are found as shown in
Figure~\ref{fig:algo_cknn}(a). The circles with ash border show
the continuous expanding of the known region and the circle with
black border represents the current known region. The data objects
$p_4$, $p_7$, $p_5$, and $p_3$ are the NNs with $cl=1$ from $c_1$,
$c_2$, $c_3$, and $c_4$, respectively because the four circles
with a dashed border are completely within the known region. In
the next step, the algorithm finds the maximum of $d_i^k$ and
$d_j^k$ for each $m_{ij}$. The distances $d_2^1$
(=$dist(m_{12},p_7)$), $d_2^1$ (=$dist(m_{12},p_7)$) (or $d_3^1$
(=$dist(m_{23},p_5)$)), $d_4^1$ (=$dist(m_{34},p_3)$), $d_1^1$
(=$dist(m_{41},p_4)$) are the maximum with respect to $m_{12}$,
$m_{23}$, $m_{34}$, and $m_{41}$, respectively. Finally,
\textsc{Clappinq} expands the search so that the four circles with
dashed border centered at $m_{12}$, $m_{23}$, $m_{34}$, and
$m_{41}$ and having radius $d_2^1$, $d_2^1$ (or $d_3^1$), $d_4^1$,
and $d_1^1$, respectively, are included in the known region (see
Figure~\ref{fig:algo_cknn}(b)). Therefore, the search stops when
$p_9$ is discovered and $P$ includes $p_1$ to $p_9$.


The following theorem shows the correctness for \textsc{Clappinq}.

\begin{theorem}
\label{th:3} \textsc{Clappinq} returns $P$, a candidate set of
data objects that includes the $k$ NNs with a confidence level at
least $cl$ for every point of the obfuscation rectangle $R_w$.
\end{theorem}

\begin{proof}
\textsc{Clappinq} expands the known region $C(o,r)$ from the
center $o$ of the obfuscation rectangle $R_w$ until it finds the
$k$ NNs with a confidence level at least $cl$ for all corner
points of $R_w$. Then it extends $C(o,r)$ to ensure that the
confidence level of each middle point $m_{ij}$ is at least $cl$
for both sets of $k$ nearest data objects for which $c_i$ and
$c_j$ have a confidence level at least $cl$. According to
Theorem~\ref{th:1}, this ensures that any point in
$\overline{m_{ij}c_i}$ or $\overline{m_{ij}c_j}$ has a confidence
level at least $cl$ for $k$ data objects. Again from
Lemma~\ref{lemma:6_zero}, we know that if a point has $k$ data
objects with a confidence level at least $cl$ then it also has a
confidence level at least $cl$ for its $k$ NNs. Thus, $P$ contains
the $k$ NNs with a confidence level at least $cl$ for all points
of the border of $R_w$.

To complete the proof, next we need to show that $P$ also contains
the $k$ nearest data objects with a confidence level at least $cl$
for all points inside $R_w$. The confidence level of the center
$o$ of $R_w$ for a data object $p_h$ within the known region
$C(o,r)$ is always $1$ because $C(o,r)$ is expanded from $o$ and
we have $dist(o,p_h) \leq r$. Since we have already shown that $P$
includes the $k$ NNs with a confidence level at least $cl$ for all
points of the border of $R_w$, according to Theorem~\ref{th:1} and
Lemma~\ref{lemma:6_zero}, $P$ also includes the $k$ NNs with a
confidence level at least $cl$ for all points inside $R_w$.
\end{proof}

We have proposed the fundamental algorithm and there are many
possible optimizations of it. For example,
one optimization could merge overlapping obfuscation rectangles
requested by different users at the same time, which will also
avoid redundant computation. Another optimization could exploit that $R_w$ and $R_{w+1}$ may have many
overlapping NNs. However, the focus of this paper is protecting
trajectory privacy of users while answering M$k$NN queries, and
exploring all possible optimizations of the algorithm is beyond
the scope of this paper.

\section{Experiments}\label{sec:exp}

In this section, we present an extensive experimental evaluation
of our proposed approach. In our experiments, we use both
synthetic and real data sets. Our two synthetic data sets are
generated from uniform (U) and Zipfian (Z) distribution,
respectively. The synthetic data sets contain locations of 20,000
data objects and the real data set contains 62,556 postal
addresses from California. These data objects are indexed using an
$R^{\ast}$-tree~\cite{beckmann90.SIGMOD} on a server (the LSP). We
run all of our experiments on a desktop with a Pentium 2.40 GHz
CPU and 2 GByte RAM.

In Section~\ref{sec:exp_static}, we evaluate the efficiency of our
proposed algorithm, \textsc{Clappinq}, to find $k$ NNs with a
specified confidence level for an obfuscation rectangle. We
measure the query evaluation time, I/Os, and the candidate answer
set size as the performance metric. In Section~\ref{sec:exp_cont},
we evaluate the effectiveness of our technique for preserving
trajectory privacy for M$k$NN queries.


\begin{table}[htbp]
  \centering
\begin{tabular}{|c|c|c|}
  \hline
  Parameter& Range& Default\\
  \hline
  Obfuscation rectangle area & 0.001\% to 0.01\% & 0.005\% \\
  \hline
  Obfuscation rectangle ratio & 1, 2, 4, 8 & 1 \\
  \hline
  Specified confidence level $cl$ & 0.5 to 1 & 1 \\
  \hline
  Specified number of NNs $k$ & 1 to 20 & 1 \\
  \hline
  Synthetic data set size & 5K, 10K, 15K, 20K & 20K \\
  \hline
\end{tabular}
\caption{Experimental Setup} \label{table:exp}
\end{table}

\subsection{$k$NN queries with respect to an obfuscation rectangle}\label{sec:exp_static}




There is no existing algorithm to process a PM$k$NN query. An
essential component of our approach for a PM$k$NN query is an
algorithm to evaluate a $k$NN query with respect to an obfuscation
rectangle. In this set of experiments we compare our proposed
$k$NN algorithm, \textsc{Clappinq}, with
Casper~\cite{mohamed06.VLDB}, because Casper is the only existing
related algorithm that can be adapted to process a PM$k$NN query;
further, even if we adapt it can only support $k=1$. To be more
specific, our privacy aware approach for M$k$NN queries needs an
algorithm that returns the \emph{known region} in addition to the
set of $k$ NNs with respect to an obfuscation rectangle. Among all
existing $k$NN
algorithms~\cite{chow09.SSTD,chow09.TODS,Hu06.TKDE,kalnis07.TKDE,mohamed06.VLDB,Xu09.TPDS}
only Casper supports the known region and if Casper were as
efficient as \textsc{Clappinq}, then we could extend Casper for
PM$k$NN queries for the restricted case $k=1$.

We set the data space as 10,000 $\times$ 10,000 square units. For
each set of experiments in this section, we generate $1000$ random
obfuscation rectangles of a specified area, which are uniformly
distributed in the total data space. We evaluate a $k$NN query
with respect to 1000 obfuscation rectangles and measure the
average performance with respect to a single obfuscation rectangle
for Casper and \textsc{Clappinq} in terms of the query evaluation
time, the number of page accesses, i.e., I/Os, and the candidate
answer set size. The page size is set to 1 KB which corresponds to
a node capacity of 50 entries.

Note that, in our experiments, the communication amount (i.e., the
answer set size) represents the communication cost independent of
communication link (e.g., wireless LANs, cellular link) used. The
communication delay can be approximated from the known latency of
the communication link. In our technique, sometimes the answer set
size may become large to satisfy the user's privacy requirement.
Though the large answer set size may result in a communication
delay, nowadays this should not be a problem. The latency of
wireless links has been significantly reduced, for example HSPA+
offers a latency as low as 10ms. Furthermore, our analysis
represents the communication delay scenario in the worst case. In
practice, the latency of first packet is higher than the
subsequent packets and thus, the communication delay does not
increase linearly with the increase of the answer set size.
\begin{figure}[htbp]
  \begin{center}
    \begin{tabular}{cc}
      \resizebox{45mm}{!}{\includegraphics{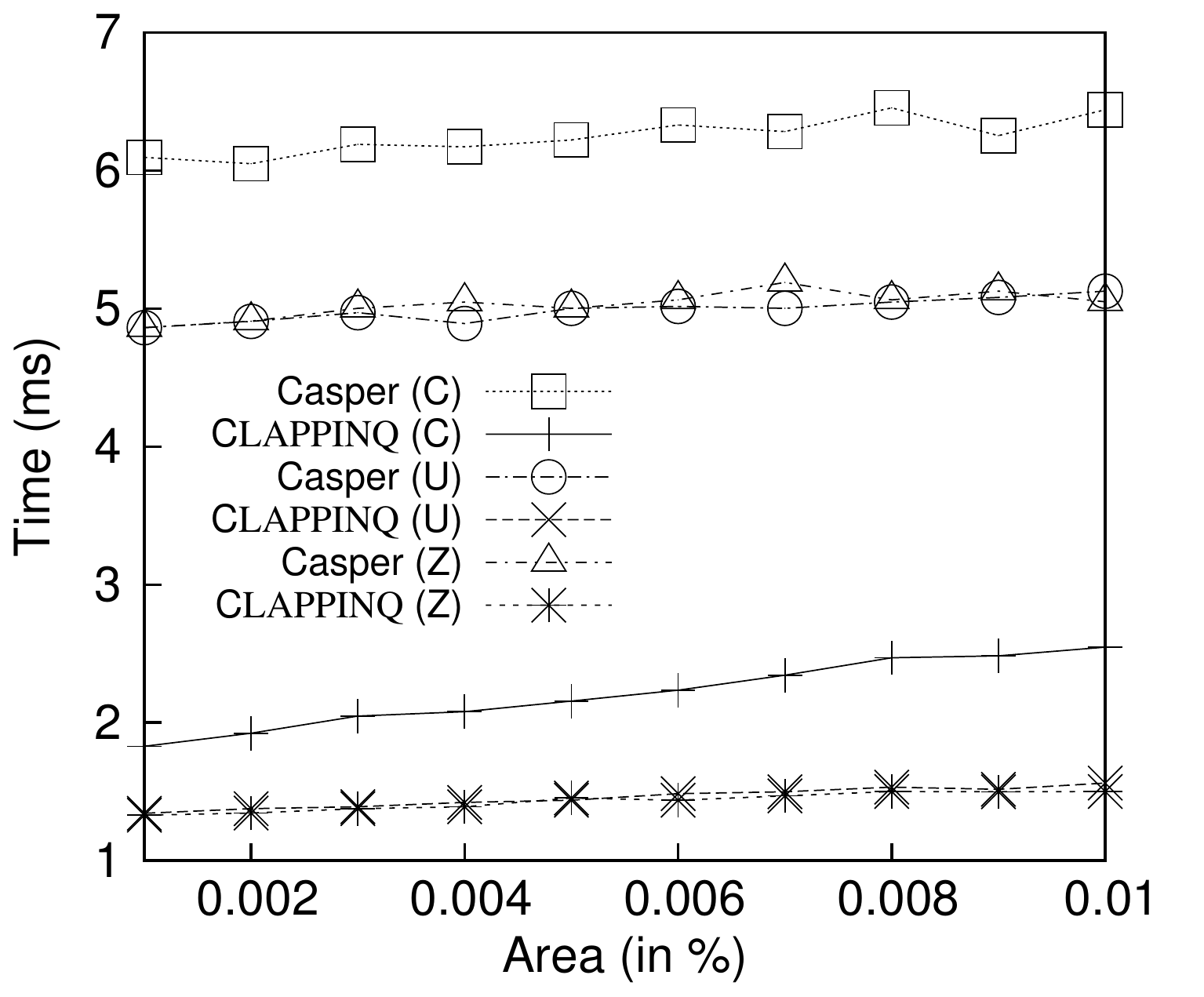}} &
      \resizebox{43mm}{!}{\includegraphics{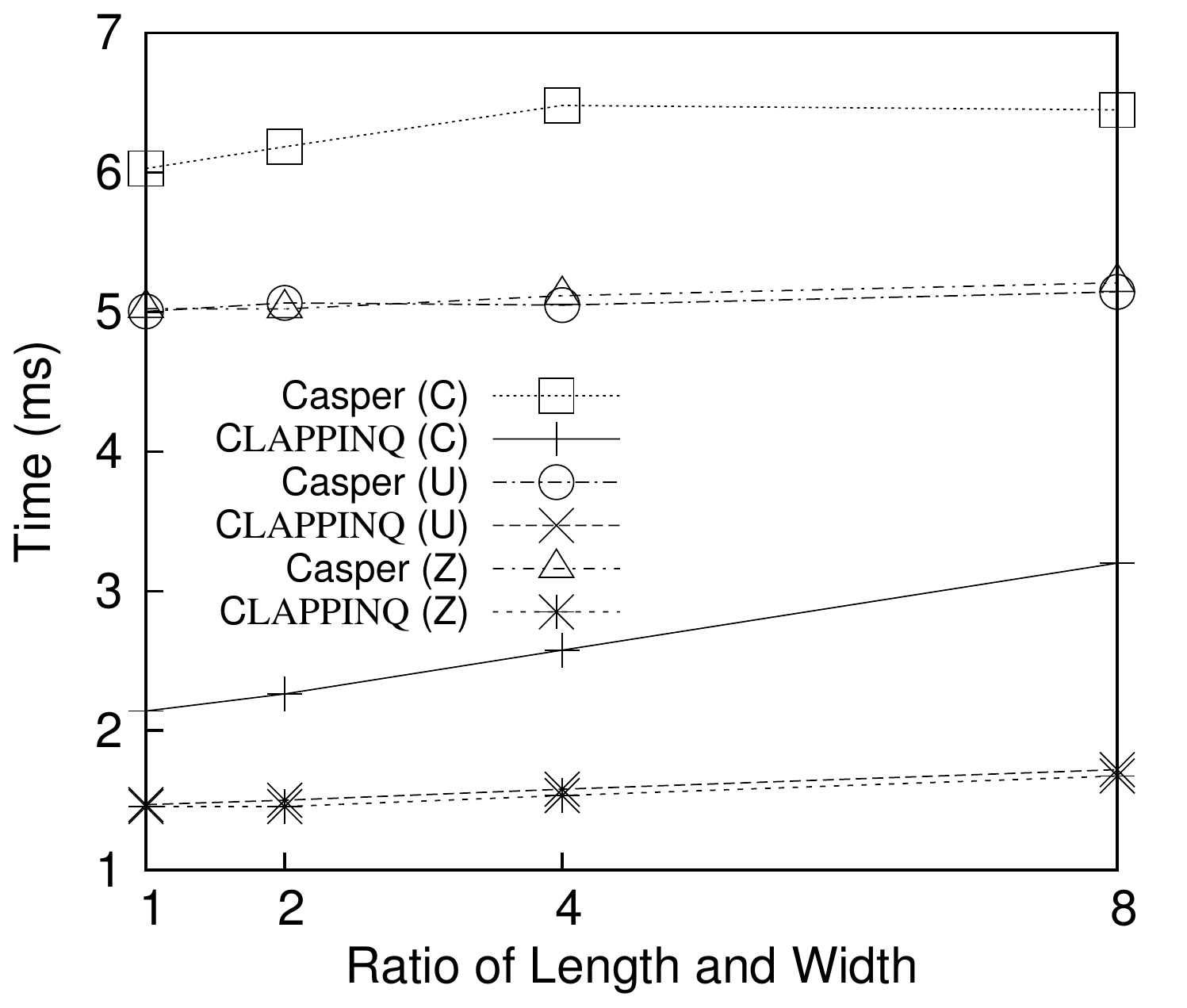}} \\
      \scriptsize{(a)\hspace{0mm}} & \scriptsize{(b)}\\
        \resizebox{45mm}{!}{\includegraphics{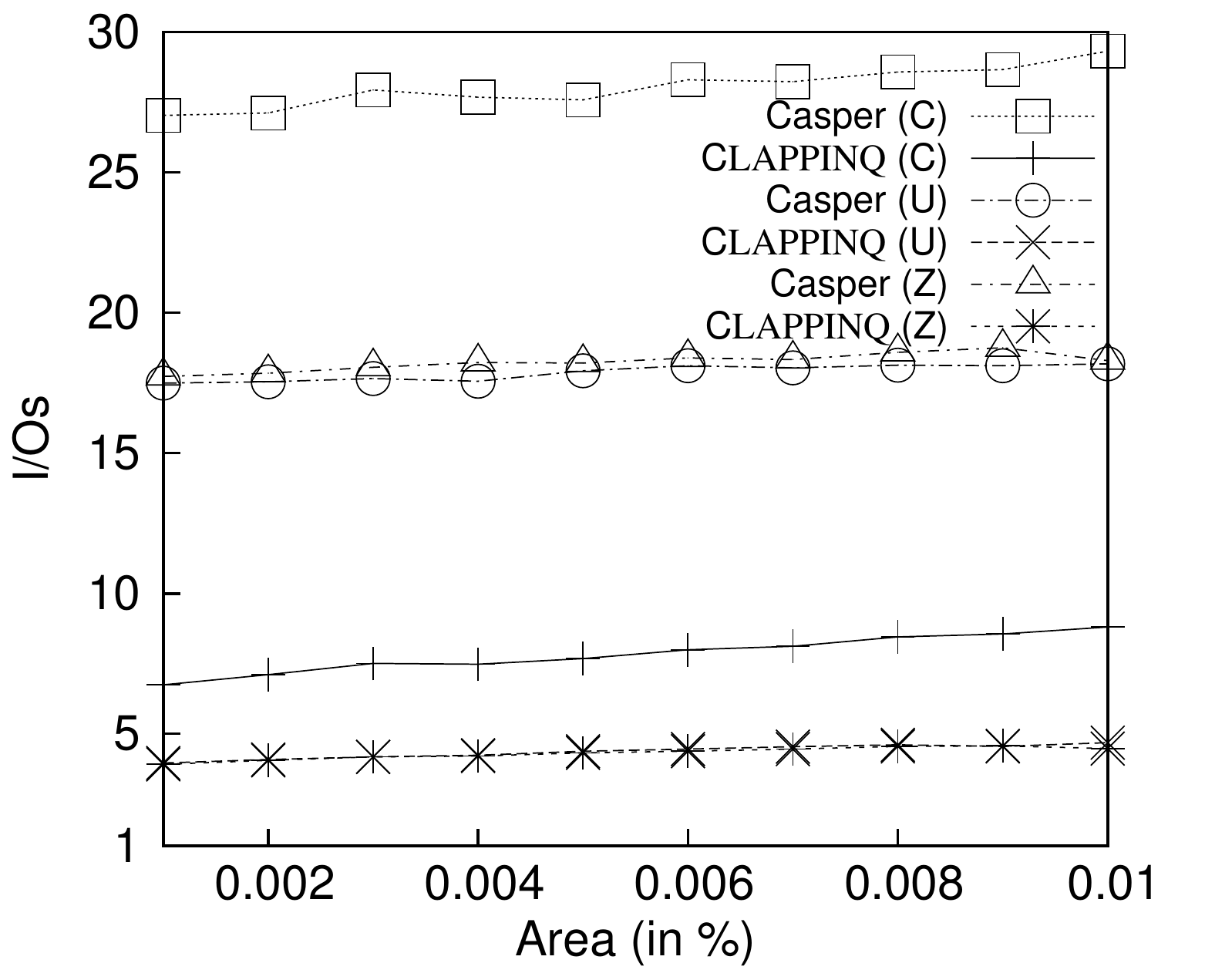}} &
      \resizebox{43mm}{!}{\includegraphics{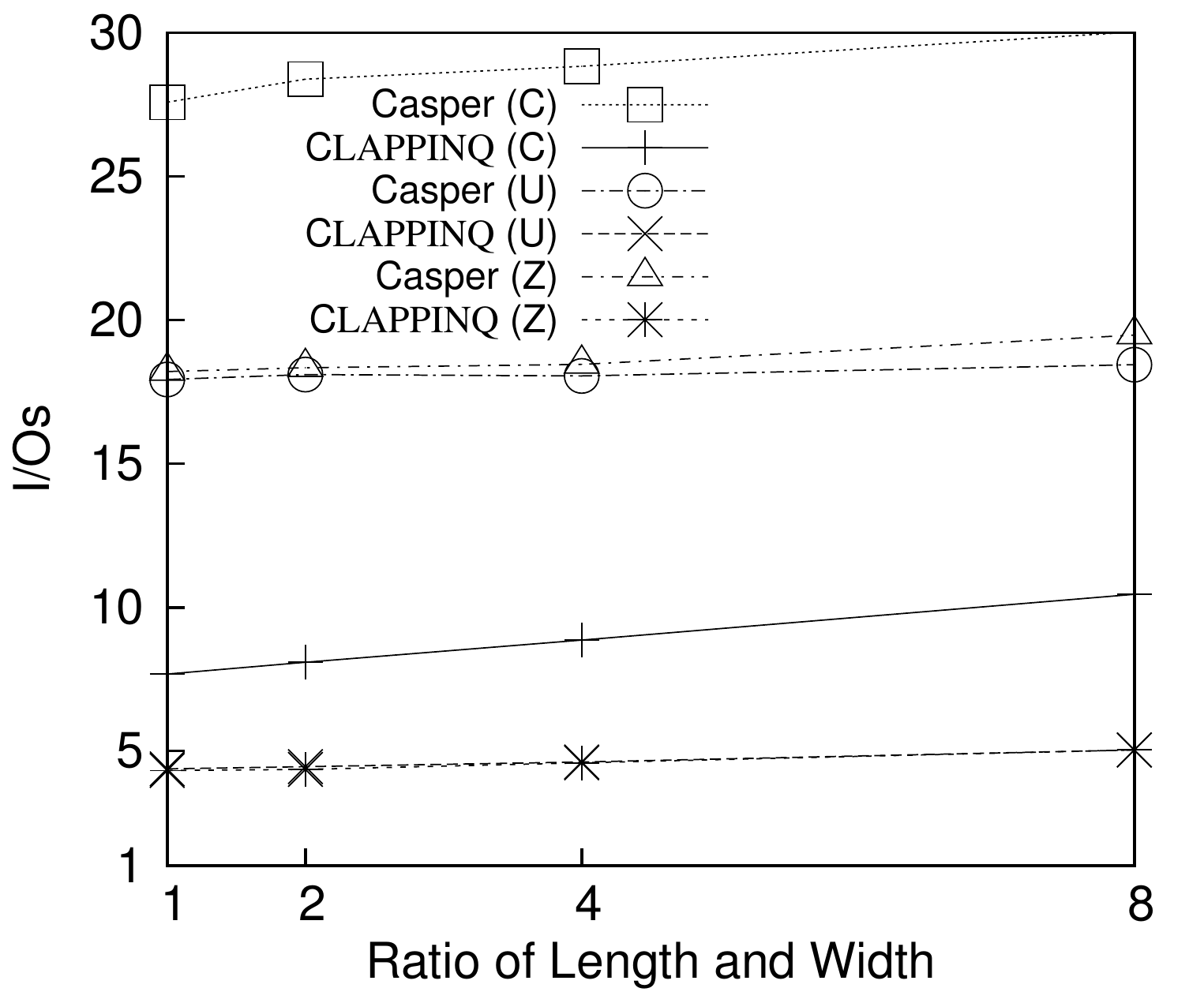}} \\
      \scriptsize{(c)\hspace{0mm}} & \scriptsize{(d)}\\
      \resizebox{45mm}{!}{\includegraphics{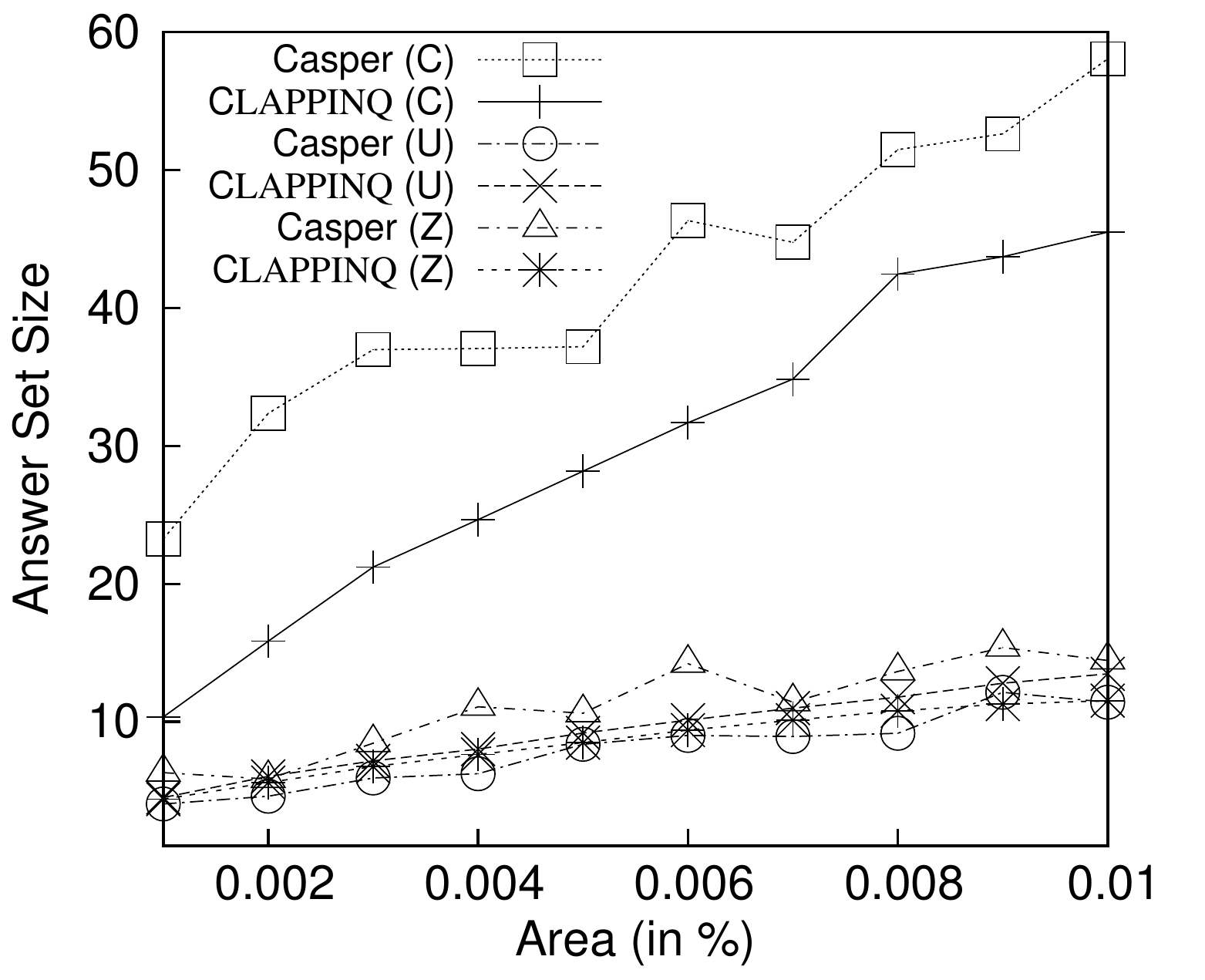}} &
      \resizebox{43mm}{!}{\includegraphics{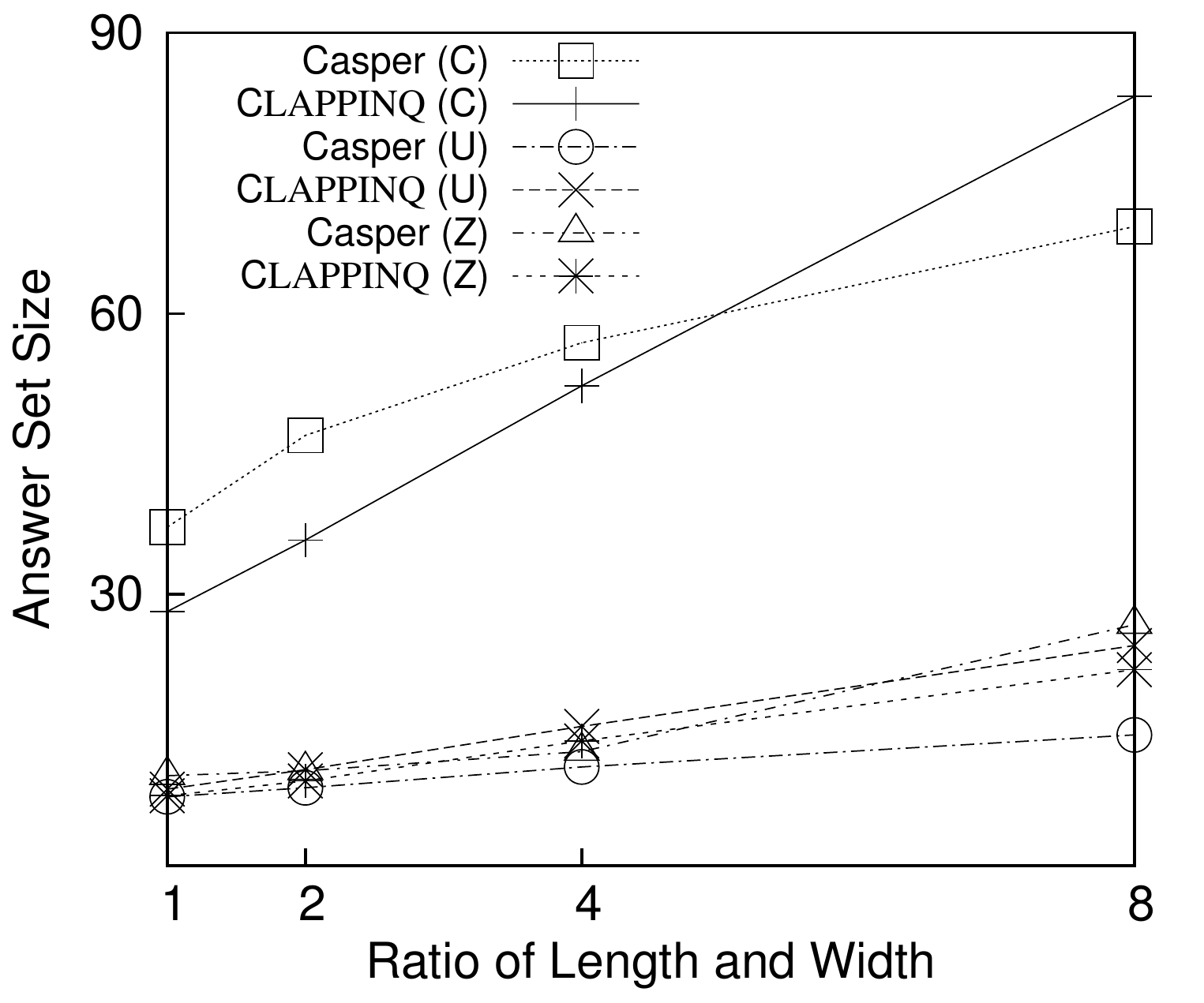}} \\
       \scriptsize{(e)\hspace{0mm}} & \scriptsize{(f)}
        \end{tabular}
    \caption{The effect of obfuscation rectangle area and ratio}
    \label{fig:graph1}
  \end{center}
\end{figure}
In different sets of experiments, we vary the following
parameters: the area of the obfuscation rectangle, the ratio of
the length and width of the obfuscation rectangle, the specified
confidence level, the specified number of NNs and the synthetic
data set size. Table~\ref{table:exp} shows the range and default
value for each of these parameters. We set 0.005\% of the total
data space as the default area for the obfuscation rectangle,
since it reflects a small suburb in California (about $20$ km$^2$
with respect to the total area of California) and is sufficient to
protect privacy of a user's location. The thinner an obfuscation
rectangle, the higher the probability to identify a user's
trajectory~\cite{Duckham05.COSIT}. Hence, we set 1 as a default
value for the ratio of the obfuscation rectangle to ensure the
privacy of the user. The original approach of Casper does not have
the concept of confidence level and only addresses 1NN queries. To
compare our approach with Casper, we set the default value in
\textsc{Clappinq} for $k$ and the confidence level as 1.

In Sections~\ref{sec:static_comp} and~\ref{sec:static_dataset}, we
evaluate and compare \textsc{Clappinq} with Casper. In
Section~\ref{sec:static_k}, we study the impact of $k$ and the
confidence level only for \textsc{Clappinq} as Casper cannot be
directly applied for $k>1$ and has no concept of a confidence
level.




\subsubsection{The effect of obfuscation rectangle area}
\label{sec:static_comp}

In this set of experiments, we vary the area of obfuscation
rectangle from 0.001\% to 0.01\% of the total data space. A larger
obfuscation rectangle represents a more imprecise location of the
user and thus ensures a higher level of privacy. We also vary the
obfuscation rectangle ratio as 1,2,4, and 8. A smaller ratio of
the width and length of the obfuscation rectangle provides the
user with a higher level of privacy.

Figures~\ref{fig:graph1}(a) and~\ref{fig:graph1}(b) show that
\textsc{Clappinq} is on an average 3times faster than Casper for
all data sets. The I/Os are also at least 3 times less than Casper
(Figures~\ref{fig:graph1}(c) and~\ref{fig:graph1}(d)). The
difference between the answer set size for \textsc{Clappinq} and
Casper is not prominent. However, in most of the cases
\textsc{Clappinq} results in a smaller answer set compared with
that of Casper (Figures~\ref{fig:graph1}(e)
and~\ref{fig:graph1}(f)). We also observe that the performance is
better when the obfuscation rectangle is a square and it continues
to degrade for a larger ratio in both \textsc{Clappinq} and Casper
(Figures~\ref{fig:graph1}(b),~\ref{fig:graph1}(d),
and~\ref{fig:graph1}(f)).

%


\begin{figure*}[htbp]
  \begin{center}
    \begin{tabular}{ccc}
        \hspace{-5mm}
      \resizebox{43mm}{!}{\includegraphics{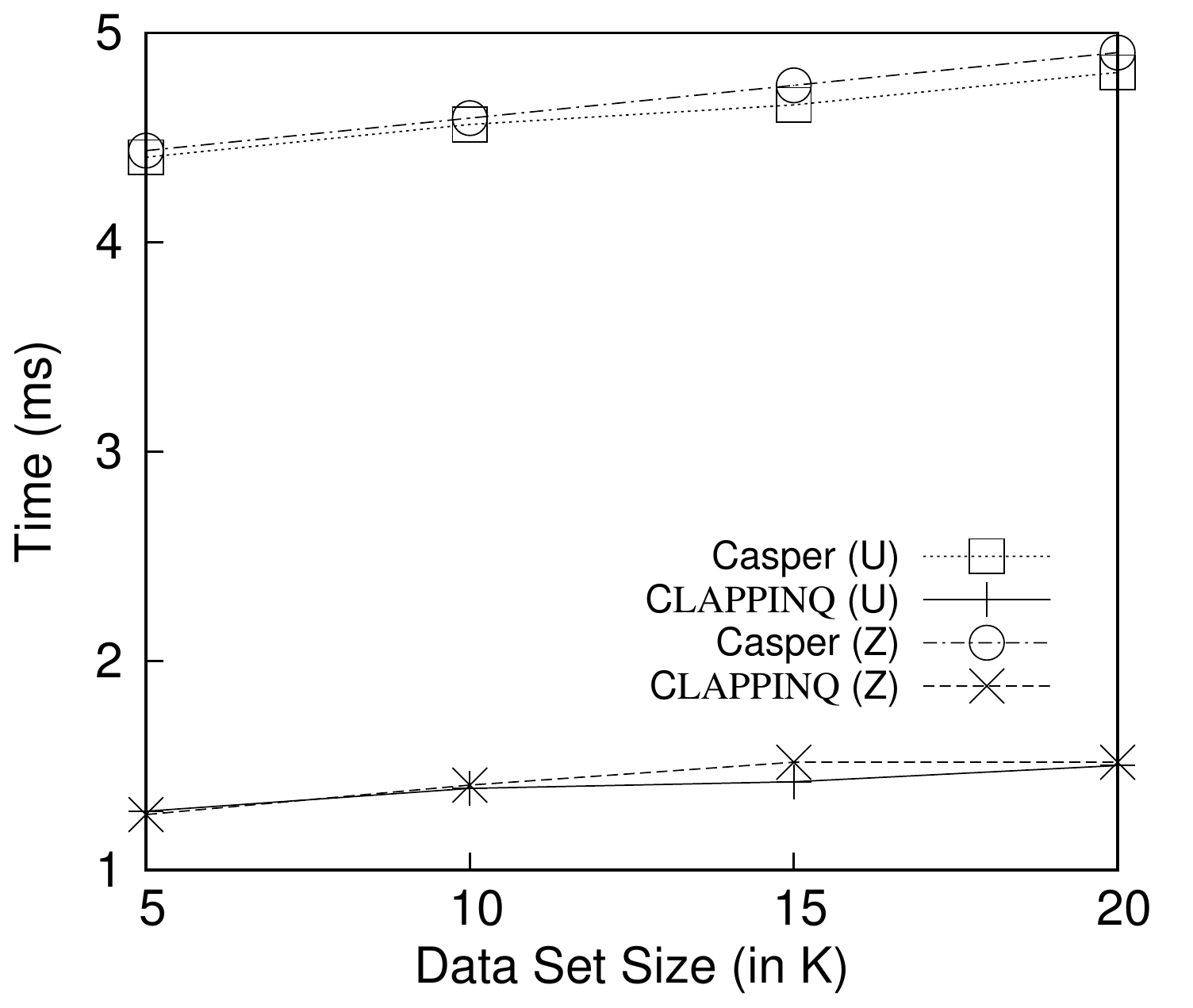}} &
        \hspace{-5mm}
       \resizebox{45mm}{!}{\includegraphics{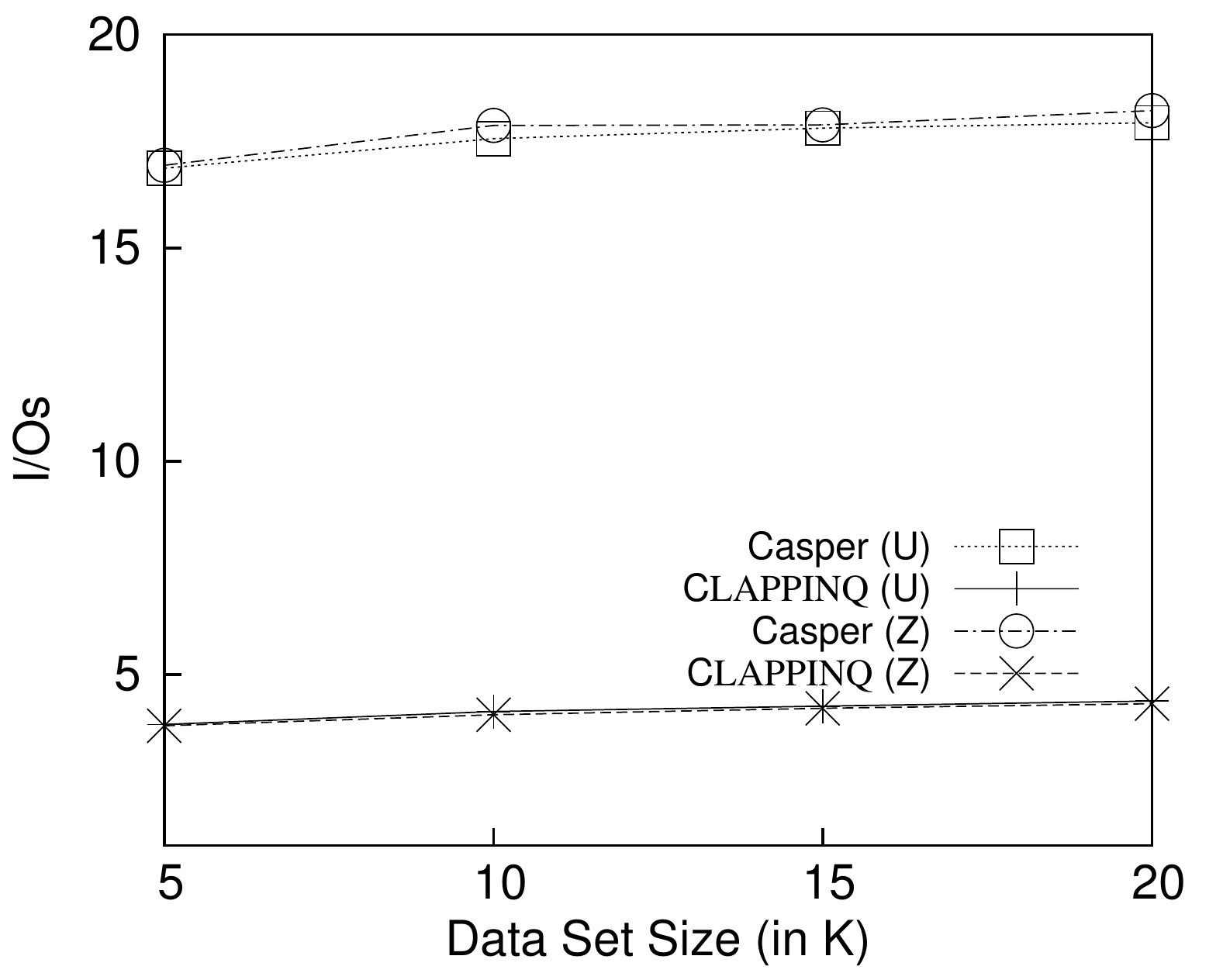}} &
        \hspace{-5mm}
      \resizebox{43mm}{!}{\includegraphics{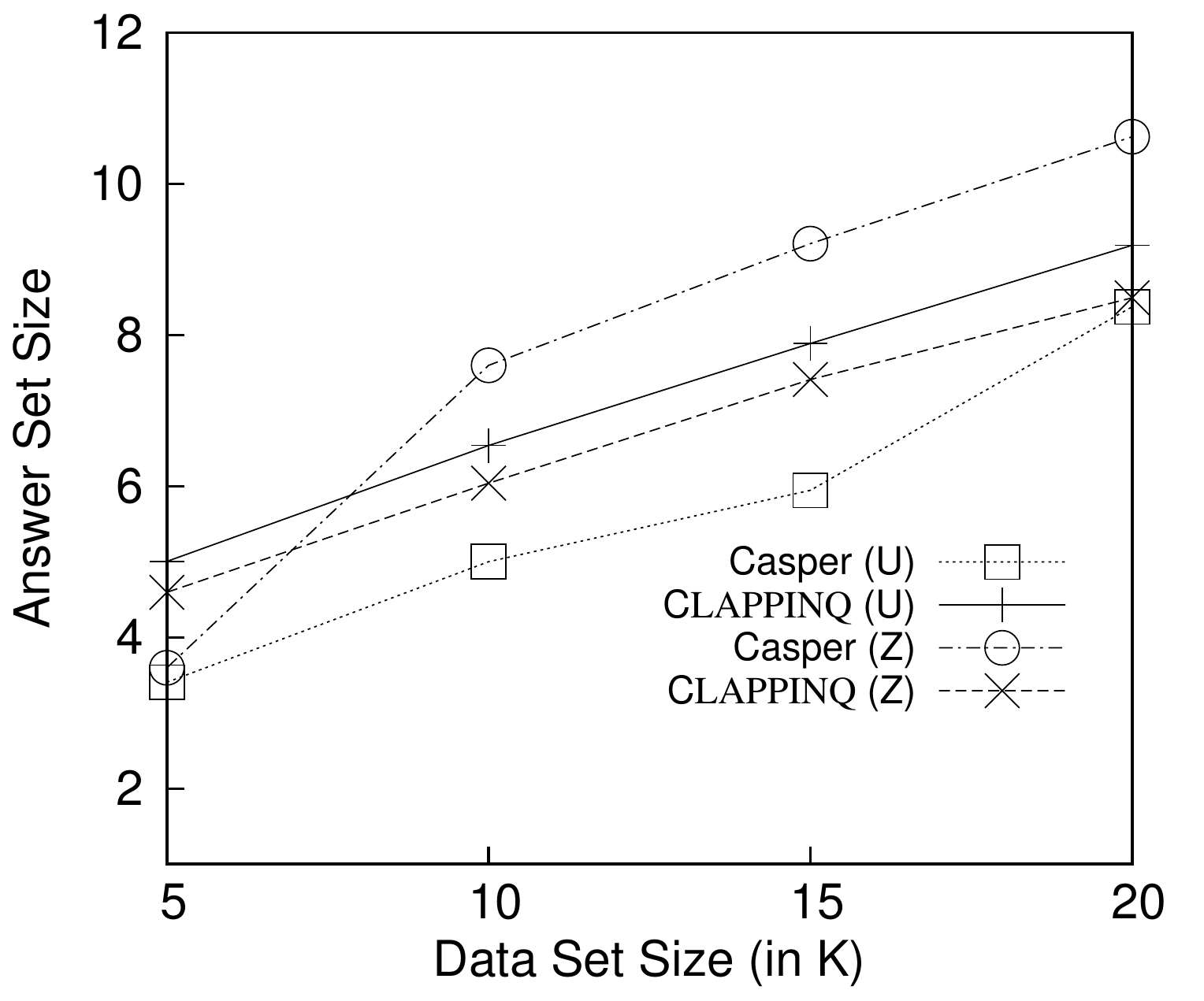}} \\
       \scriptsize{(a)\hspace{0mm}} & \scriptsize{(b)} & \scriptsize{(c)}
        \end{tabular}
    \caption{The effect of data set size}
    \label{fig:graph3}
  \end{center}
\end{figure*}

\subsubsection{The effect of the data set size}
\label{sec:static_dataset}

We vary the size of the synthetic data set as 5K, 10K, 15K and
20K, and observe that \textsc{Clappinq} is significantly faster
than that of Casper for any size of data set.
Figure~\ref{fig:graph3} shows the results for the query evaluation
time, I/Os and the answer set size. We find that \textsc{Clappinq}
is at least 3 times faster and the I/Os of \textsc{Clappinq} is at
least 4 times less than that of Casper. The time, the I/Os and the
answer set size slowly increases with the increase of data set
size.

\begin{figure*}[htbp]
  \begin{center}
    \begin{tabular}{ccc}
        \hspace{-5mm}
      \resizebox{43mm}{!}{\includegraphics{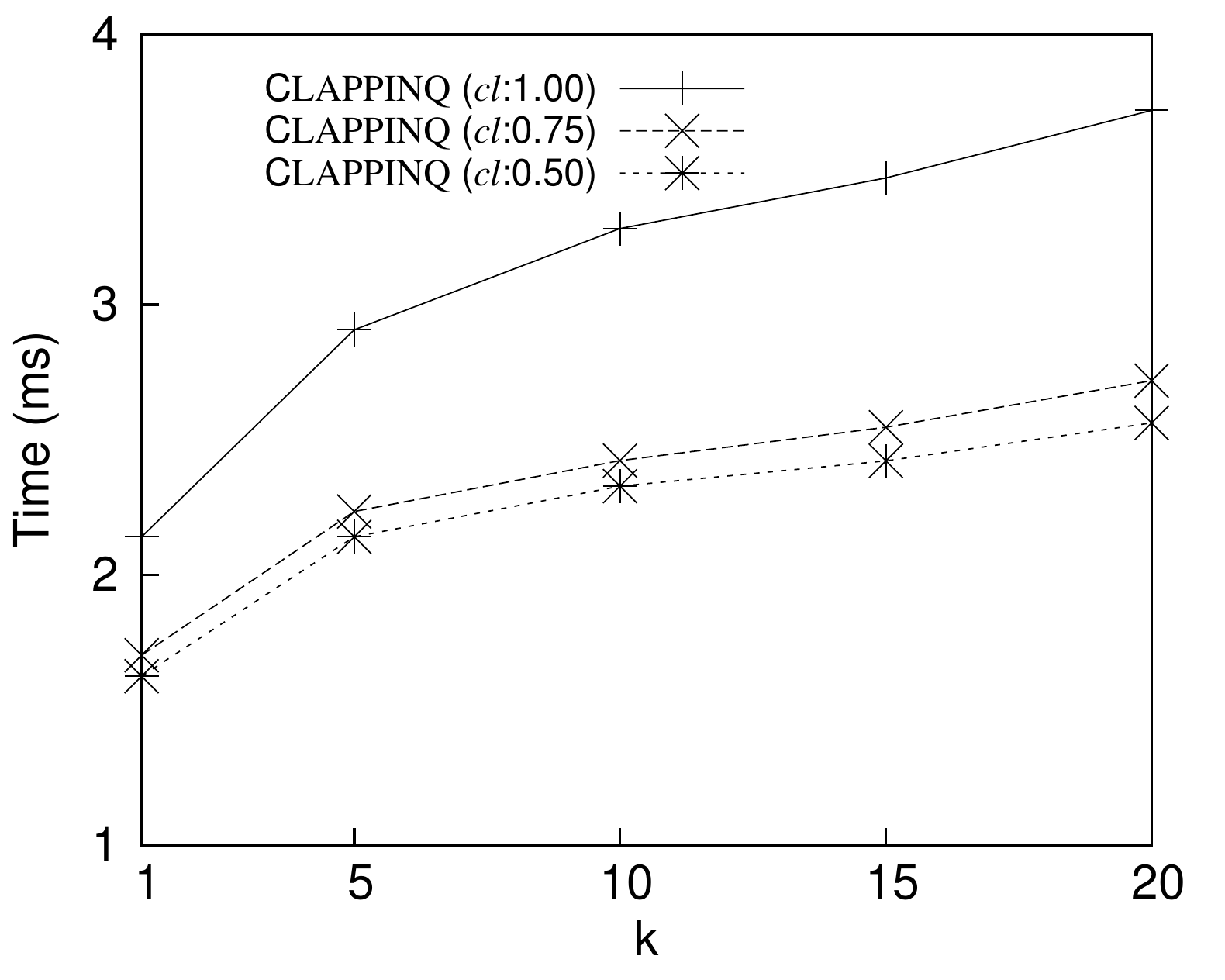}} &
        \hspace{-5mm}
       \resizebox{43mm}{!}{\includegraphics{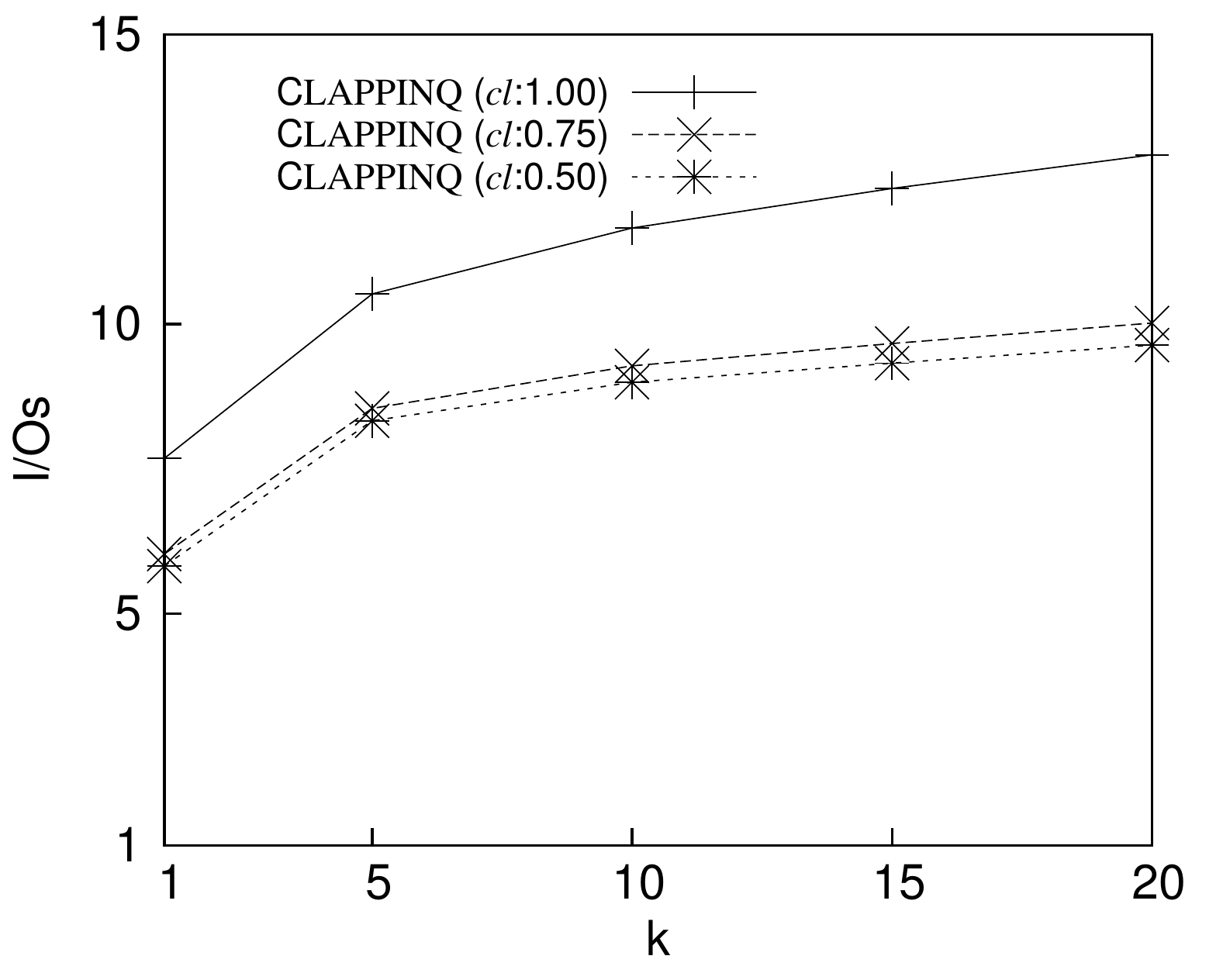}} &
        \hspace{-5mm}
      \resizebox{43mm}{!}{\includegraphics{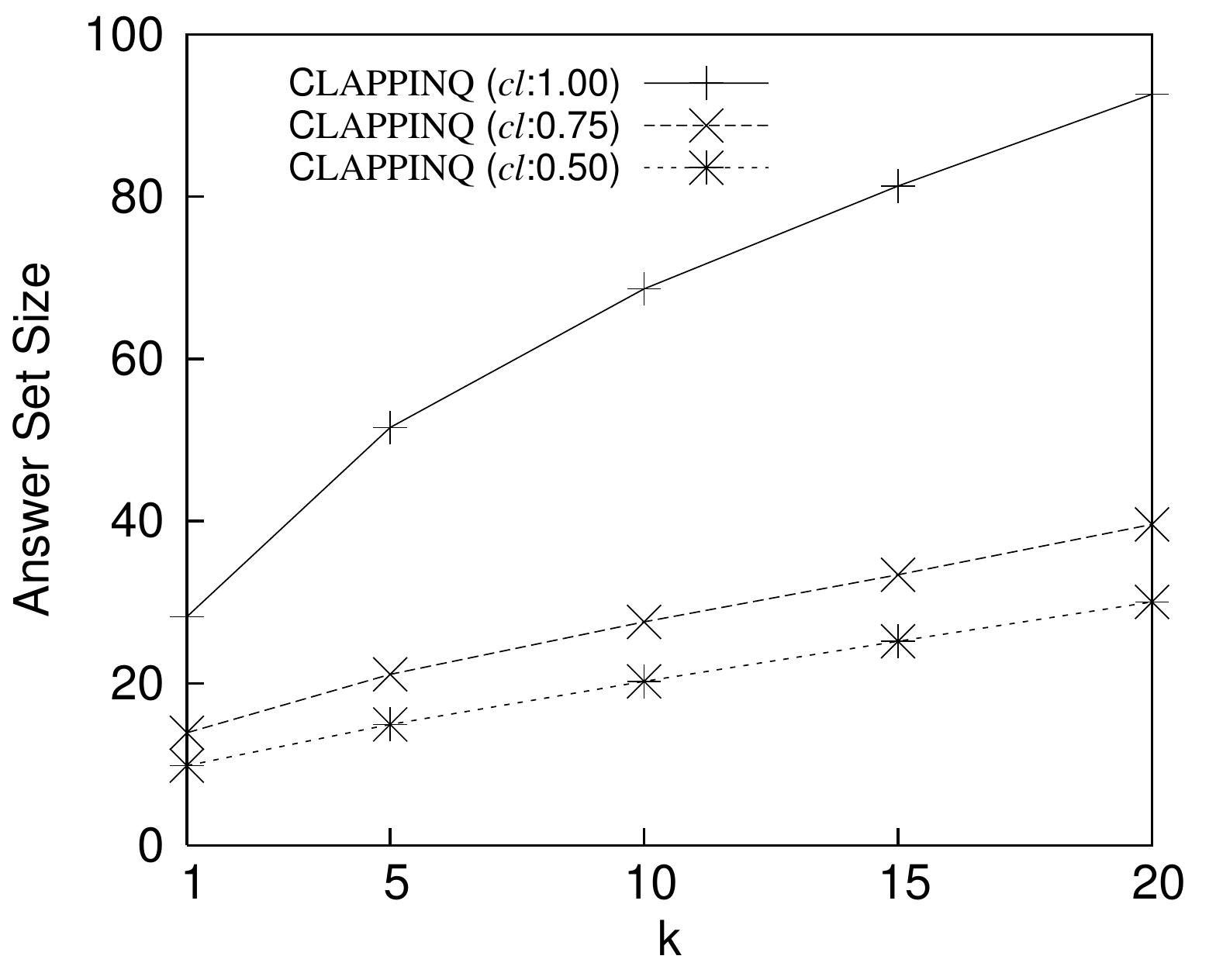}} \\
       \scriptsize{(a)\hspace{0mm}} & \scriptsize{(b)} & \scriptsize{(b)}
        \end{tabular}
    \caption{The effect of the parameter $k$ and confidence level}
    \label{fig:graph4}
  \end{center}
\end{figure*}
\subsubsection{The effect of $k$ and the confidence level}
\label{sec:static_k}

In this set of experiments, we observe that the query evaluation
time, I/Os, and the answer set size for \textsc{Clappinq} increase
with the increase of $k$ for all data sets. However, these
increasing rates decrease as $k$ increases
(Figure~\ref{fig:graph4} for the California data set). We also
vary the confidence level $cl$ and expect that a lower $cl$ incurs
less query processing and communication overhead.
Figure~\ref{fig:graph4} also shows that the average performance
improves as $cl$ decreases and the improvement is more pronounced
for higher values of $cl$. For example, the answer set size
reduces by an average factor of 2.35 and 1.37 when $cl$ decreases
from 1.00 to 0.75 and from 0.75 to 0.50, respectively.


\subsubsection{\textsc{Clappinq} vs. Casper for PM$k$NN queries}
The paper that proposed Casper~\cite{mohamed06.VLDB} did not
address trajectory privacy for M$k$NN queries. Even if we extend
it for PM$k$NN queries using our technique, Casper would only work
for $k=1$. More importantly, since we have found that
\textsc{Clappinq} is at least 2 times faster and requires at least
3 times less I/Os than Casper for finding the NNs for an
obfuscation rectangle, and an M$k$NN query requires the evaluation
of a large number of consecutive obfuscation rectangles,
\textsc{Clappinq} would outperform Casper by a greater margin for
PM$k$NN queries. Therefore, we do not perform such experiments and
conclude that \textsc{Clappinq} is efficient than Casper for
PM$k$NN queries.


\subsection{Effectiveness of our technique for trajectory privacy protection}\label{sec:exp_cont}

%
We first define a measure for trajectory privacy in
Section~\ref{sec:exp_measure_trajprivacy}. Then based on our
measure, we evaluate the effectiveness of our technique. In
Section~\ref{sec:exp_cont_qarea}, we compare three possible
options of our algorithm \textsc{Request\_PM$k$NN} for different
obfuscation rectangle areas: (i) hiding the required confidence
level, (ii) hiding the required number of nearest data objects,
and (iii) hiding both of them. We report the experimental results
for different required and specified confidence levels in
Section~\ref{sec:exp_cont_cl} and for different required and
specified number of nearest data objects in
Section~\ref{sec:exp_cont_k}. We also present the experimental
results by varying the value of $\delta$ in
Section~\ref{sec:exp_cont_delta}.

To simulate moving users, we first randomly generate starting
points of 20 trajectories which are uniformly distributed in the
data space and then generate the complete trajectory for each of
these starting points. Each trajectory has the length of $5000$
units and consists of a series of random points, where the
consecutive points are connected with a straight line of a random
length between $1$ to $10$ units. Note that the data space is set
as 10,000 $\times$ 10,000 square units. We generate the
obfuscation rectangle with a specified area when a moving user
needs to send a request. Though it is not always possible to have
the ratio of the obfuscation rectangle's length and width as 1,
our algorithm keeps the ratio as close as possible to 1: the
obfuscation rectangle needs to be inside the current known region;
sometimes the user's location is close to the boundary of the
known region and to include the user's obfuscation rectangle
inside the known region (circle), a ratio of 1 might not be
possible. Therefore we adjust the ratio of the length and width of
the obfuscation rectangle to accommodate it within the known
region.). Since the obfuscation rectangle generation procedure
is random, for each trajectory we repeat every experiment 25
times, and present the average performance results. According to
Algorithm~\ref{algo:RequestkNN}, a user can modify $cl$, $cl_r$,
$k$, $k_r$ and $\delta$ with her requirement in the consecutive
request of obfuscation rectangles for an M$k$NN query. However, in
our experiments, for the sake of simplicity, we use fixed values
for these parameters in the consecutive requests of obfuscation
rectangles for an M$k$NN query. The default value for the user's safe distance $\delta$
is set to 10.

We consider the overlapping rectangle attack and the combined
attack (i.e., the overlapping rectangle attach and the maximum
movement bound attack) in our experiments. The
combined attack arises when the user's maximum velocity is known
to the LSP. To derive the maximum movement bound in case of
combined attack, we set the user's maximum velocity as 60 km/hour.
For simplicity, we assume that the user also moves at constant
velocity of 60 km/hour.

%
%

%

The query evaluation time, I/Os, and the answer set size for a
PM$k$NN query is measured by adding the required query evaluation
time, I/Os, and answer set size for every requested obfuscation
rectangle per trajectory of length 5000 units in the data space of
10,000 $\times$ 10,000 square units.


\subsubsection{Measuring the level of trajectory privacy}
\label{sec:exp_measure_trajprivacy}

In our experiments, we measure the level of trajectory privacy by
two parameters: (i) the trajectory area, i.e., the approximated
location of the user's trajectory by the LSP, and (ii) the
frequency, i.e., the number of requested obfuscation rectangles
per a user's trajectory for a fixed obfuscation rectangle area.

The trajectory area is computed from the available knowledge of
the LSP. The LSP knows the set of obfuscated rectangles provided
by a user and the known region for each obfuscated rectangle. The
LSP does not know the user's required confidence level $cl_r$ and
the required number of data objects $k_r$ and thus, cannot compute
$GCR(cl_r,k_r)$. Although the LSP can compute $GCR(cl,k)$ from the
user's specified confidence level $cl$ and the specified number of
data object $k$, $GCR(cl,k)$ does not guarantee that the user's
location resides in $GCR(cl,k)$ for the current obfuscation
rectangle. We know that the user needs to reside within
$GCR(cl_r,k_r)$ of the current obfuscation rectangle to ensure the
required confidence level for the required number of data objects.
However, the LSP knows the known region $C(o,r)$ and that
$GCR(cl_r,k_r)$ must be inside the known region of the current
obfuscation rectangle because the confidence level of the user for
any data object outside the known region is 0. Thus, the
trajectory area for a user's trajectory is defined as the union of
the known regions with respect to the set of obfuscation
rectangles provided by the user for that trajectory. When the LSP
knows the maximum velocity, then the LSP can use the maximum
movement bound in addition to the known region to determine the
trajectory area. Formally, we define trajectory area as follows:

\textbf{Definition 8.1} \textbf{(\textit{Trajectory Area})} Let
$\{R_1, R_2,...,R_n\}$ be a set of $n$ consecutive rectangles
requested by a user to an LSP in an M$k$NN query, $C_i(o,r)$ be
the known region corresponding to $R_i$, and $M_i$ be the maximum
movement bound corresponding to $R_i$. The trajectory area is
computed as $\cup_{1 \leq i \leq n-1}{(C_i(o,r) \cap M_i) \cup
C_n(o,r)}$. If the maximum bound is unknown to the LSP then the
trajectory area is expressed as $\cup_{1 \leq i \leq
n}{C_i(o,r)}$.

\begin{figure}[htbp]
  \begin{center}
    \begin{tabular}{cc}
      \resizebox{40mm}{!}{\includegraphics{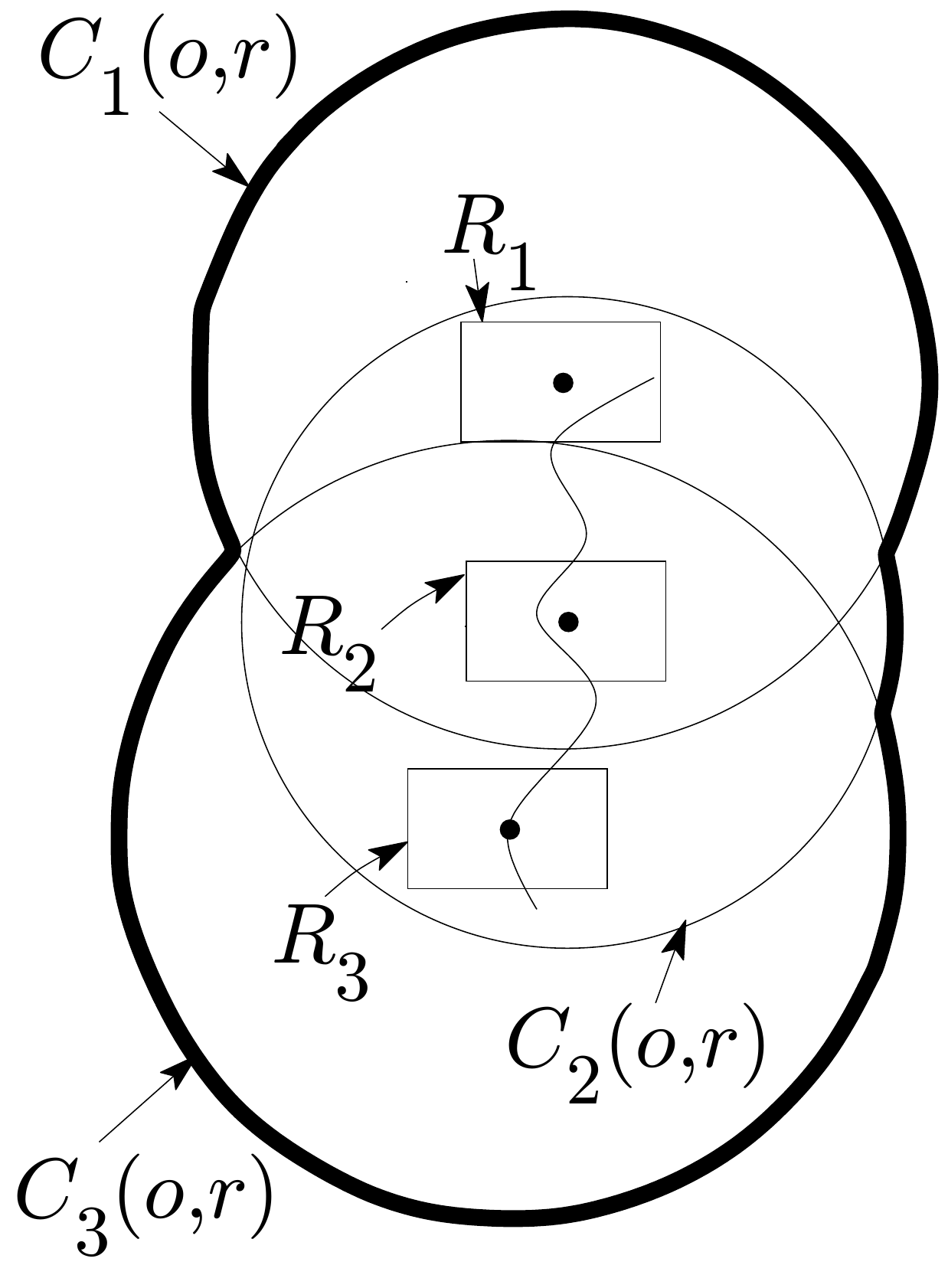}} &
      \resizebox{40mm}{!}{\includegraphics{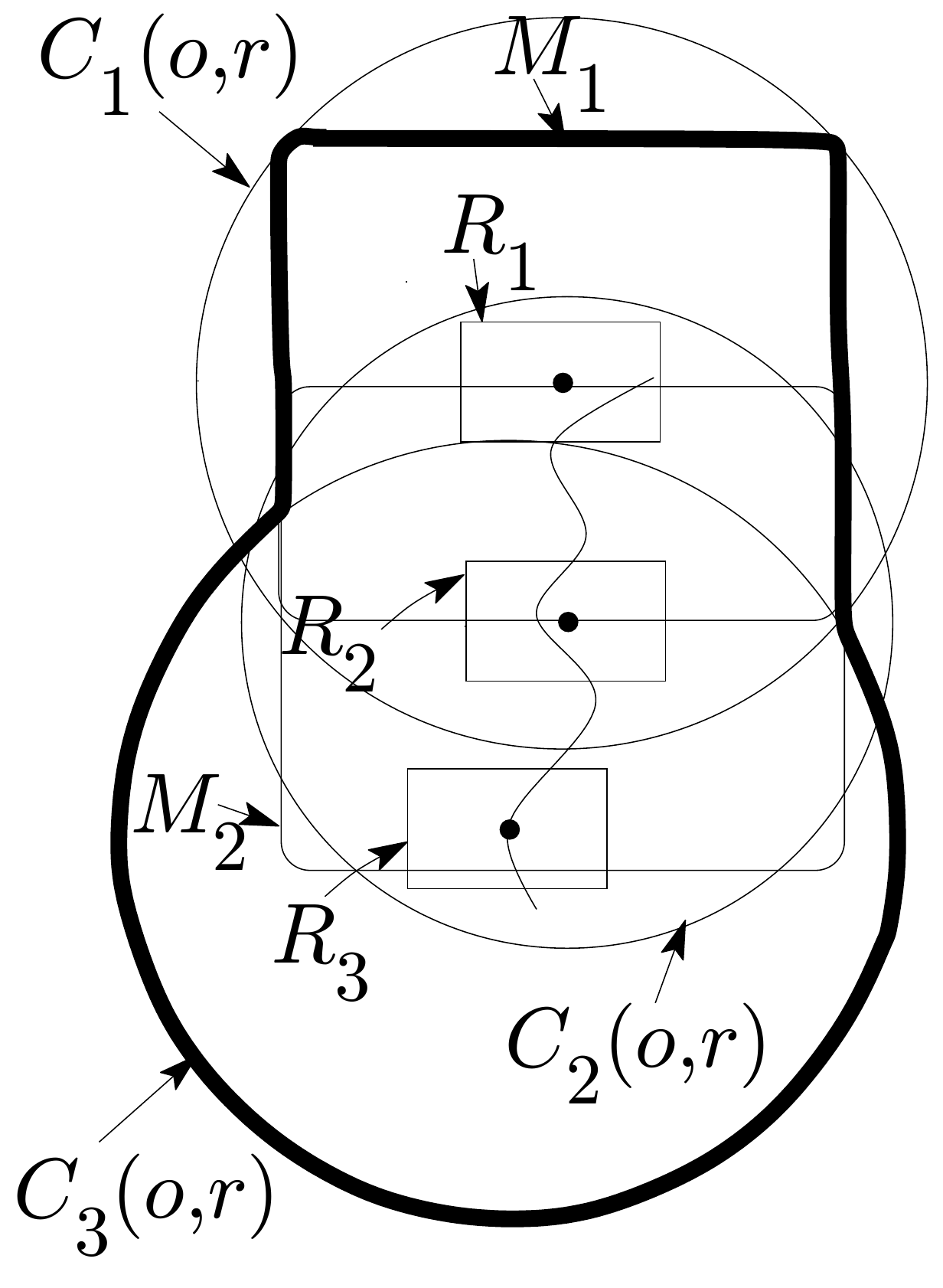}} \\
       \scriptsize{(a)\hspace{0mm}} & \scriptsize{(b)\hspace{-0mm}}
        \end{tabular}
        \end{center}
    \caption{The bold line shows the trajectory area if the maximum velocity is (a) unknown to the LSP, (b) known to the LSP}
    \label{fig:traj_area}
\end{figure}

Figure~\ref{fig:traj_area}(a) and~\ref{fig:traj_area}(b) show
trajectory areas when the user's maximum velocity is either unknown or
known to the LSP, respectively. The larger the trajectory area,
the higher the privacy for the user. This is because the
probability is high for a large trajectory area to contain
different sensitive locations and the probability is low that an
LSP can link the user's trajectory with a specific location. On
the other hand, for a fixed obfuscation rectangle area, a lower
frequency for a trajectory represents high level of trajectory
privacy since a smaller number of spatial constraints are
available for an LSP to predict the user's trajectory.

In our experiments, we compute the trajectory area through Monte
Carlo Simulation. We randomly generate 1 million points in the
total space. For the overlapping rectangle attack, we determine
the trajectory area as the percentage of points that fall inside
$\cup_{1 \leq i \leq n}{C_i(o,r)}$. On the other hand, for the
combined attack (i.e., the maximum velocity is known to the LSP),
we determine the trajectory area as the percentage of points that
fall inside $\cup_{1 \leq i \leq n-1}{(C_i(o,r) \cap M_i) \cup
C_n(o,r)}$.

Thus, the trajectory area is measured as percentage of the total
data space. On the other hand, the frequency is measured as the
number of requested obfuscation rectangles per trajectory of
length 5000 units in the data space of 10,000 $\times$ 10,000
square units.

\begin{figure*}[htbp]
  \begin{center}
    \begin{tabular}{ccccc}
        \hspace{-5mm}
      \resizebox{43mm}{!}{\includegraphics{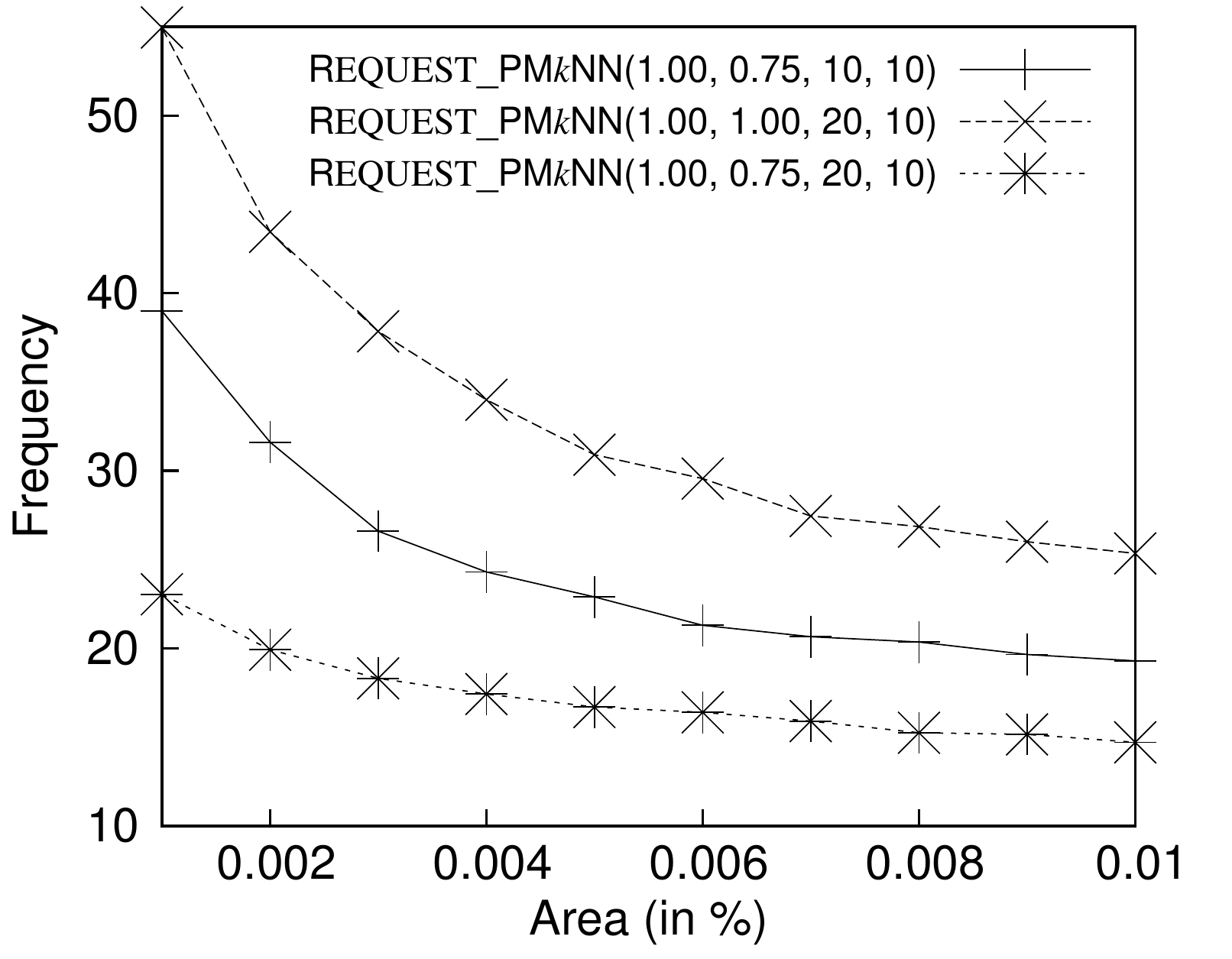}} &
      \hspace{-5mm}
      \resizebox{43mm}{!}{\includegraphics{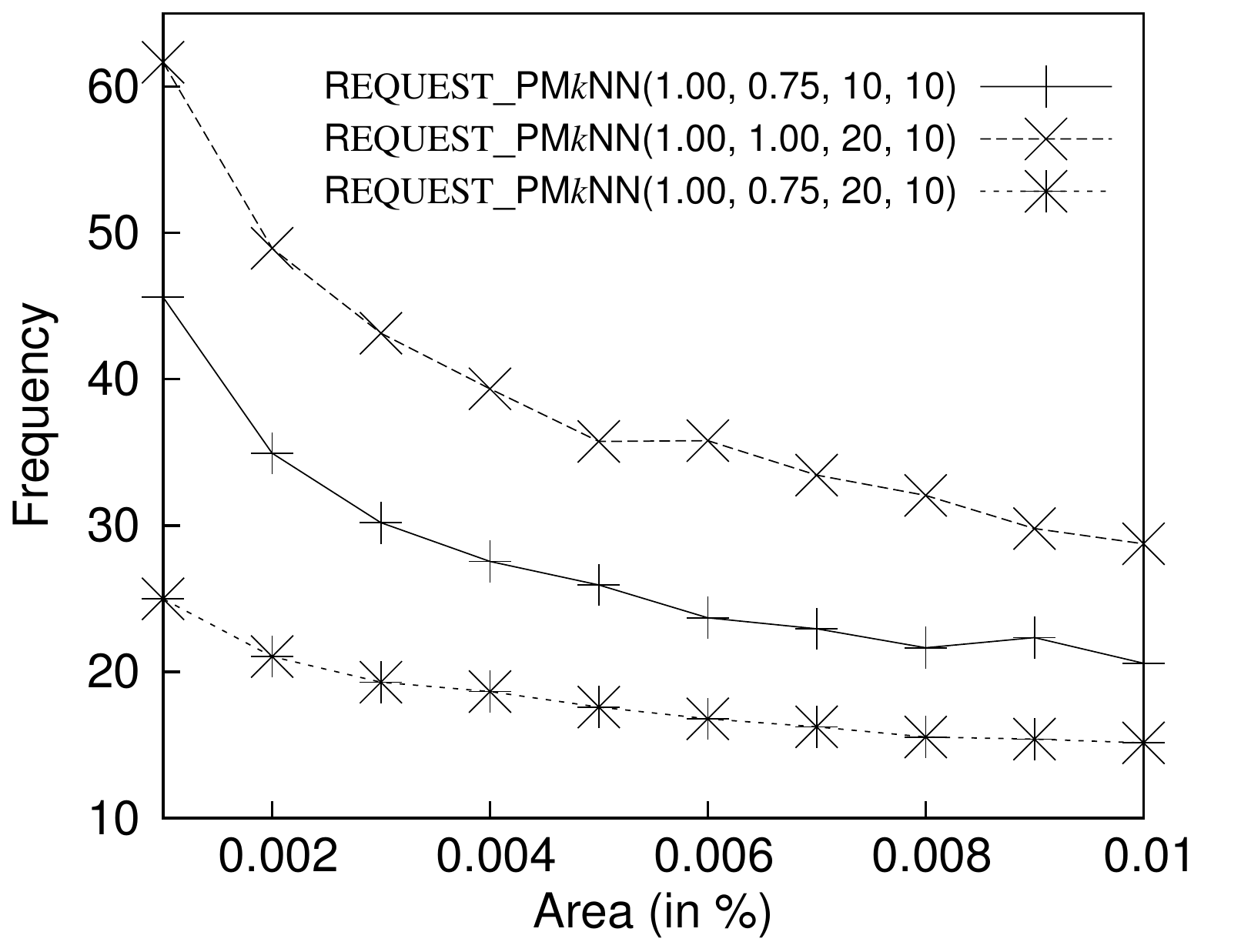}} &
      \hspace{-5mm}
      \resizebox{43mm}{!}{\includegraphics{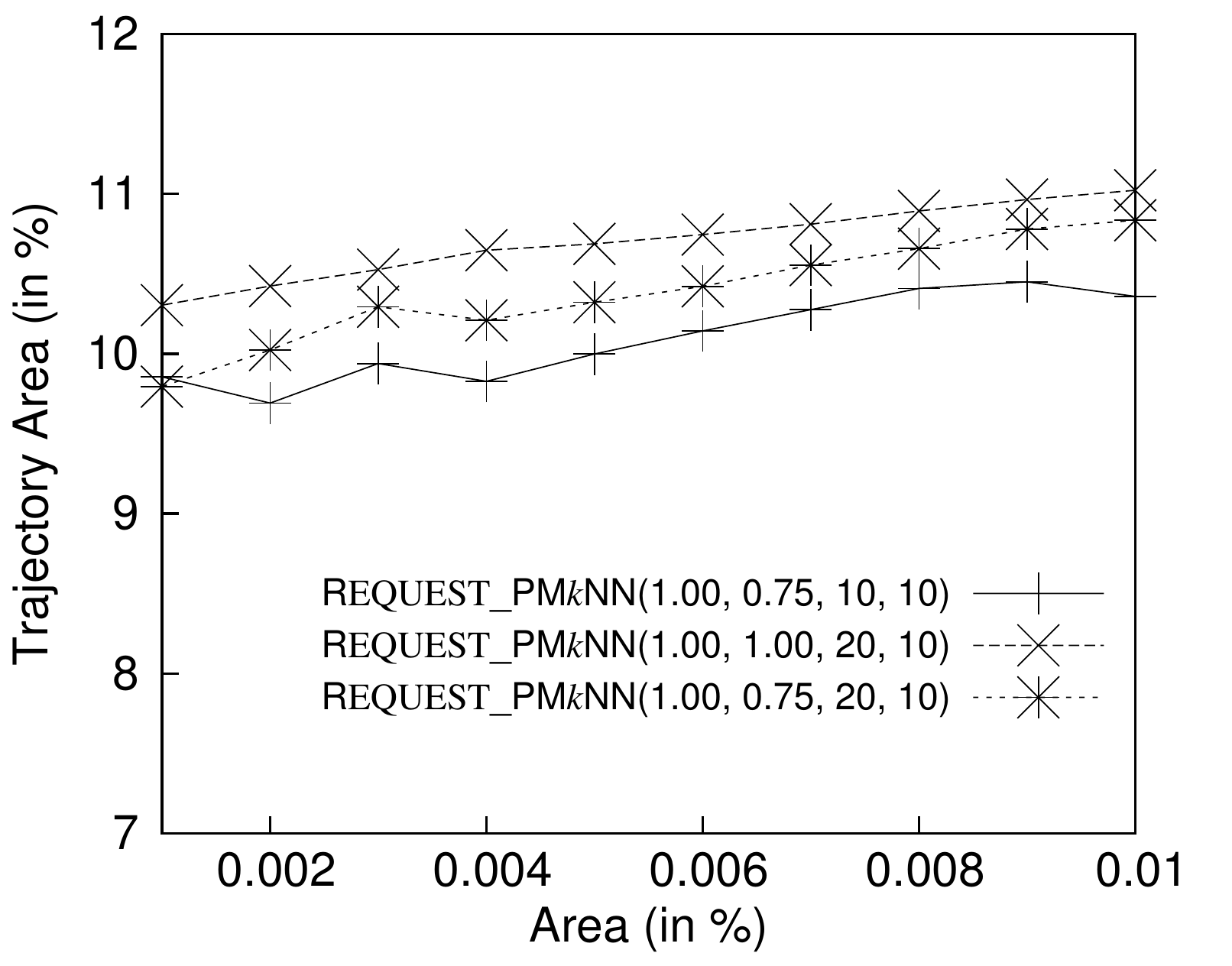}} &
        \hspace{-5mm}
      \resizebox{43mm}{!}{\includegraphics{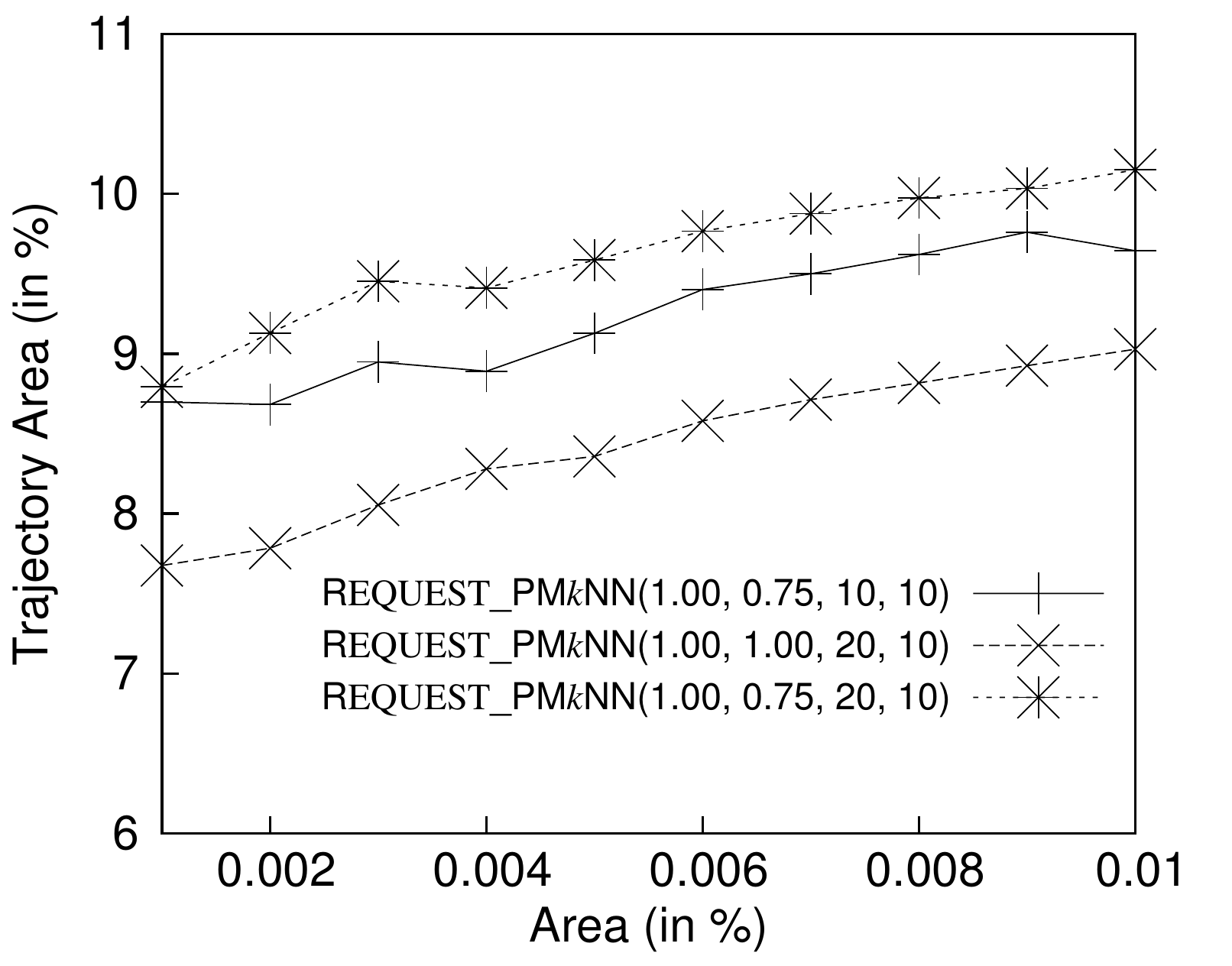}} \\
       \scriptsize{(a) Overlapping Rectangle Attack\hspace{0mm}} & \scriptsize{(b) Combined Attack} &
       \scriptsize{(c) Overlapping Rectangle Attack} & \scriptsize{(d) Combined Attack}
        \end{tabular}
    \caption{The effect of the obfuscation rectangle area on the level of trajectory privacy for the California data set}
    \label{fig:m_freq_va}
  \end{center}
\end{figure*}

\subsubsection{The effect of obfuscation rectangle area}
\label{sec:exp_cont_qarea}

In this set of experiments, we evaluate the effect of obfuscation
rectangle area on the three privacy protection options for our
algorithm \textsc{Request\_PM$k$NN}. In the first option, the user
sacrifices the accuracy of answers to achieve trajectory privacy.
Using this option, the user's required confidence level is lower
than 1 and the user specifies higher confidence level to the LSP
than her required one. In this set of experiments, we represent
the first option for our algorithm
\textsc{Request\_PM$k$NN($cl,cl_r,k,k_r$)} as
\textsc{Request\_PM$k$NN(1,0.75,10,10)}, where the user hides the
required confidence level 0.75 from the LSP, instead specifies 1
for the confidence level. In the second option, the user does not
sacrifice the accuracy of the answers for her trajectory privacy;
instead the user specifies a higher number of data objects to the
LSP than her required one. For the second option, we set the
parameters of \textsc{Request\_PM$k$NN($cl,cl_r,k,k_r$)} as
\textsc{Request\_PM$k$NN(1,1,20,10)}. In the third option, the
user hides both of the required confidence level and the required
number of data objects. Thus, the third option is represented as
\textsc{Request\_PM$k$NN(1,0.75,20,10)}.



We vary the obfuscation rectangle area from 0.001\% to 0.01\% of
the total data space. For all the three options, we observe in
Figures~\ref{fig:m_freq_va}(a) and~\ref{fig:m_freq_va}(b) that the
frequency decreases with the increase of the obfuscation rectangle
area for both overlapping rectangle attack and combined attack,
respectively. On the other hand, Figures~\ref{fig:m_freq_va}(c)
and~\ref{fig:m_freq_va}(d) show that the trajectory area increases
with the increase of the obfuscation rectangle area for
overlapping rectangle attack and combined attack, respectively.
Thus, the larger the obfuscation rectangle area, the higher the
trajectory privacy in terms of both frequency and trajectory area.
This is because the larger the obfuscation rectangle the higher
the probability that the obfuscation rectangle covers a longer
part of a user's trajectory.

\begin{figure}[htbp]
  \begin{center}
    \begin{tabular}{ccc}
      \resizebox{42mm}{!}{\includegraphics{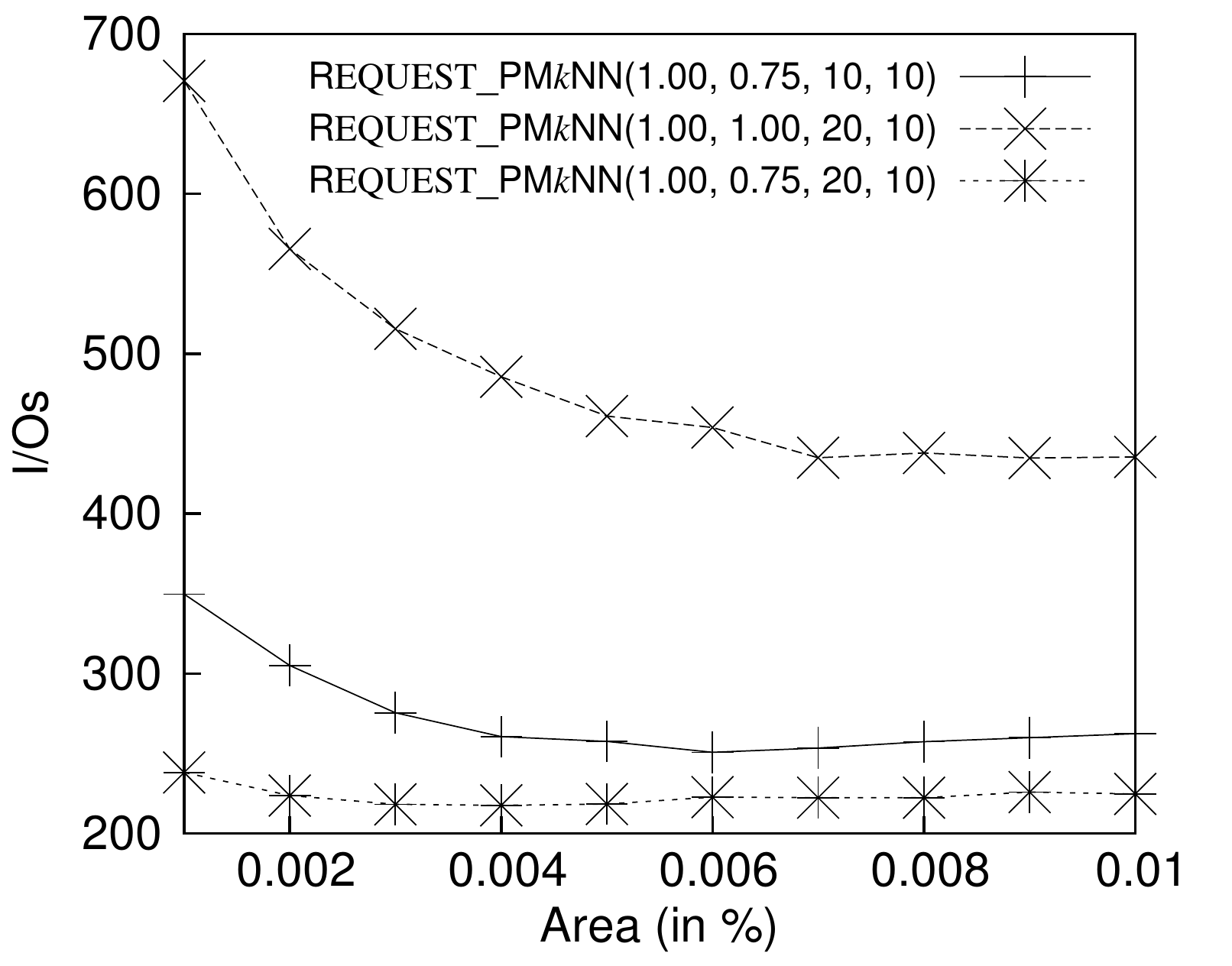}} &
      \resizebox{42mm}{!}{\includegraphics{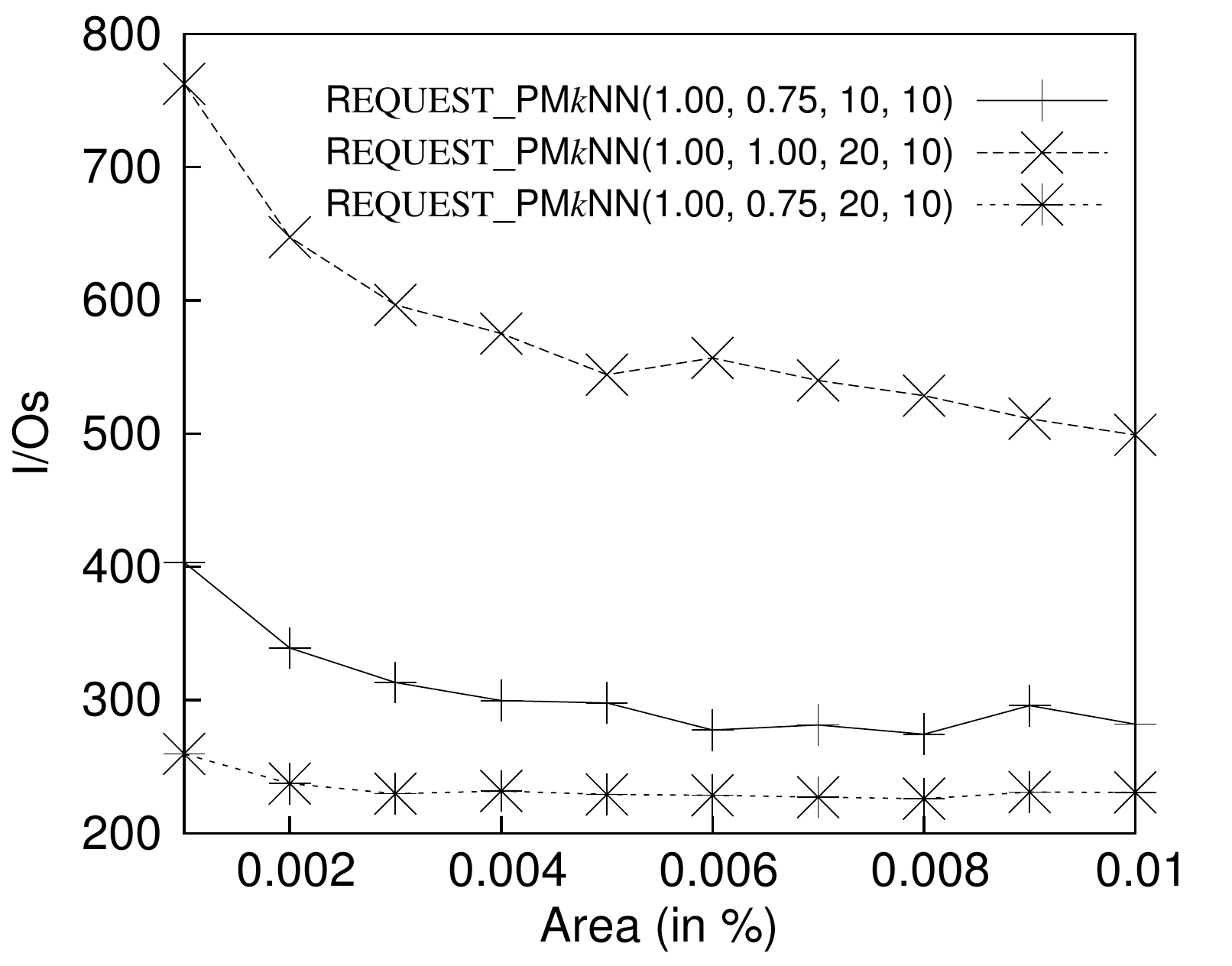}} \\
       \scriptsize{(a) Overlapping Rectangle Attack\hspace{0mm}} & \scriptsize{(b) Combined Attack}
        \end{tabular}
       \begin{tabular}{ccc}
    \resizebox{42mm}{!}{\includegraphics{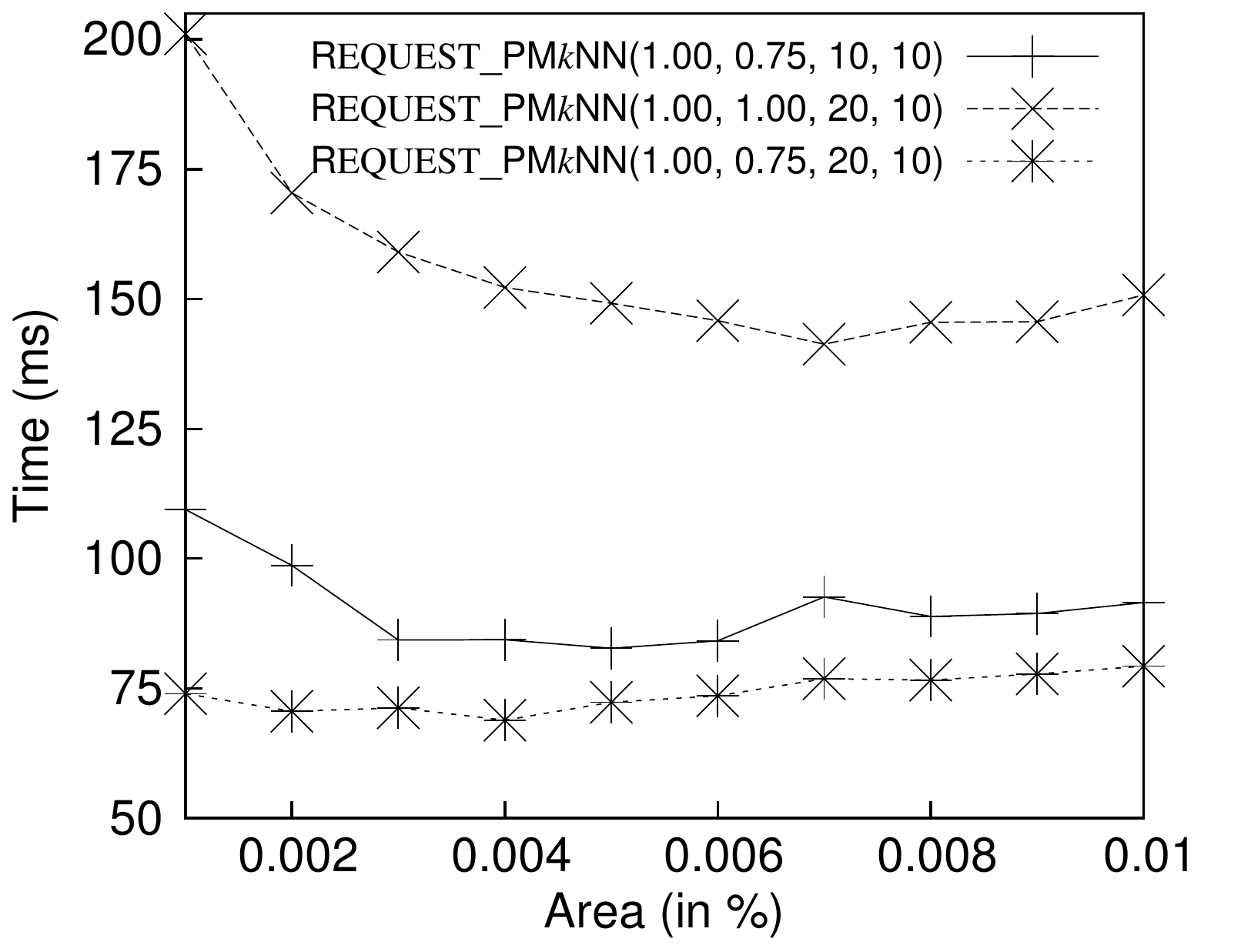}} &
      \resizebox{42mm}{!}{\includegraphics{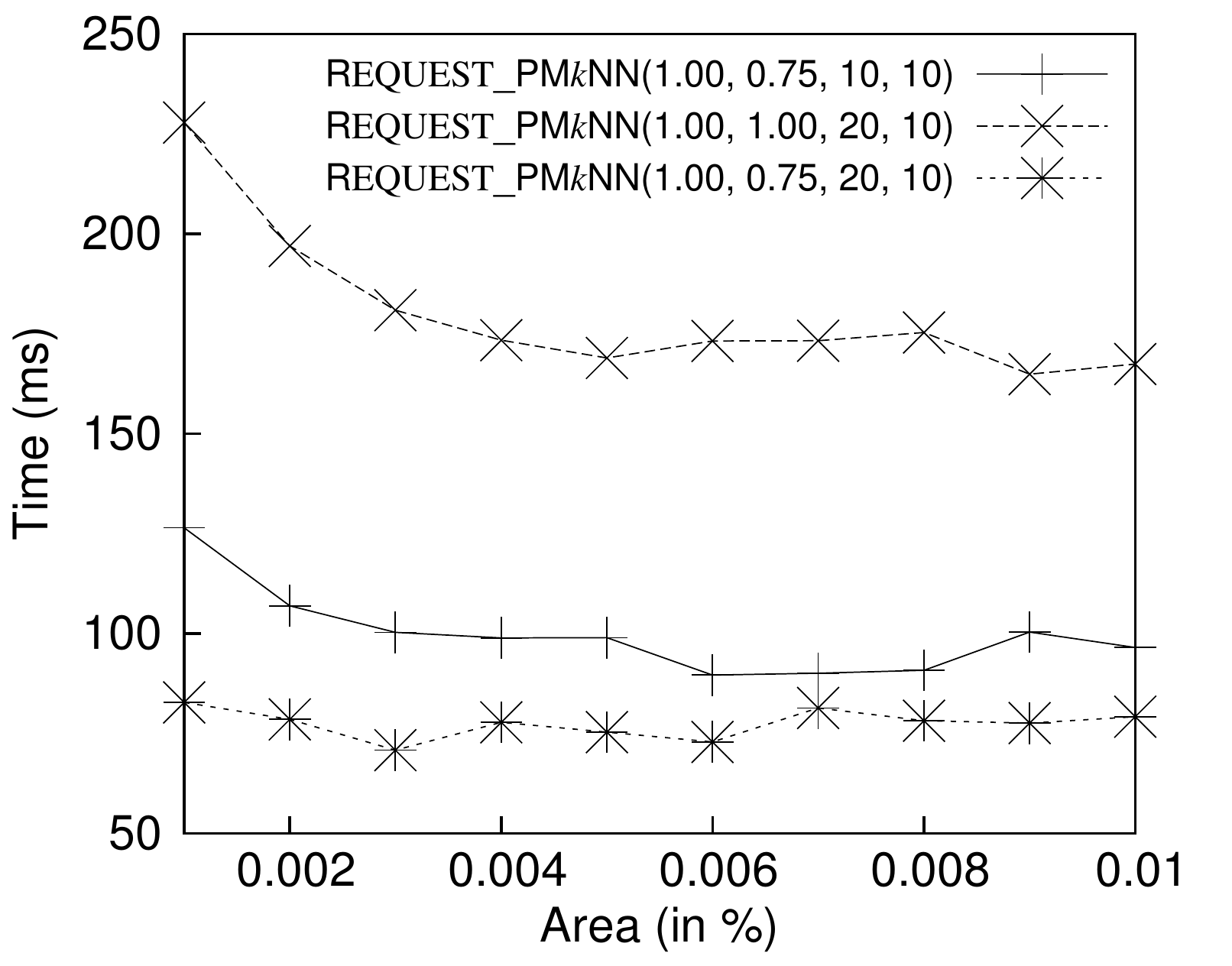}} \\
       \scriptsize{(c) Overlapping Rectangle Attack\hspace{0mm}} & \scriptsize{(d) Combined Attack}
        \end{tabular}
       \begin{tabular}{ccc}
     \resizebox{42mm}{!}{\includegraphics{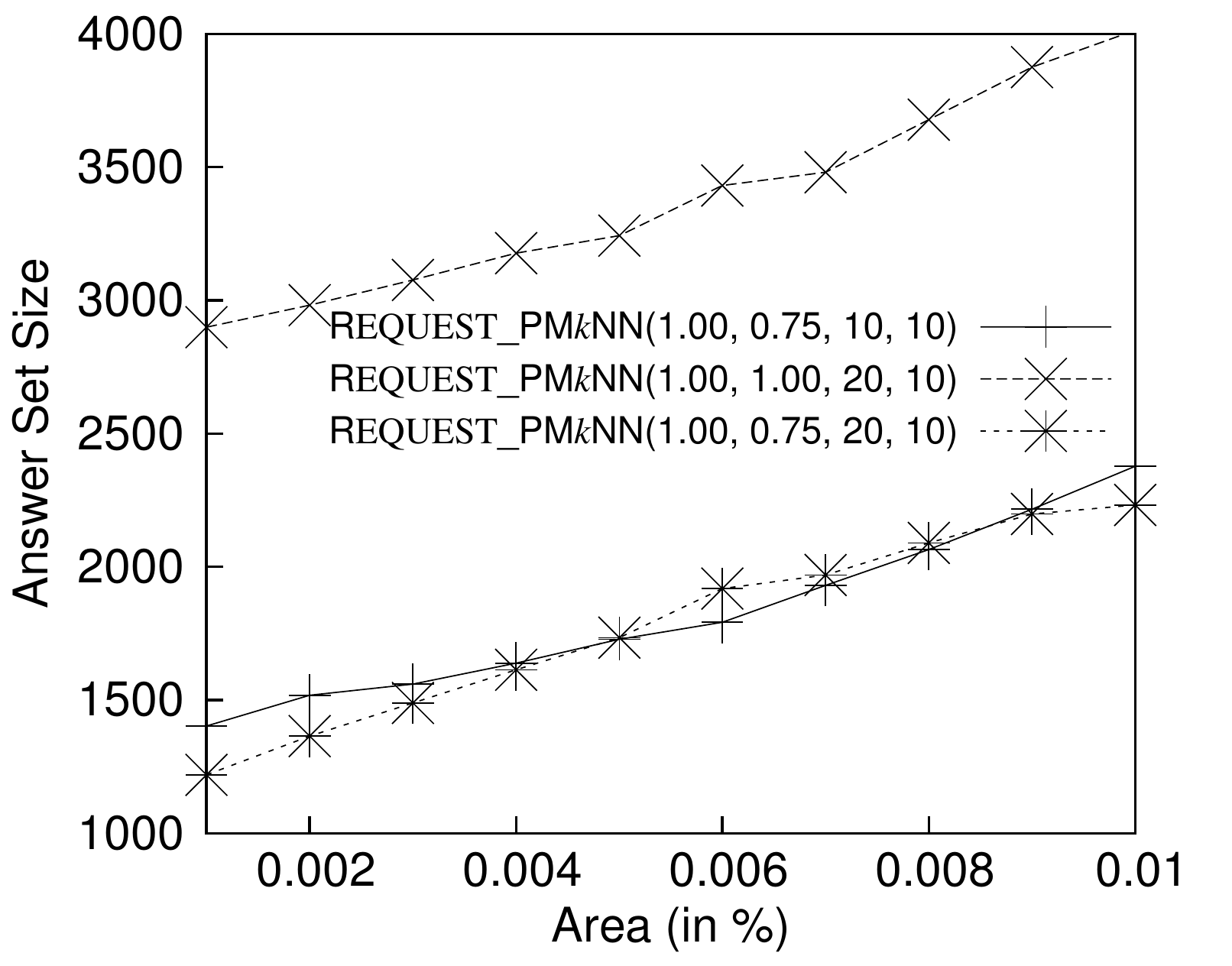}} &
      \resizebox{42mm}{!}{\includegraphics{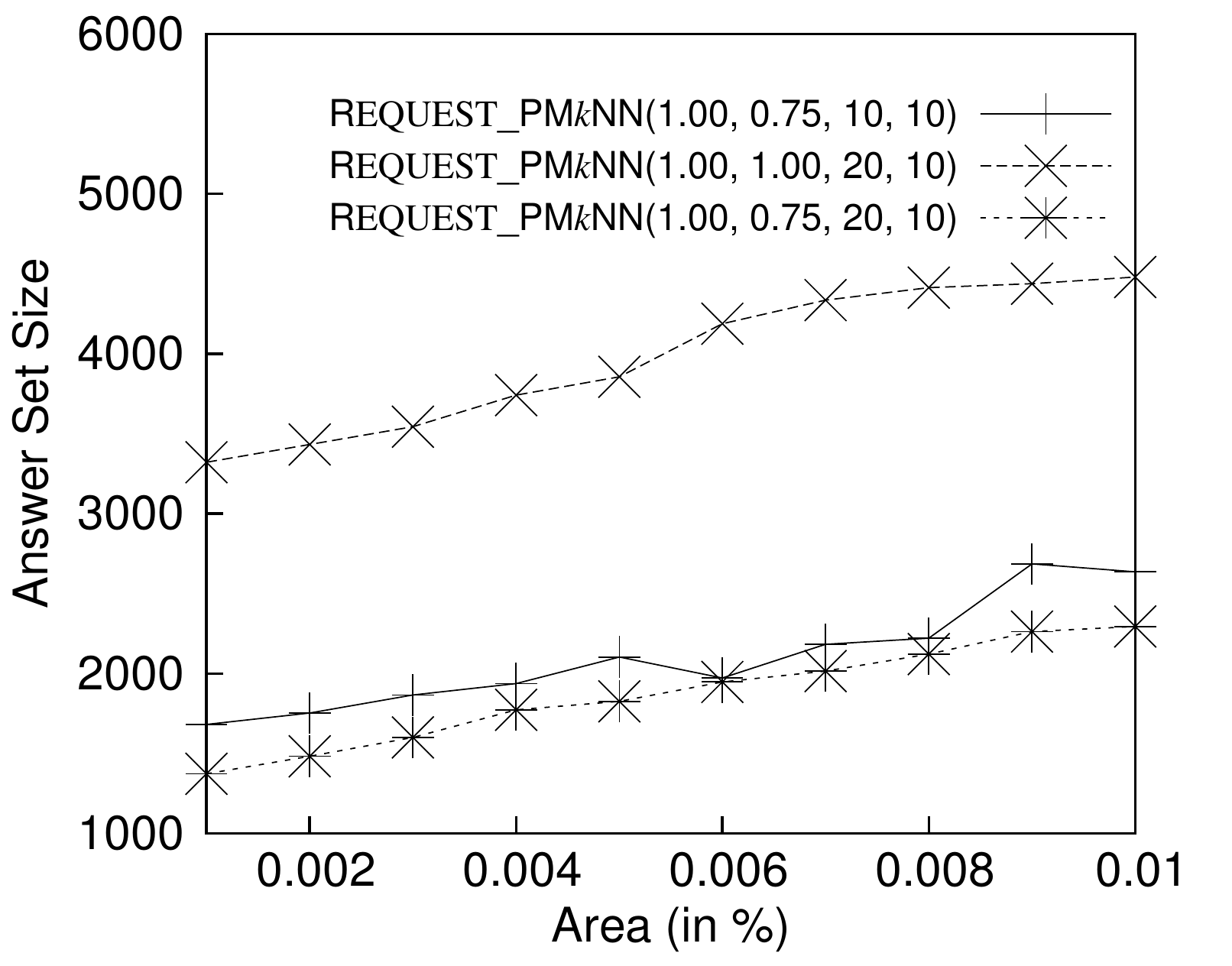}} \\
 \scriptsize{(e) Overlapping Rectangle Attack\hspace{0mm}} & \scriptsize{(f) Combined Attack}
        \end{tabular}
    \caption{The effect of the obfuscation rectangle area on the query processing performance for the California data set}
    \label{fig:m_perf_va}
  \end{center}
\end{figure}

Figures~\ref{fig:m_freq_va}(a) and~\ref{fig:m_freq_va}(b) also
show that the frequency for hiding both confidence level and the
number of NNs is smaller than those for hiding them independently
for any obfuscation rectangle area, since each of them contributes
to extend the $GCR(cl_r,k_r)$. In addition, we observe that the
rate of decrease of frequency with the increase of the obfuscation
rectangle area is more significant for the option of hiding the
confidence level than the option of hiding the number of NNs.

We observe from Figures~\ref{fig:m_freq_va}(a)
and~\ref{fig:m_freq_va}(b) that the frequency in the combined
attack is higher than that of the overlapping rectangle attack.
The underlying cause is as follows. In our algorithm to protect
the overlapping rectangle attack the obfuscation rectangle needs
to be generated inside the current known region. On the other
hand, in case of the combined attack the obfuscation rectangle
needs to be inside the intersection of maximum movement bound and
the known region. Due to the stricter constraints while generating
the obfuscation rectangle to overcome the combined attack, the
frequency becomes higher for the combined attack than that of the
overlapping rectangle attack. For the same reason, the trajectory
area is smaller for the combined attack than that of the
overlapping rectangle attack as shown in
Figures~\ref{fig:m_freq_va}(c) and~\ref{fig:m_freq_va}(d).




In Figures~\ref{fig:m_perf_va}(a)-(d), we observe that both I/Os
and time follow the similar trend of frequency, as expected. On
the other hand, the answer set size shows an increasing trend with
the increase of the obfuscation rectangle area in
Figure~\ref{fig:m_perf_va}(e)-(f). We also run all of these
experiments for other data sets and the results show similar
trends to those of California data set except that of the answer
set size. The different trends of the answer set size may result
from different density and distributions of data objects.

\subsubsection{The effect of $cl_r$ and $cl$}
\label{sec:exp_cont_cl}

In these experiments, we observe the effect of the required and
specified confidence level on the level of trajectory privacy. We
vary the value of the required confidence level and the specified
confidence level from 0.5 to 0.9 and 0.6 to 1, respectively.
\begin{figure*}[htbp]
  \begin{center}
    \begin{tabular}{ccccc}
        \hspace{-5mm}
      \resizebox{43mm}{!}{\includegraphics{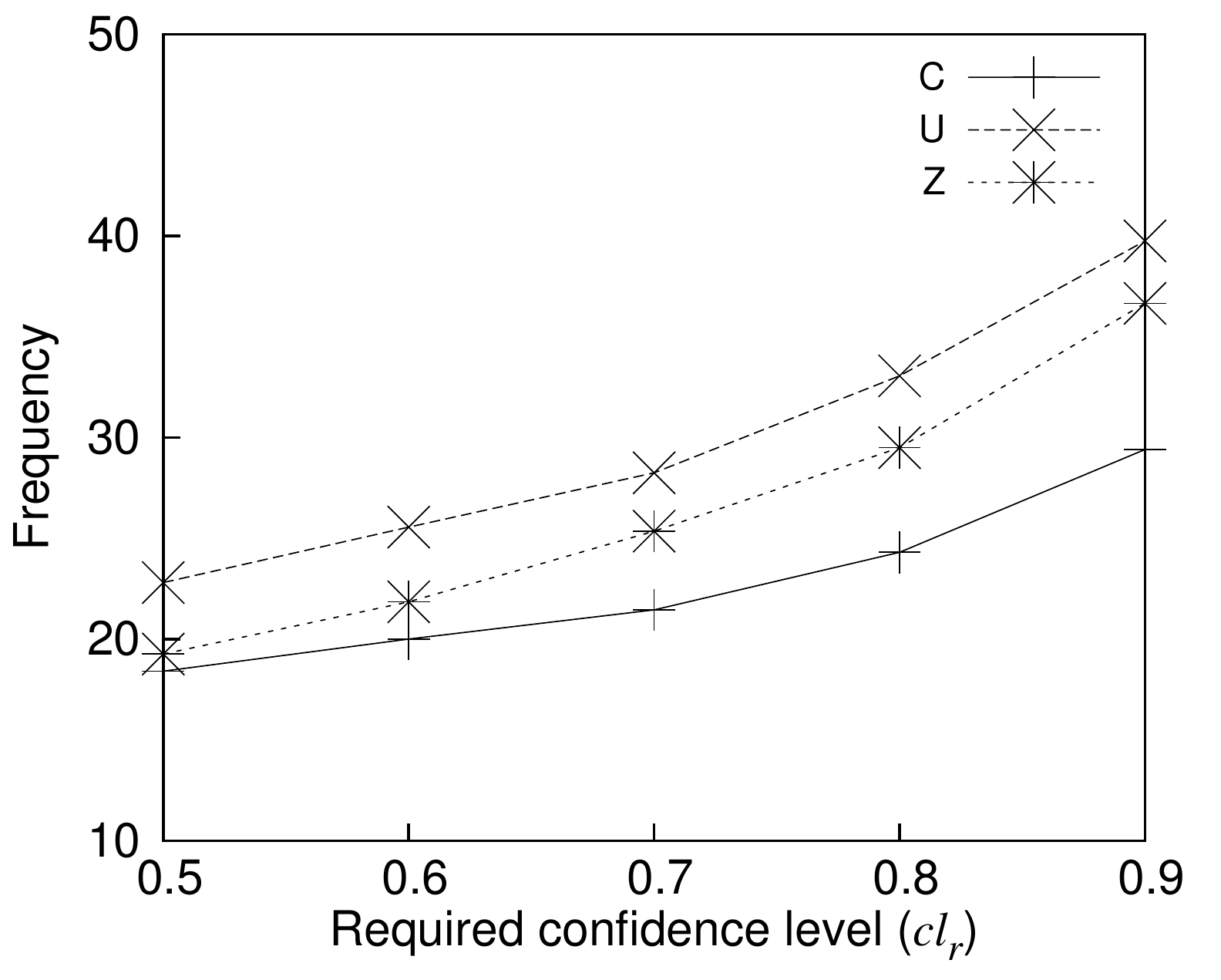}} &
      \hspace{-5mm}
      \resizebox{43mm}{!}{\includegraphics{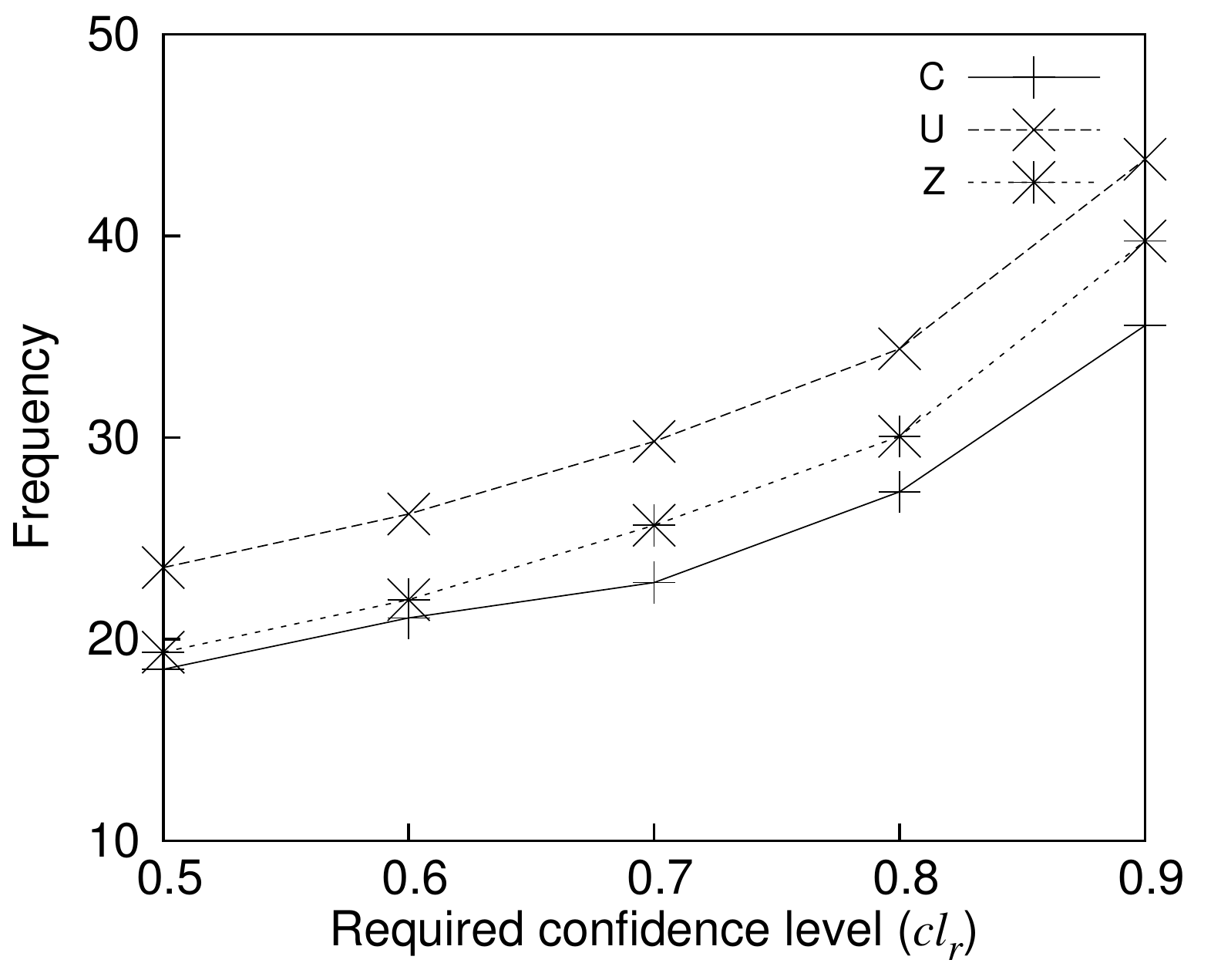}} &
      \hspace{-5mm}
      \resizebox{43mm}{!}{\includegraphics{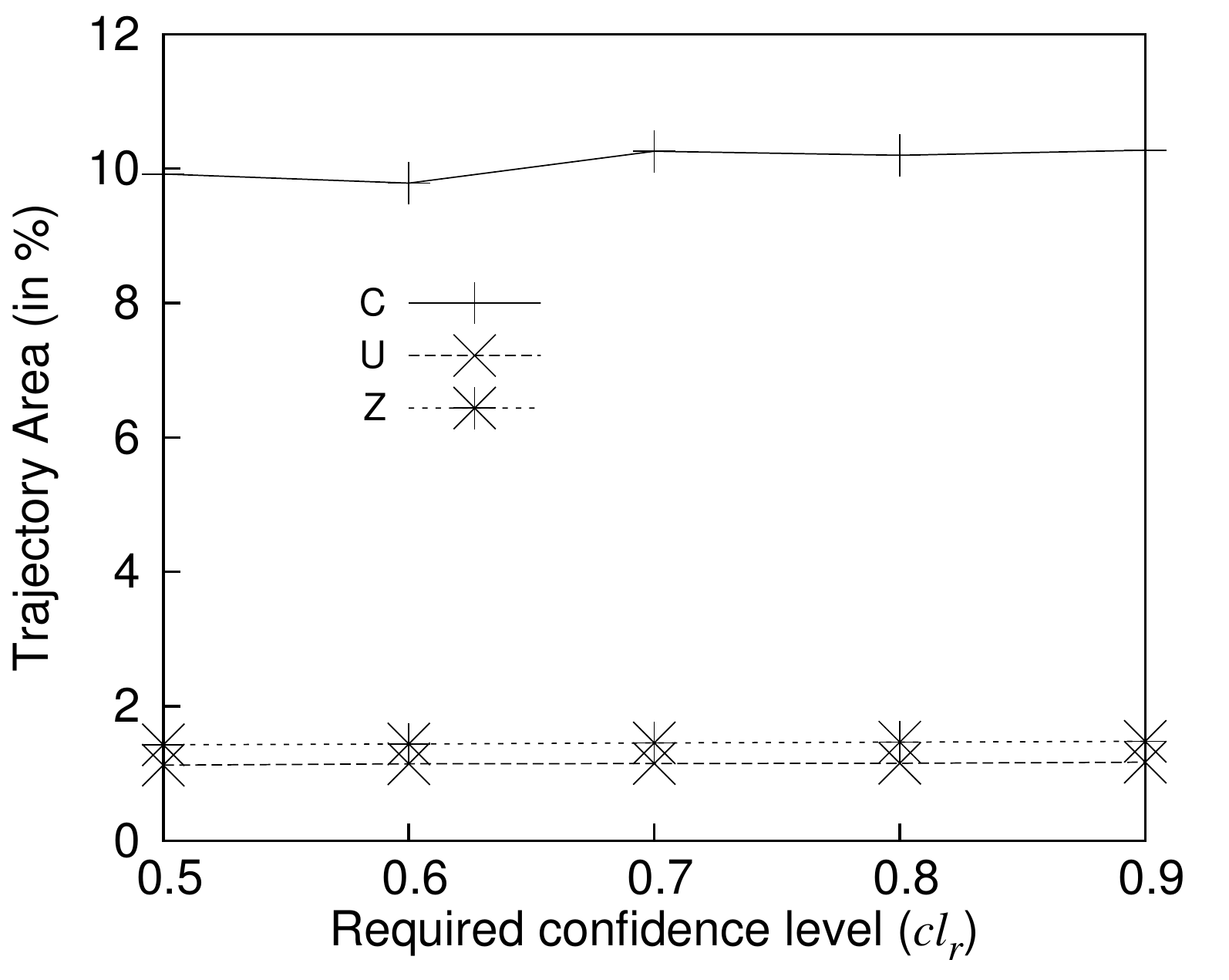}} &
        \hspace{-5mm}
      \resizebox{43mm}{!}{\includegraphics{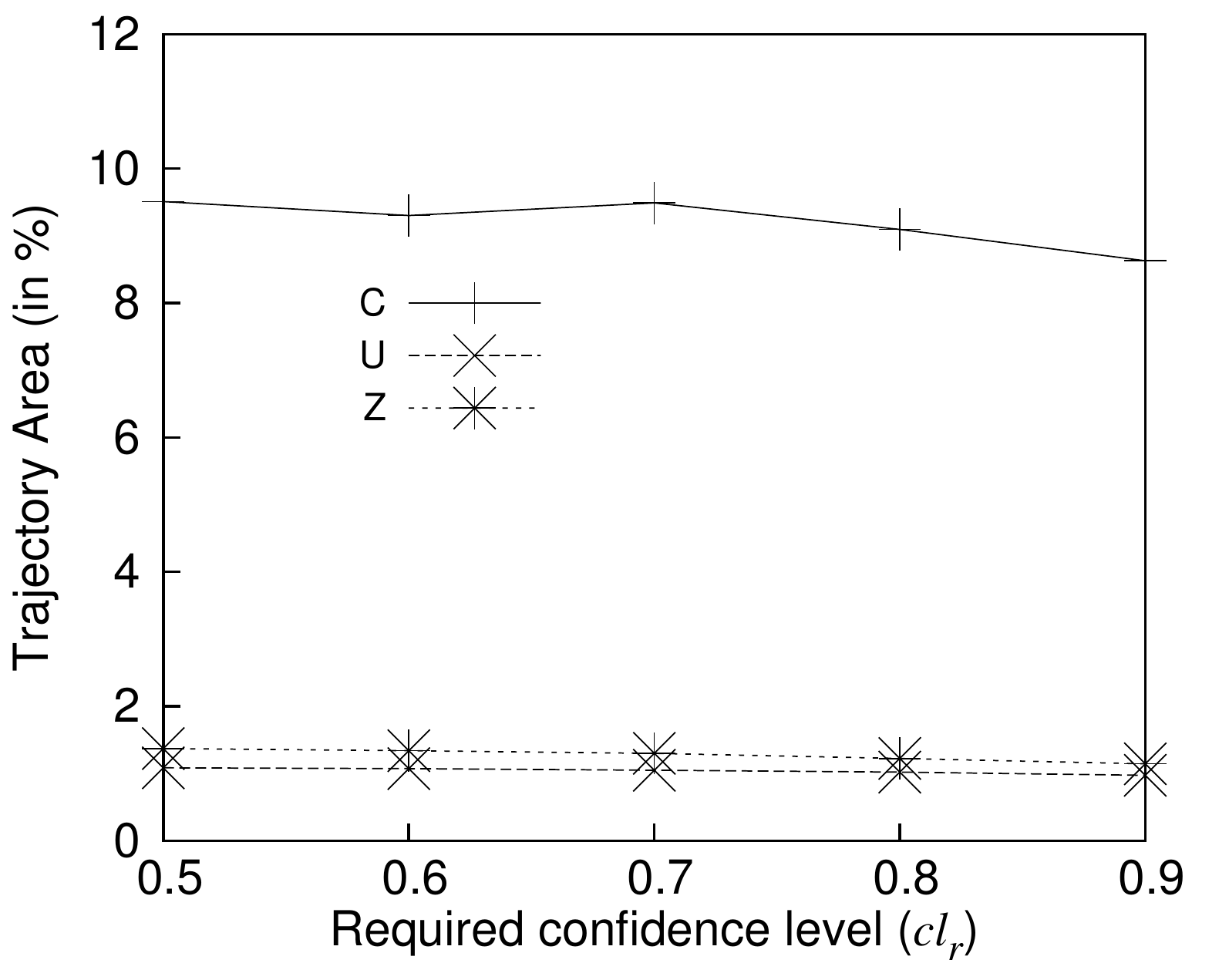}} \\
       \scriptsize{(a) Overlapping Rectangle Attack\hspace{0mm}} & \scriptsize{(b) Combined Attack} &
       \scriptsize{(c) Overlapping Rectangle Attack} & \scriptsize{(d) Combined Attack}
        \end{tabular}
    \caption{The effect of hiding the required confidence level on the level of trajectory privacy}
    \label{fig:m_freq_cl2}
  \end{center}
\end{figure*}

\begin{figure*}[htbp]
  \begin{center}
    \begin{tabular}{ccccc}
        \hspace{-5mm}
      \resizebox{43mm}{!}{\includegraphics{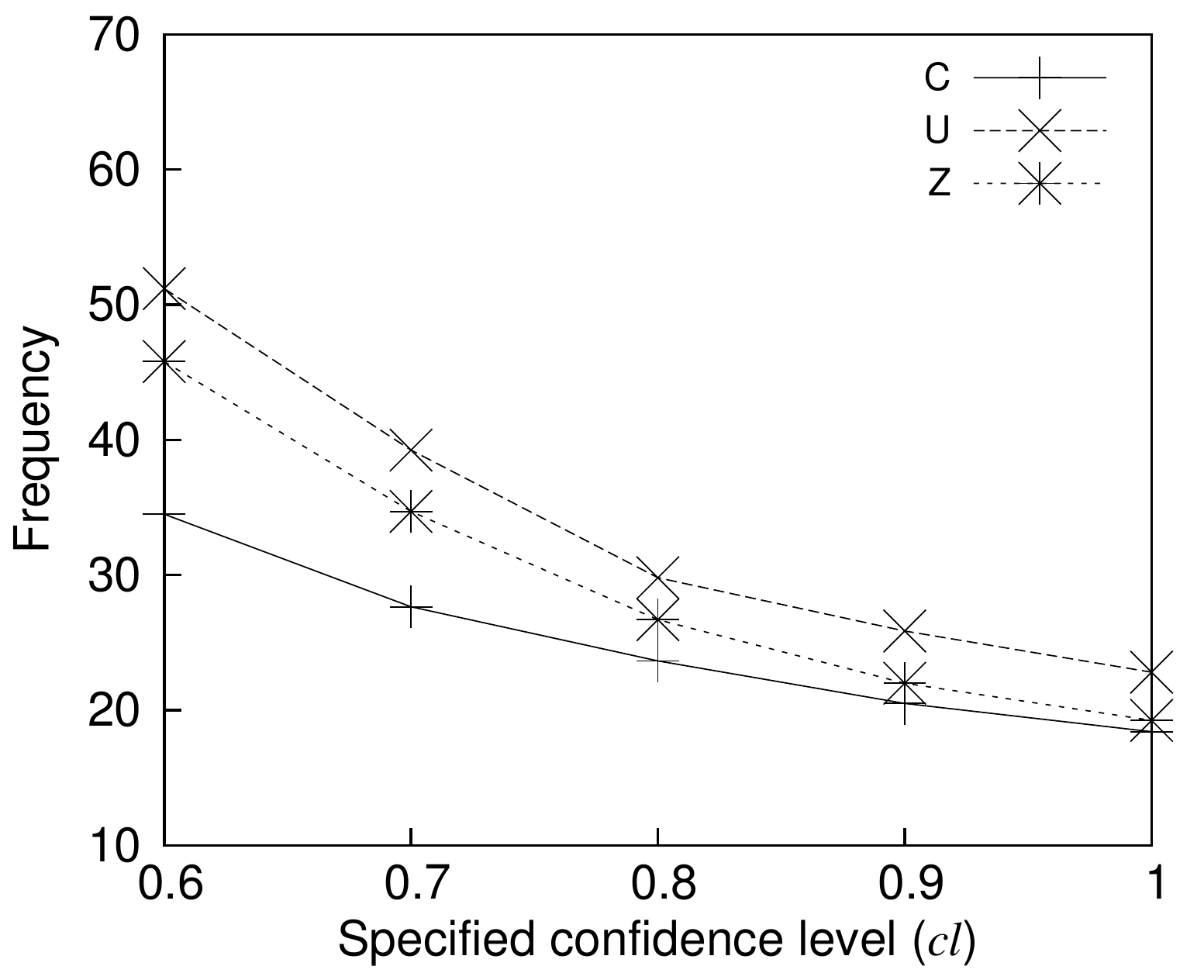}} &
      \hspace{-5mm}
      \resizebox{43mm}{!}{\includegraphics{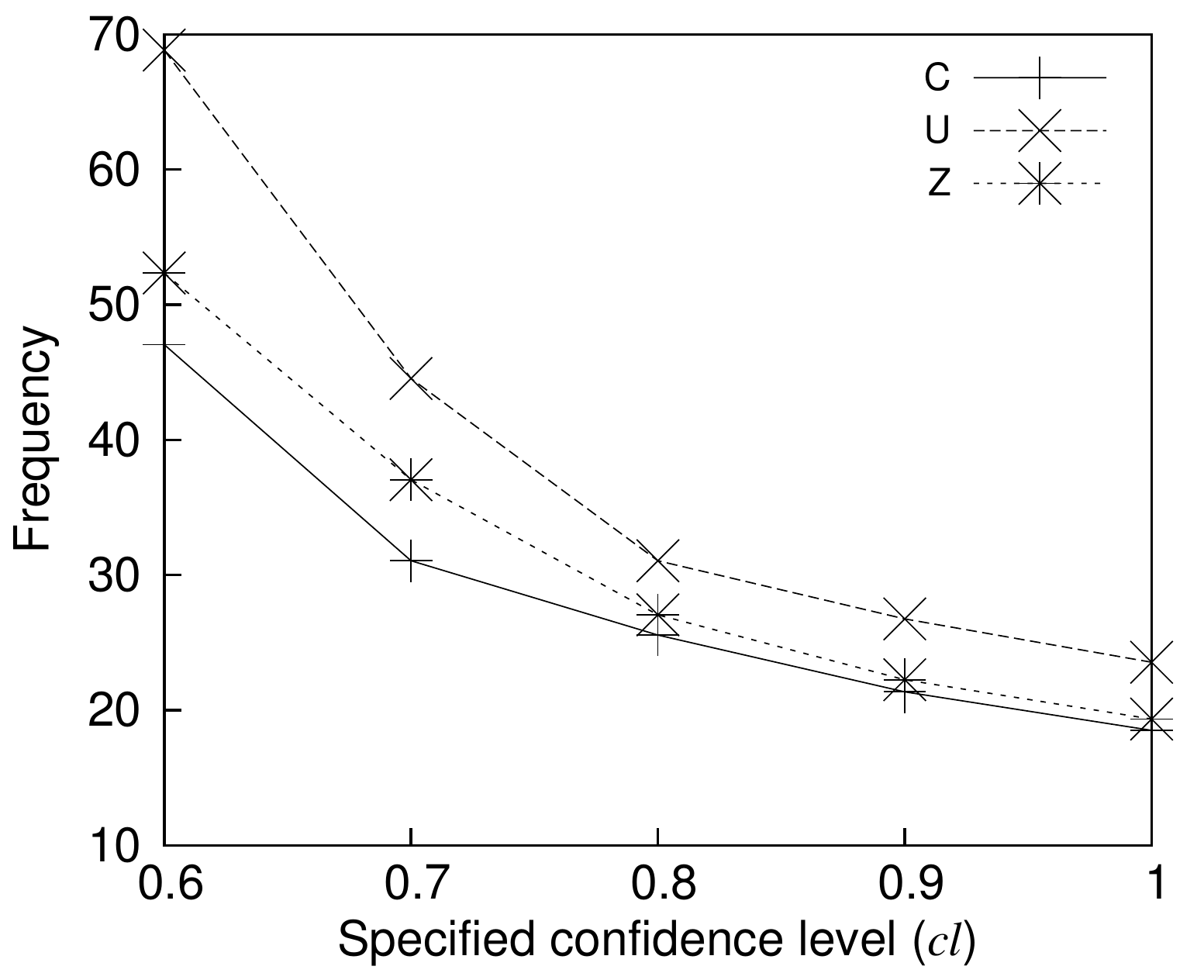}} &
      \hspace{-5mm}
      \resizebox{43mm}{!}{\includegraphics{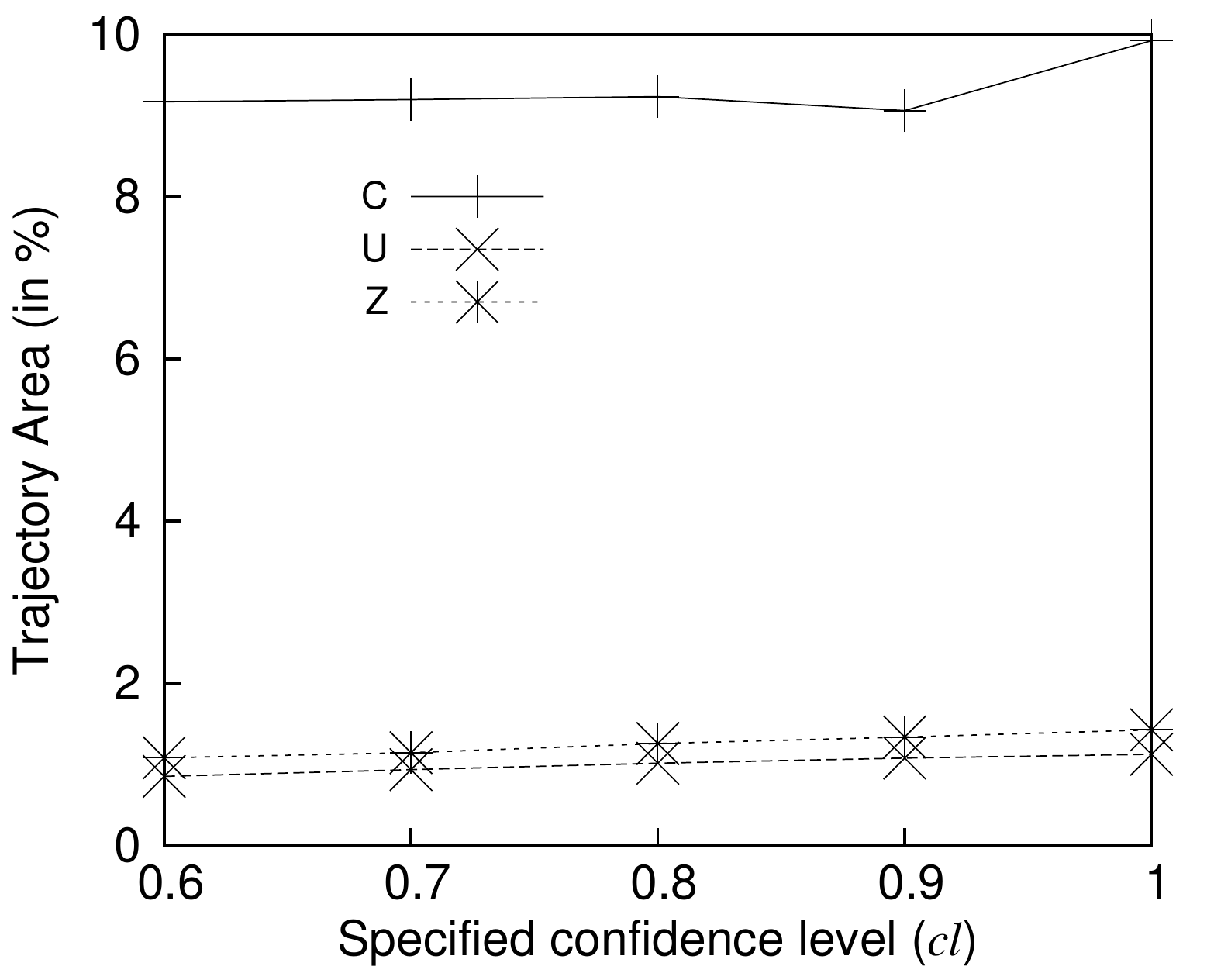}} &
        \hspace{-5mm}
      \resizebox{43mm}{!}{\includegraphics{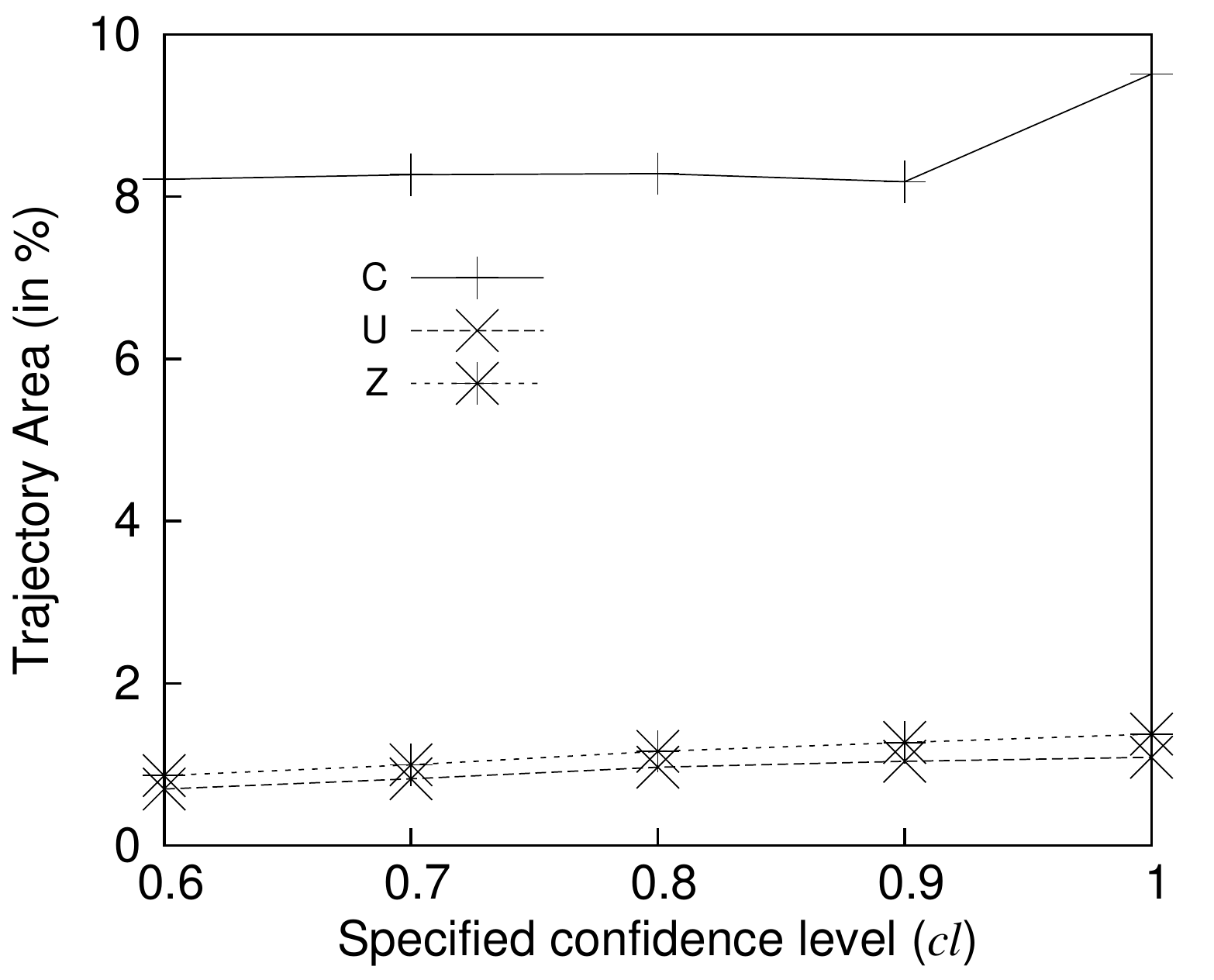}} \\
       \scriptsize{(a) Overlapping Rectangle Attack\hspace{0mm}} & \scriptsize{(b) Combined Attack} &
       \scriptsize{(c) Overlapping Rectangle Attack} & \scriptsize{(d) Combined Attack}
        \end{tabular}
    \caption{The effect of hiding the specified confidence level on the level of trajectory privacy}
    \label{fig:m_freq_cl1}
  \end{center}
\end{figure*}
Figures~\ref{fig:m_freq_cl2}(a)-(b) show that the frequency
increases with the increase of the required confidence level
$cl_{r}$ for a fixed specified confidence level $cl=1$. We observe
that the larger the difference between required and specified
confidence level, the higher the level of trajectory privacy in
terms of the frequency because the larger difference causes the
larger extension of $GCR(cl_r,k_r)$. On the other hand,
Figures~\ref{fig:m_freq_cl2}(c)-(d) show that the trajectory area
almost remain constant for different $cl_{r}$ as $cl$ remains
fixed.

Figure~\ref{fig:m_freq_cl1}(a)-(b)) shows that the frequency
decreases with the increase of the specified confidence level $cl$
for a fixed required confidence level $cl_r=0.5$. With the
increase of $cl$, for a fixed $cl_r$, the extension of
$GCR(cl_r,k_r)$ becomes larger and the level of trajectory privacy
in terms of frequency increases. On the other hand,
Figures~\ref{fig:m_freq_cl1}(c)-(d) show that the trajectory area
increases with the increase of $cl$, as expected.

We observe from Figures~\ref{fig:m_freq_cl2}
and~\ref{fig:m_freq_cl1} that the frequency is higher and the
trajectory area is smaller in case of the combined attack than
those for the case of the overlapping rectangle attack, which is
expected due to stricter constraints in the generation of
obfuscation rectangle in the combined attack than that of the
overlapping rectangle attack.


We also see that a user can achieve a high level of trajectory
privacy in terms of frequency by reducing the value of $cl_r$
slightly. For example, in case of the overlapping rectangle
attack, the average rate of decrease of frequency are 19\% and
10\% for reducing the $cl_{r}$ from 0.9 to 0.8 and from 0.6 to
0.5, respectively, for a fixed $cl=1$. In case of the combined
attack, the average rate of decrease of frequency are 23\% and
11\% for reducing the $cl_{r}$ from 0.9 to 0.8 and from 0.6 to
0.5, respectively, for a fixed $cl=1$. Since the trajectory area
almost remains constant for different $cl_{r}$, and we can
conclude that a user can achieve a high level of trajectory
privacy by sacrificing the accuracy of query answers slightly. On
the other hand, from Figures~\ref{fig:m_freq_cl1}, we can see that
the level of trajectory privacy in terms of both frequency and
trajectory area achieves maximum when the specified confidence
level is set to 1.

Note that the query processing overhead for a PM$k$NN query can be
approximated by multiplying the frequency for that query with the
query processing overhead of single obfuscation rectangle
(Section~\ref{sec:exp_static}).

\begin{figure*}[htbp]
  \begin{center}
    \begin{tabular}{ccccc}
        \hspace{-5mm}
      \resizebox{43mm}{!}{\includegraphics{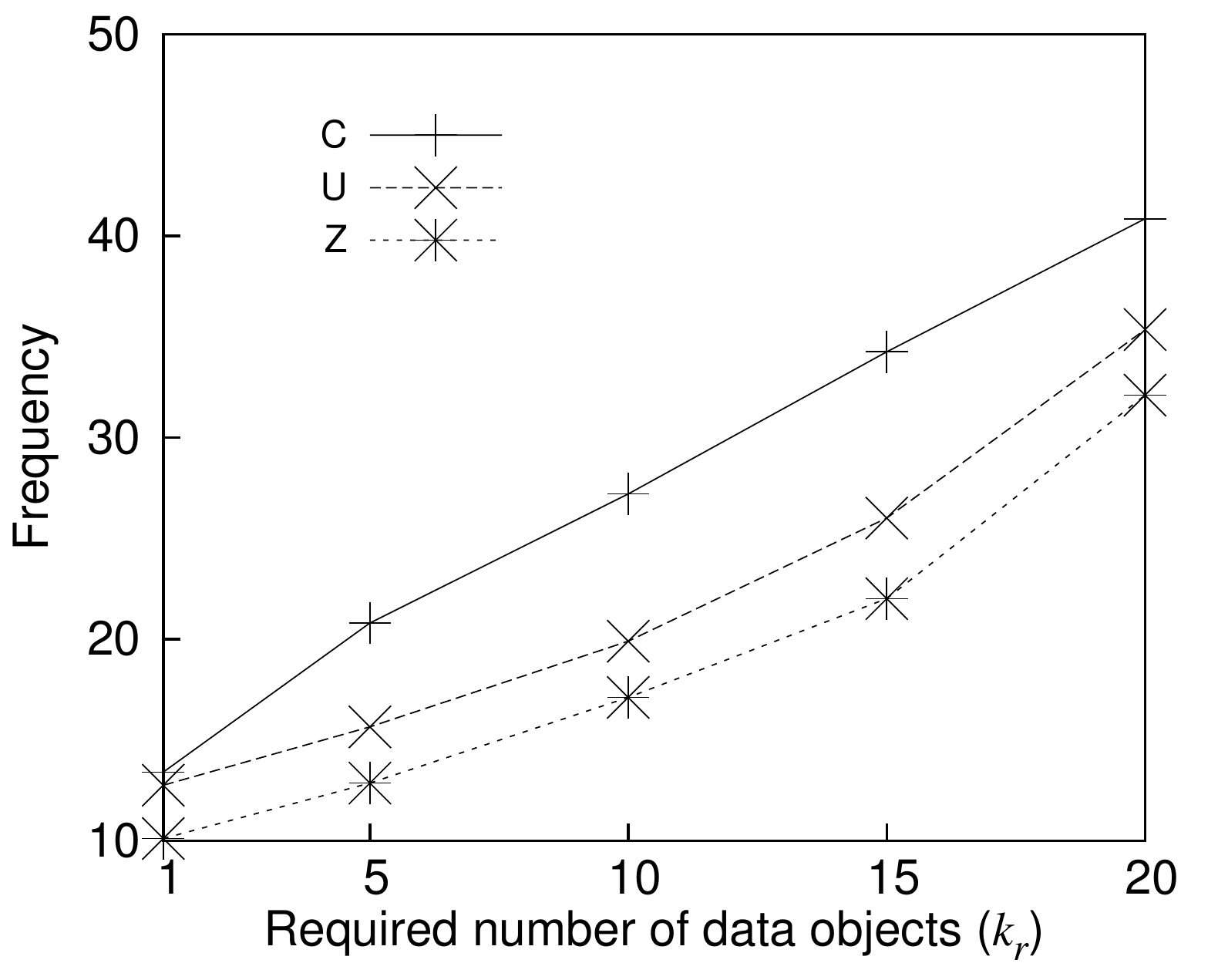}} &
      \hspace{-5mm}
      \resizebox{43mm}{!}{\includegraphics{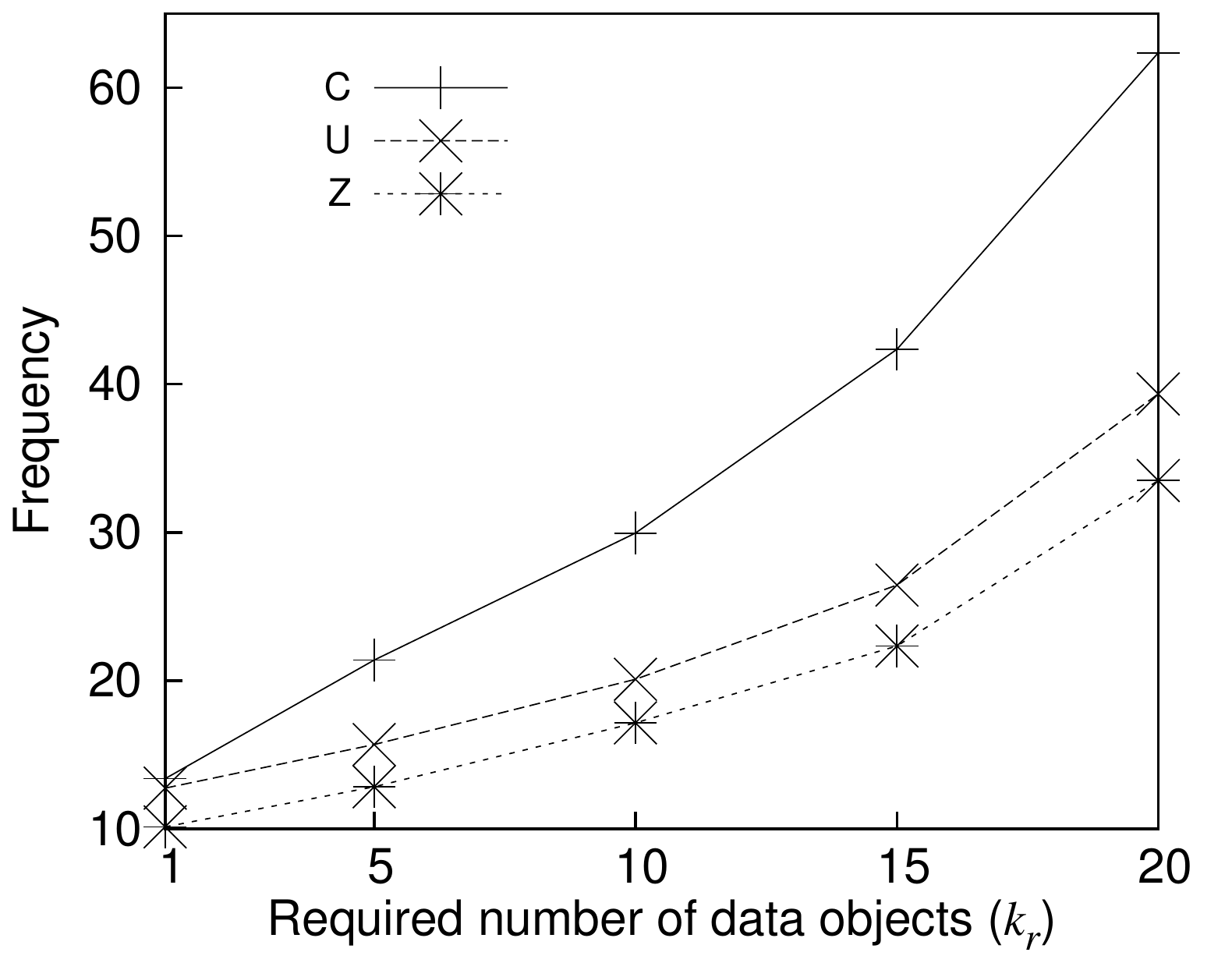}} &
      \hspace{-5mm}
      \resizebox{43mm}{!}{\includegraphics{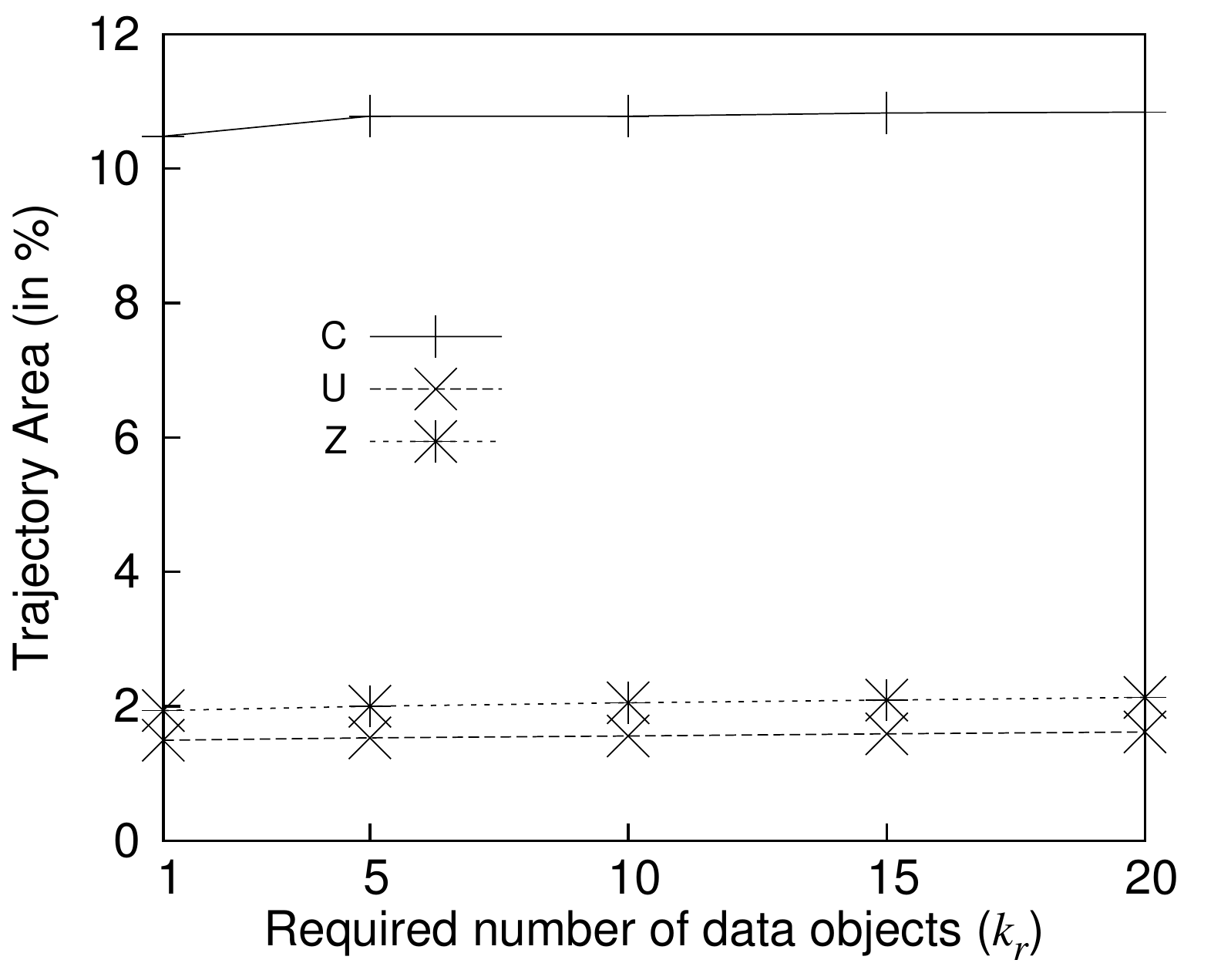}} &
        \hspace{-5mm}
      \resizebox{43mm}{!}{\includegraphics{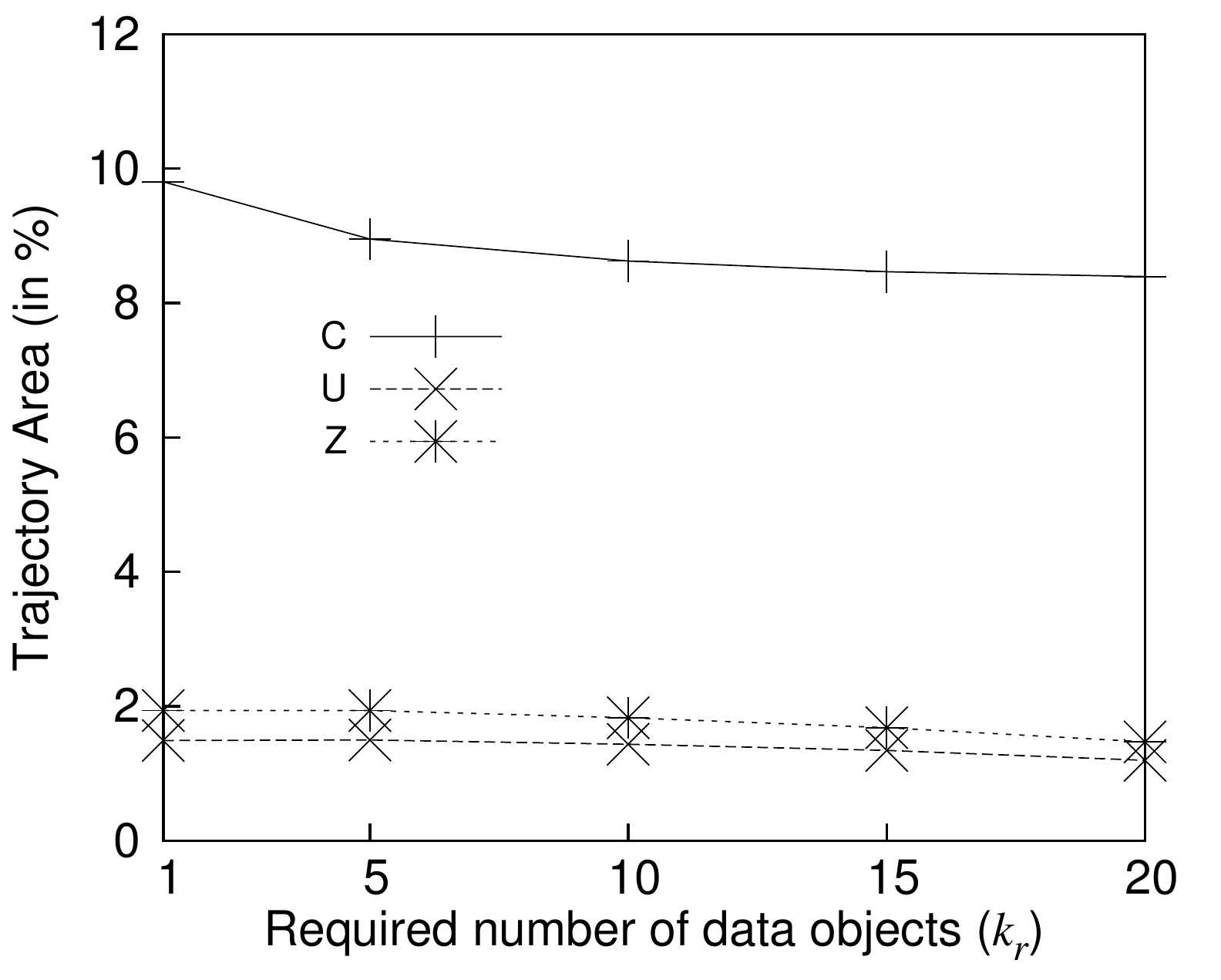}} \\
       \scriptsize{(a) Overlapping Rectangle Attack\hspace{0mm}} & \scriptsize{(b) Combined Attack} &
       \scriptsize{(c) Overlapping Rectangle Attack} & \scriptsize{(d) Combined Attack}
        \end{tabular}
    \caption{The effect of hiding the required number of NNs on the level of trajectory privacy}
    \label{fig:m_freq_k2}
  \end{center}
\end{figure*}

\begin{figure*}[htbp]
  \begin{center}
    \begin{tabular}{ccccc}
        \hspace{-5mm}
      \resizebox{43mm}{!}{\includegraphics{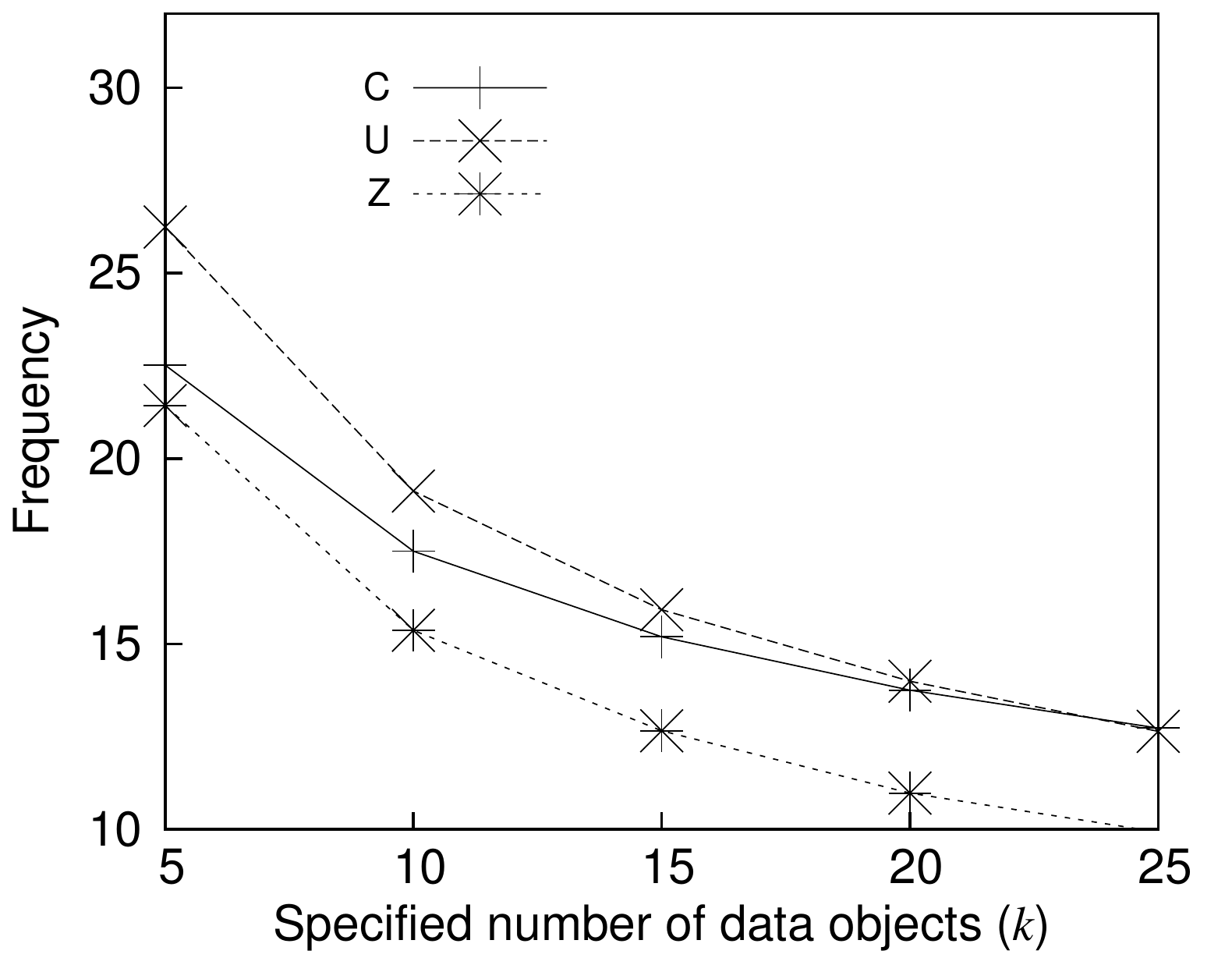}} &
      \hspace{-5mm}
      \resizebox{43mm}{!}{\includegraphics{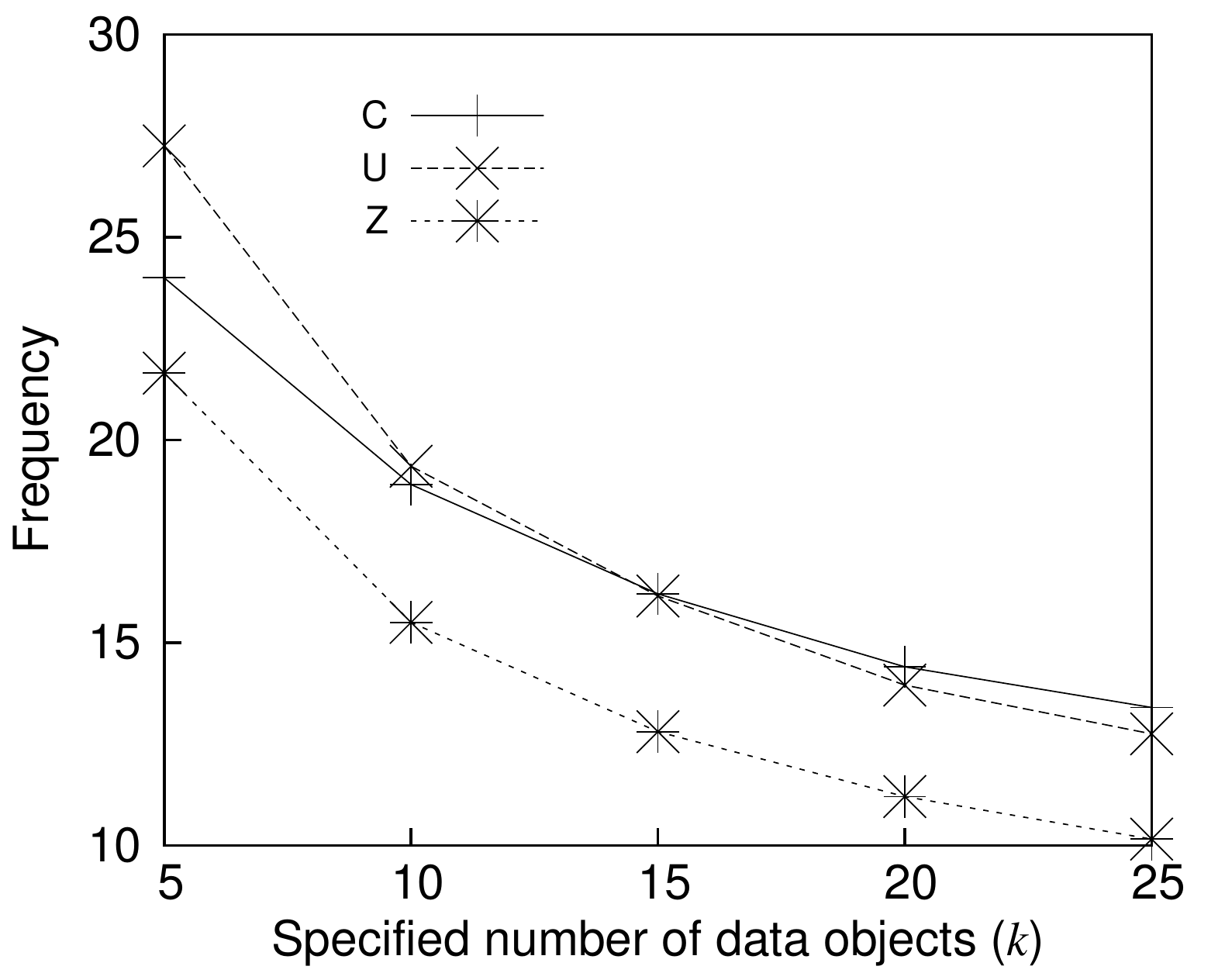}} &
      \hspace{-5mm}
      \resizebox{43mm}{!}{\includegraphics{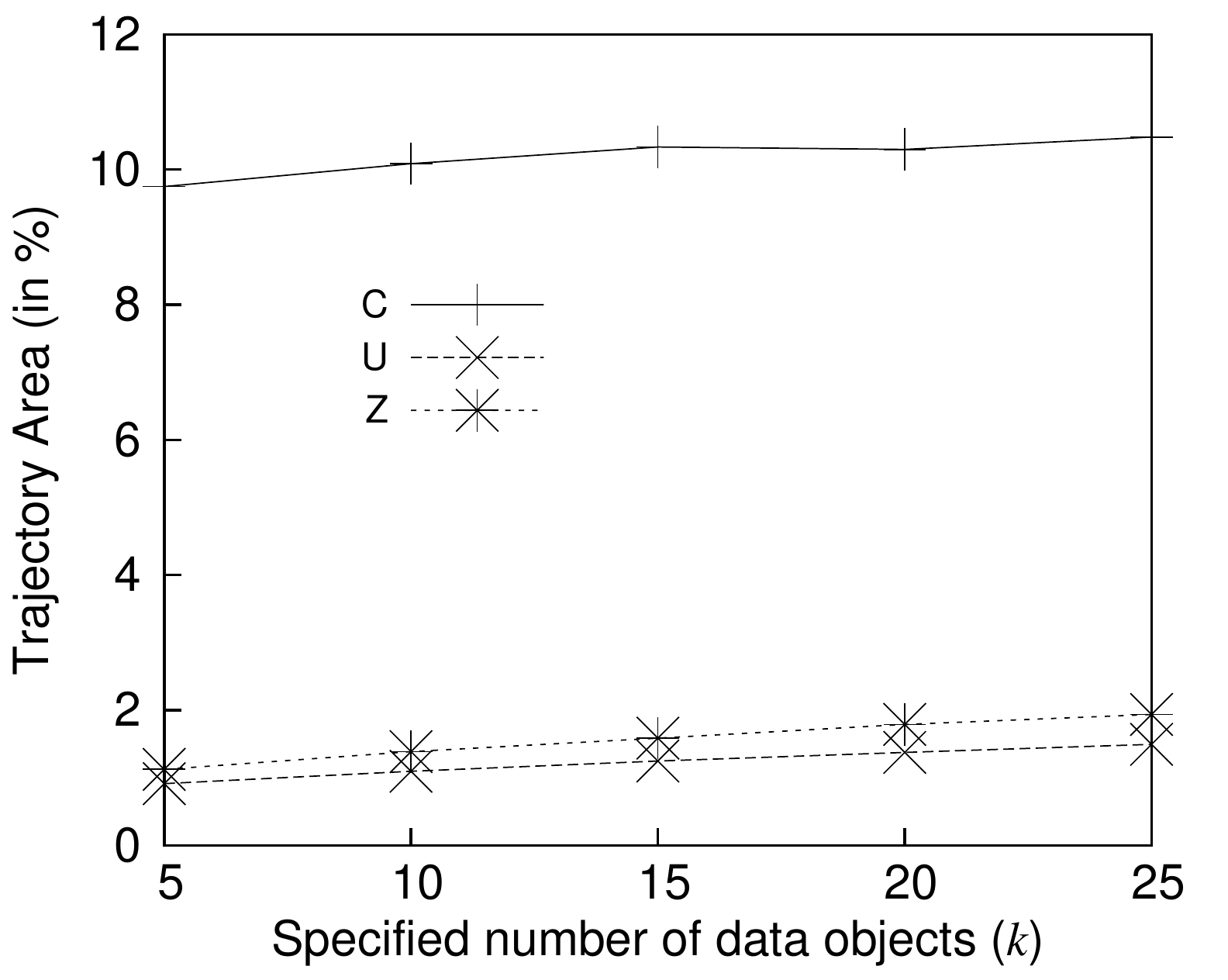}} &
        \hspace{-5mm}
      \resizebox{43mm}{!}{\includegraphics{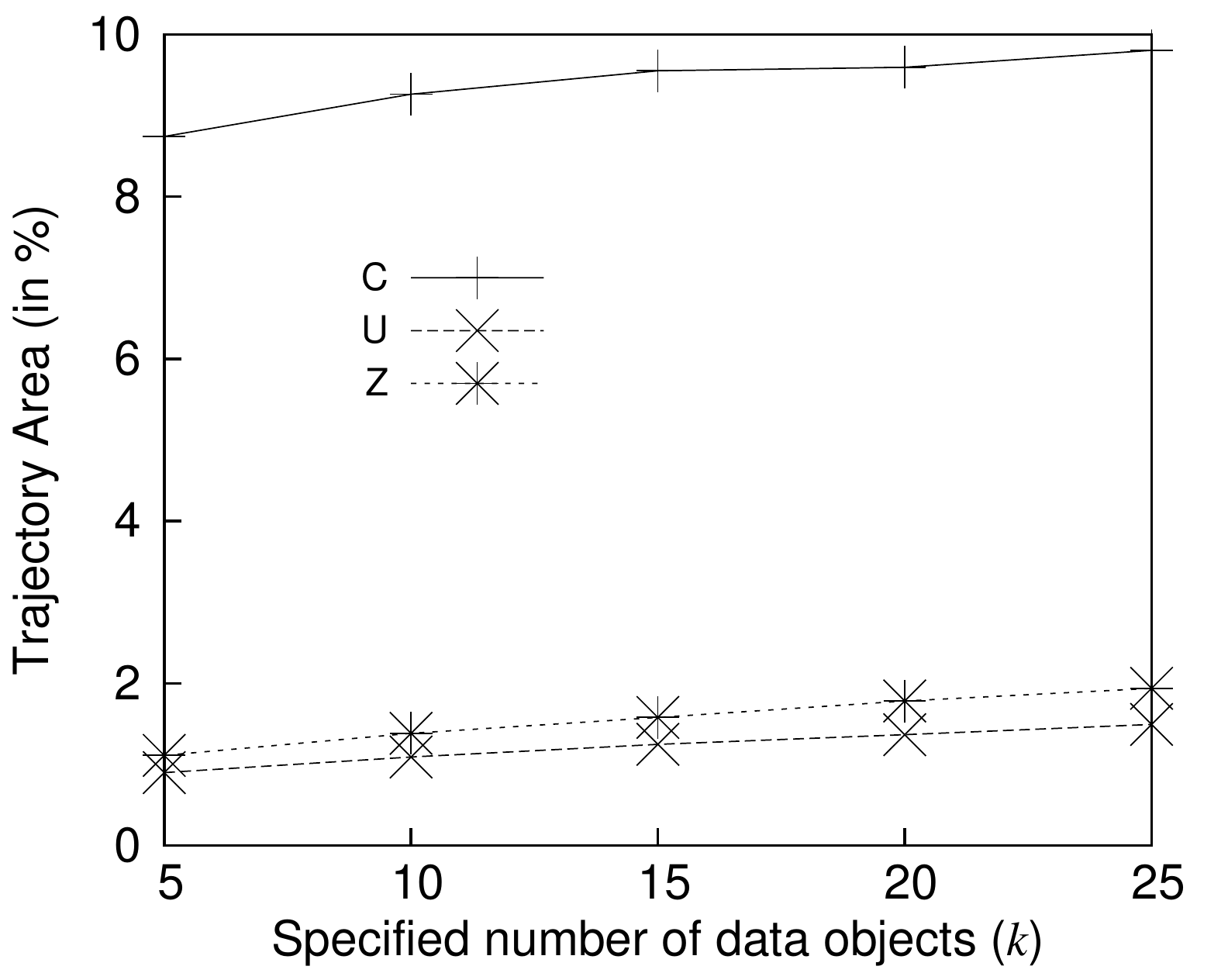}} \\
       \scriptsize{(a) Overlapping Rectangle Attack\hspace{0mm}} & \scriptsize{(b) Combined Attack} &
       \scriptsize{(c) Overlapping Rectangle Attack} & \scriptsize{(d) Combined Attack}
        \end{tabular}
    \caption{The effect of hiding the specified number of NNs on the level of trajectory privacy}
    \label{fig:m_freq_k1}
  \end{center}
\end{figure*}

\subsubsection{The effect of $k_r$ and $k$}
\label{sec:exp_cont_k}

In these experiments, we observe the effect of the required and
the specified number of nearest data objects on the level of
trajectory privacy. We vary the value of the required and the
specified number of nearest data objects from 1 to 20 and 5 to 25,
respectively.

Figures~\ref{fig:m_freq_k2}(a)-(b) show that the frequency
increases with the increase of the required number of nearest data
objects $k_{r}$ for a fixed specified number of nearest data
objects $k=25$. Similar to the case of confidence level, we find
that the larger the difference between required and specified
number of nearest data objects, the higher the level of trajectory
privacy in terms of frequency. On the other hand,
Figures~\ref{fig:m_freq_k2}(c)-(d) show that the trajectory area
almost remains constant for different $k_{r}$.

Figures~\ref{fig:m_freq_k1} show that the frequency decreases and
the trajectory area increases with the increase of $k$ for a fixed
$k_r=1$, which is expected as seen in case of confidence level.

Similar to confidence level, we also observe from
Figures~\ref{fig:m_freq_k2} and~\ref{fig:m_freq_k1} that the
frequency is higher and the trajectory area is smaller in case of
the combined attack than those for the case of the overlapping
rectangle attack.

In Figures~\ref{fig:m_freq_k1}, we also see that the rate of
increase of the level of trajectory privacy in terms of both
frequency and trajectory area decreases with the increase of $k$.
For example, the highest gain in the level of trajectory privacy
for both frequency and trajectory area is achieved when the value
of $k$ is increased from 5 to 10. Thus, we conclude that the value
of $k$ can be set to 10 to achieve a good level of trajectory
privacy for a fixed $k_r = 1$.

%
%

\begin{figure*}[htbp]
  \begin{center}
    \begin{tabular}{ccccc}
        \hspace{-5mm}
      \resizebox{43mm}{!}{\includegraphics{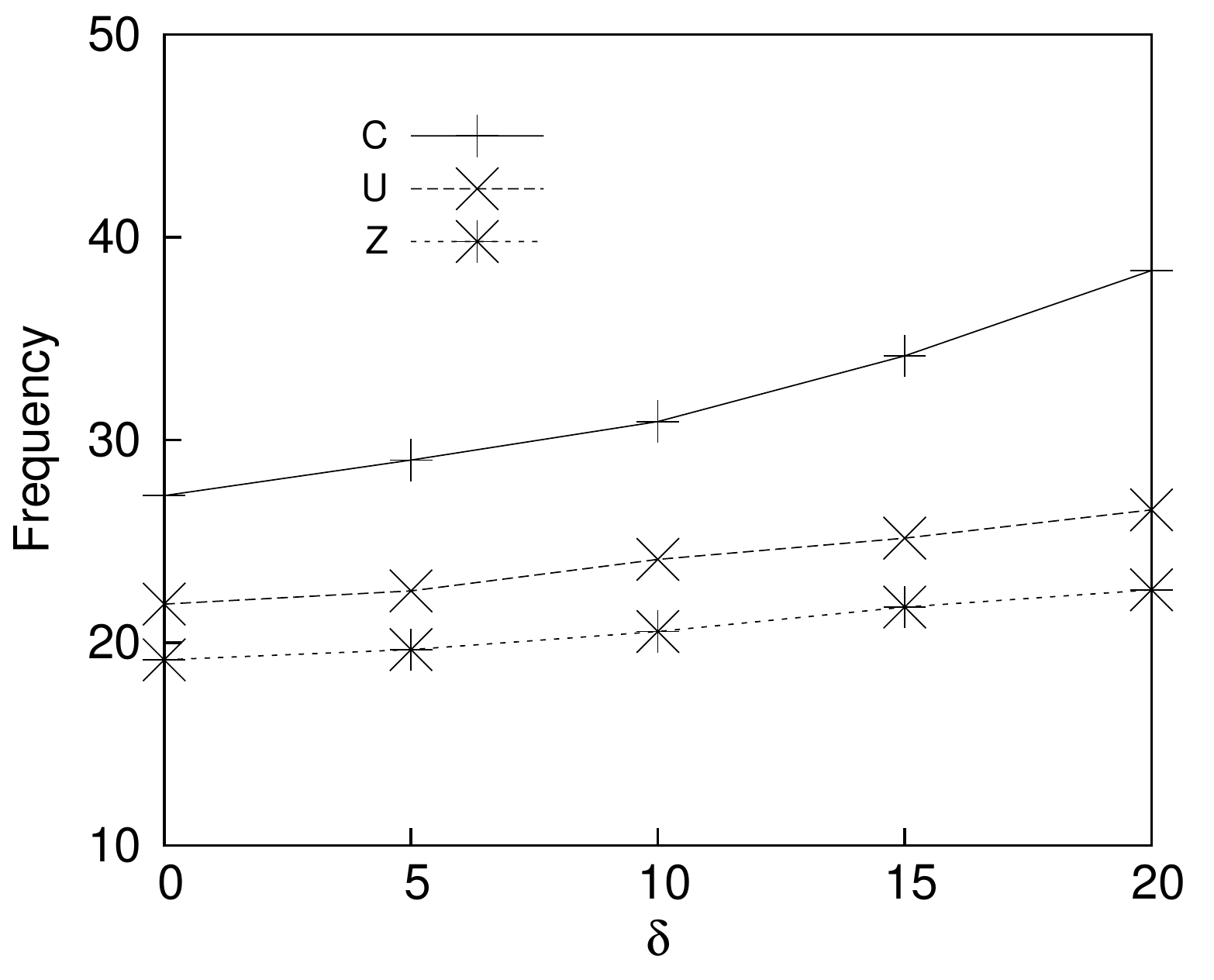}} &
      \hspace{-5mm}
      \resizebox{43mm}{!}{\includegraphics{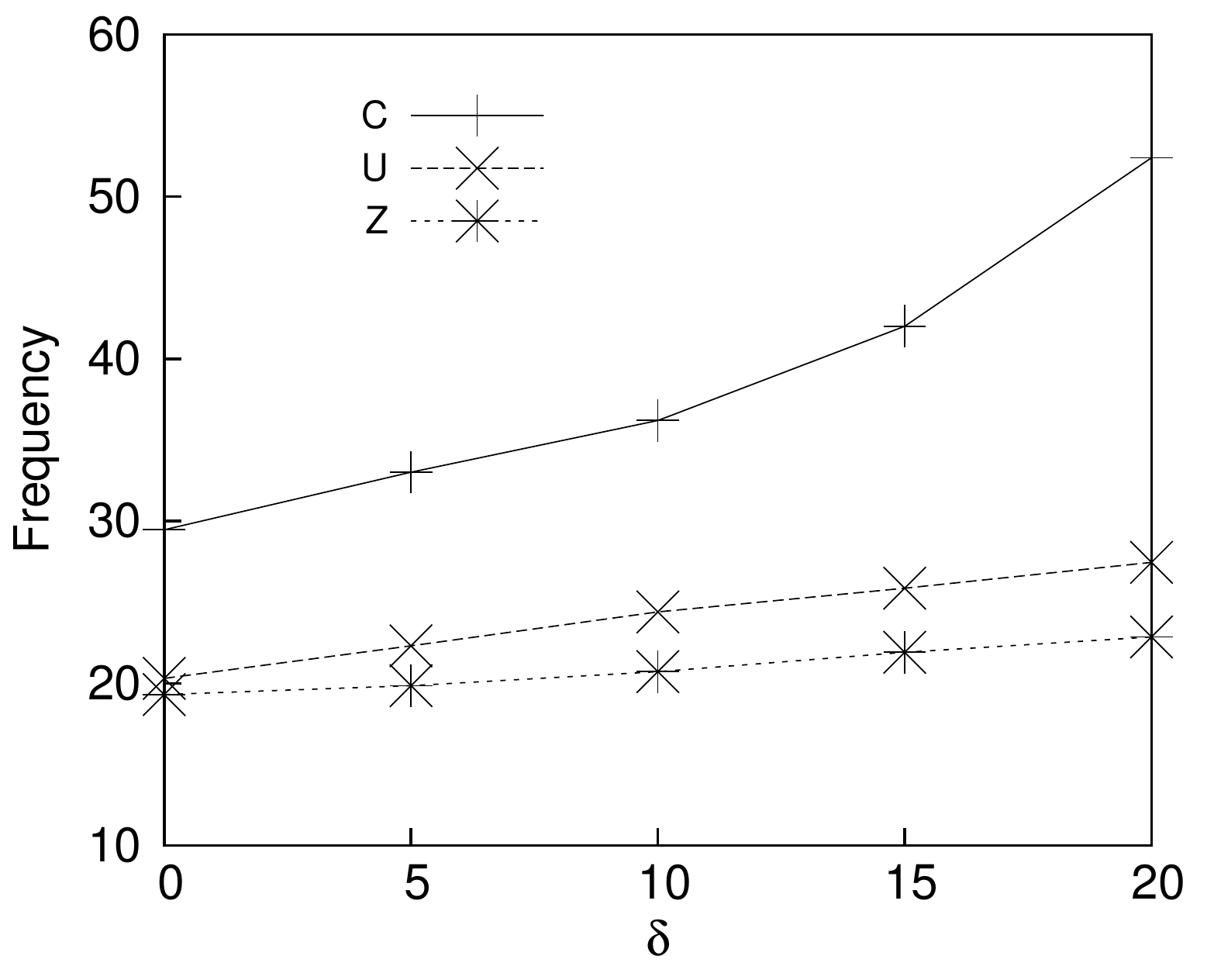}} &
      \hspace{-5mm}
      \resizebox{43mm}{!}{\includegraphics{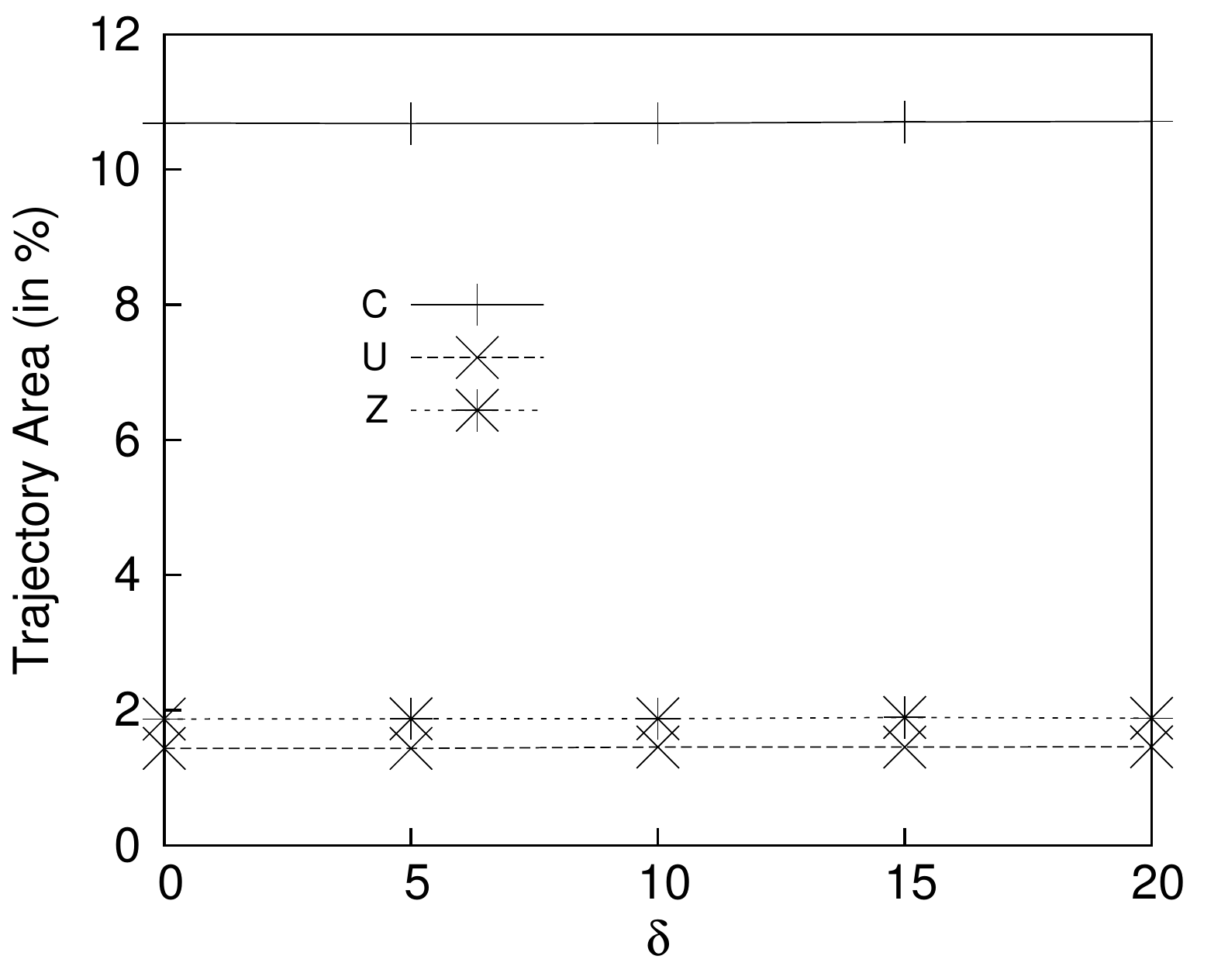}} &
        \hspace{-5mm}
      \resizebox{43mm}{!}{\includegraphics{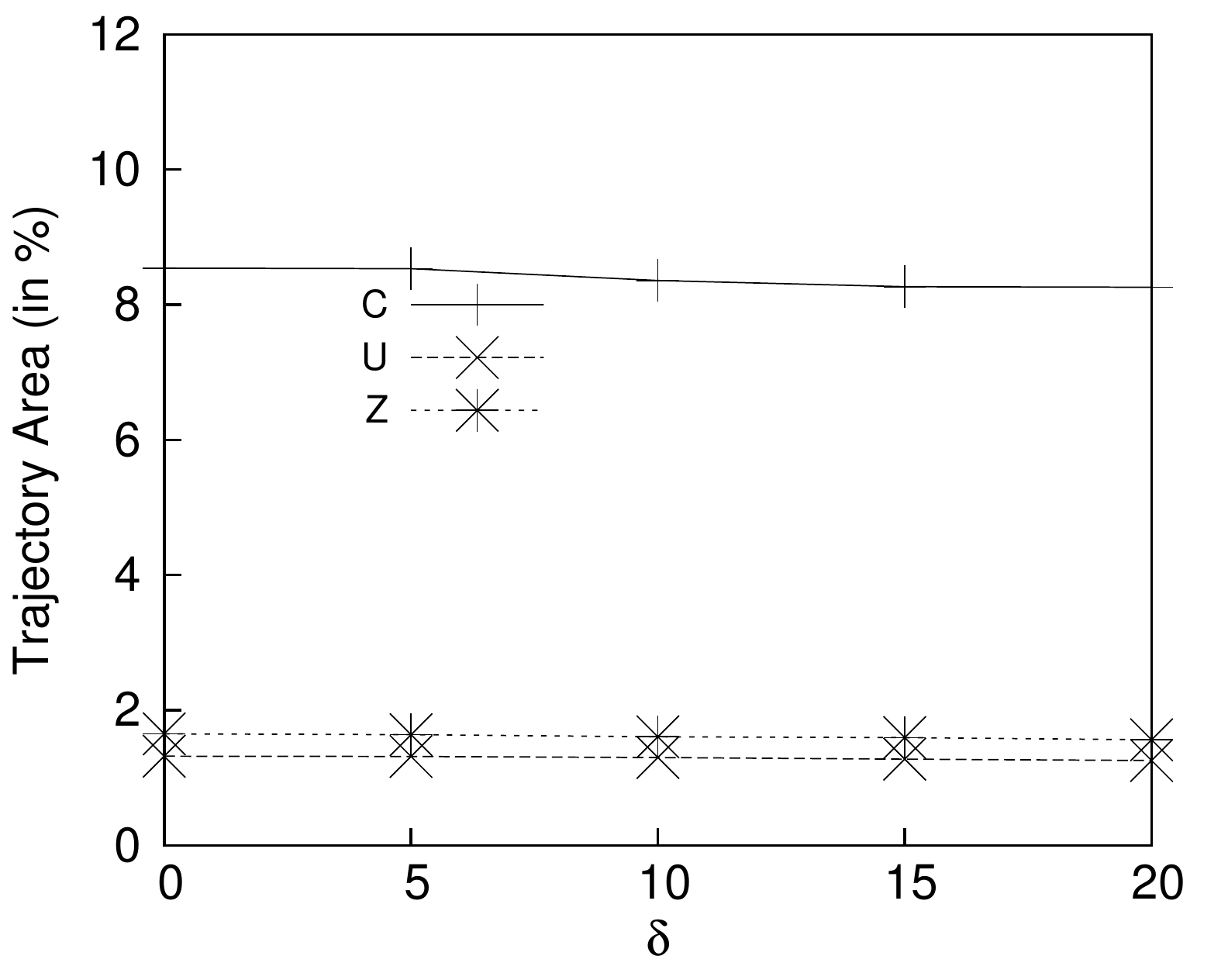}} \\
       \scriptsize{(a) Overlapping Rectangle Attack\hspace{0mm}} & \scriptsize{(b) Combined Attack} &
       \scriptsize{(c) Overlapping Rectangle Attack} & \scriptsize{(d) Combined Attack}
        \end{tabular}
    \caption{The effect of $\delta$ on the level of trajectory privacy}
    \label{fig:m_freq_d}
  \end{center}
\end{figure*}

\subsubsection{The effect of $\delta$}
\label{sec:exp_cont_delta}

We vary $\delta$ from 0 to 20 and find the effect of $\delta$ on
the level of trajectory privacy in terms of frequency and
trajectory area. Figures~\ref{fig:m_freq_d}(a)-(b) show that the
frequency increases with the increase of $\delta$ for both the
overlapping rectangle attack and the combined attack. On the other
hand, Figures~\ref{fig:m_freq_d}(c)-(d) show that the trajectory
area almost remains constant for different $\delta$.



\section{Conclusions}
\label{sec:conc}

We have developed the first approach to protect a user's trajectory privacy for M$k$NN queries. We have identified the overlapping
rectangle attack in an M$k$NN query and proposed a technique to
issue an M$k$NN query request (i.e., request $k$ NNs for
consecutive obfuscation rectangles) that overcomes this attack.
Our technique provides a user with three options: if a user does
not want to sacrifice the accuracy of answers then the user can
protect her privacy by specifying (i) a higher number of NNs than
required; otherwise, the user can specify (ii) a higher confidence
level than required or (iii) higher values for both confidence
level and the number of NNs. We have validated our trajectory
privacy protection technique with experiments and have found that
the larger the difference between the specified confidence level
(or the specified number of NNs) and the required confidence level
(or the required number of NNs), the higher the level of
trajectory privacy for M$k$NN queries. An additional advantage of
using a lower confidence level is reduced query processing cost.
We have also proposed an efficient algorithm, \textsc{Clappinq},
that evaluates the $k$ NNs for an obfuscation rectangle with a
specified confidence level, which is an essential component for
processing PM$k$NN queries. Experimental results have shown that
\textsc{Clappinq} is at least two times faster than Casper and
requires at least three times less I/Os.

In the future, we aim to extend our approach for the privacy of
data objects. For example, in a friend finder application, where
users wish to track their $k$-nearest friends continuously,
privacy is required for both the user issuing the query and the
data objects (i.e., friends). We also plan to integrate the
constraints of a road network while protecting trajectory privacy
for M$k$NN queries.

\bibliographystyle{acm}
\bibliography{vldb_tanzima}


\end{document}